\documentclass[acmsmall,screen,noacm]{acmart}

\renewcommand\footnotetextcopyrightpermission[1]{}

\AtBeginDocument{}

\usepackage{amsmath}
\usepackage{listings}
\usepackage{hyperref}
\usepackage{mathpartir}
\usepackage{tcolorbox}
\usepackage{mathtools}
\usepackage{stmaryrd}
\usepackage{xspace}
\usepackage{thm-restate}
\usepackage{adjustbox}
\usepackage{subcaption}
\usepackage{wrapfig}
\usepackage{enumitem}
\usepackage{stackengine}
\usepackage{microtype}
\usepackage{wrapfig}
\usepackage{upquote}

\definecolor{bluekeywords}{rgb}{0.13,0.13,1}
\definecolor{greencomments}{rgb}{0,0.5,0}
\definecolor{turqusnumbers}{rgb}{0.17,0.57,0.69}
\definecolor{redstrings}{rgb}{0.5,0,0}

\lstdefinelanguage{scribble}{
  morekeywords={
  	global, protocol, role, from, to, interruptible, with, do, instantiates, par, and, rec, continue, choice, at, initiates, handle, returning, call, local, or
  },
  otherkeywords={ },
  keywordstyle=\color{bluekeywords},
  sensitive=true,
basicstyle=\linespread{0.9}\ttfamily\footnotesize,
	breaklines=true,
  xleftmargin=\parindent,
  belowskip=\bigskipamount,
  aboveskip=\bigskipamount,
  tabsize=4,
  morecomment=[l][\color{greencomments}]{///},
  morecomment=[l][\color{greencomments}]{//},
  morecomment=[s][\color{greencomments}]{{(*}{*)}},
  morestring=[b]",
  showstringspaces=false,
  literate={`}{\`}1,
  frame=none,
  showlines=false,
  frame=single,
  stringstyle=\color{redstrings},
}

 \usepackage{xcolor}

\usepackage{extarrows}  \usepackage{tikz}
\usetikzlibrary{snakes,arrows,shapes,shapes.multipart,decorations.pathreplacing,calc,positioning}
\usepackage{adjustbox}

\definecolor{blueborder}{HTML}{005073}
\definecolor{bluebackground}{HTML}{107dac}

\newtcolorbox[auto counter, number within=section]{stbox}[1][]{colframe=blueborder, colback=white,
coltitle=white, colbacktitle=bluebackground, sharp corners, fonttitle=\footnotesize\sf\bfseries, #1}

\definecolor{greencomments}{rgb}{0,0.5,0}

\AtEndPreamble{\theoremstyle{acmdefinition}
    \newtheorem{remark}[theorem]{Remark}
}

\newcommand{\secref}[1]{\S\ref{#1}}

\newtcolorbox{fixmebox}{colback=black!5!white,
colframe=black!75!black,fonttitle=\small\bfseries,
size=small,sharp corners=all,title=FIXME}

\newtcolorbox{todobox}{colback=red!5!white,
colframe=red!75!black,fonttitle=\small\bfseries,
size=small,sharp corners=all,title=TODO}

\newtcolorbox{todoboxsjf}{colback=red!5!white,
colframe=red!75!black,fonttitle=\small\bfseries,
size=small,sharp corners=all,title=TODO (SJF)}

\newtcolorbox{notebox}{colback=blue!5!white,
colframe=blue!75!black,fonttitle=\small\bfseries,
size=small,sharp corners=all,title=NOTE}

\newcommand{\header}[1]{
  \begin{flushleft}
    \textbf{#1}
  \end{flushleft}
}
\newcommand{\headersig}[2]{
  \begin{flushleft}
    \textbf{#1} \hfill \framebox{#2}
  \end{flushleft}
}

\newcommand{\headertwo}[2]{
  \begin{flushleft}
    \textbf{#1} \hfill #2
  \end{flushleft}
}

\newcommand{\tybool}{\mkwd{Bool}}
\newcommand{\tyint}{\mkwd{Int}}
\newcommand{\typid}{\mkwd{Pid}}
\newcommand{\tma}{M}
\newcommand{\tmb}{N}
\newcommand{\vala}{V}
\newcommand{\valb}{W}
\newcommand{\valc}{U}
\newcommand{\sta}{S}
\newcommand{\stb}{T}
\newcommand{\calcwd}[1]{\textbf{\textsf{#1}}}
\newcommand{\mkwd}[1]{\textsf{#1}}
\newcommand{\spawn}[1]{\calcwd{spawn} \; #1 }
\newcommand{\register}[3]{\calcwd{register} \; #1 \; #2 \; #3}
\newcommand{\registertwo}[2]{\calcwd{register} \; #1 \; #2}

\newcommand{\newap}[1]{\calcwd{newAP}[{#1}]}

\newcommand{\seq}[1]{\overrightarrow{#1}}

\newcommand{\midspace}{\;\mid\;}

\newcommand{\app}{\;}
\newcommand{\typair}[2]{#1 \times #2}
\newcommand{\tylist}[1]{[#1]}

\newcommand{\one}{\mkwd{Unit}}
\newcommand{\fun}[2]{\lambda #1 .\, #2}

\newcommand{\letin}[3]{\calcwd{let} \; #1 = #2 \; \calcwd{in} \; #3}
\newcommand{\letintwo}[2]{\calcwd{let} \; #1 = #2 \; \calcwd{in}}

\newcommand{\cseq}[4][]{#2 {;} #3 \vdash_{#1} #4}

\newcommand{\aname}{a}
\newcommand{\actorsep}{,}
\newcommand{\actor}[4][\aname]{\langle #1 \actorsep #2 \actorsep #3 \actorsep #4 \rangle}
\newcommand{\config}[1]{\mathcal{#1}}
\newcommand{\teval}{\longrightarrow_{\textsf{M}}}
\newcommand{\ceval}{\longrightarrow}
\newcommand{\cevalstar}{\longrightarrow^{*}}
\newcommand{\cevalann}[1]{\xlongrightarrow{#1}}
\newcommand{\cevaltau}{\cevalann{\tau}}
\newcommand{\cevalannstar}[1]{\xlongrightarrow{#1}^{*}}

\newcommand{\threadctx}{\mathcal{M}}
\newcommand{\role}[1]{\ensuremath{{\color{purple} \mathsf{#1}}}}
\newcommand{\efflet}[3]{\calcwd{let} \; #1 = #2 \; \calcwd{in} \; #3}
\newcommand{\effreturn}[1]{\calcwd{return}\; #1}
\newcommand{\localoffer}[2]{\role{#1} \mathop{\&} \{ #2 \}}
\newcommand{\localofferone}[1]{\role{#1} \mathop{\&} \{}
\newcommand{\localoffersingle}[3]{\role{#1} \mathop{\&} \msg{#2}{#3}}
\newcommand{\localselect}[2]{\role{#1} \mathop{\oplus} \{ #2 \}}
\newcommand{\localselectone}[1]{\role{#1} \mathop{\oplus} \{}
\newcommand{\localselectsingle}[3]{\role{#1} \mathop{\oplus} \msg{#2}{#3}}
\newcommand{\localselectnone}[1]{\role{#1} \mathop{\oplus}}

\newcommand{\recty}[2]{\mu \: #1 . #2}
\newcommand{\rectyone}[1]{\mu \: #1 .}
\newcommand{\recvar}{X}
\newcommand{\localend}{\mkwd{end}}
\newcommand{\msg}[2]{#1(#2)}
\newcommand{\msgg}[2]{\mkwd{#1}(#2)}

\newcommand{\tya}{A}
\newcommand{\tyb}{B}
\newcommand{\tyc}{C}
\newcommand{\basety}{D}
\newcommand{\handlerty}[2]{\mkwd{Handler}(#1, #2)}
\newcommand{\handler}[3]{\calcwd{handler} \; \role{#1} \; #2 \; \{ #3 \}}
\newcommand{\handlerone}[2]{\calcwd{handler} \; \role{#1} \; #2 \; \{}
\newcommand{\send}[3]{\role{#1} \mathop{!} \msg{#2}{#3}}
\newcommand{\threadt}{\mathcal{T}}

\newcommand{\smallsep}{\mspace{1.5mu}}
\newcommand{\mtseq}[6][]{#2 \mid #3 \smallsep \triangleright_{#1} \smallsep #4 \mathop{:} #5 \smallsep \triangleleft \smallsep #6}

\newcommand{\mtseqst}[7][]{#2 \mid #3 \mid #4 \, \triangleright_{#1} \, #5 \mathop{:} #6 \,\triangleleft\, #7}

\newcommand{\tyenv}{\Gamma}
\newcommand{\rtenv}{\Delta}
\newcommand{\vseq}[4][]{#2 \vdash_{#1} #3 : #4}
\newcommand{\stout}[1][S]{#1^{!}}
\newcommand{\stin}[1][S]{#1^{?}}

\newcommand{\roleidx}[2]{#1[\role{#2}]}
\newcommand{\roleidxx}[2]{#1[#2]}

\newcommand{\emptyq}{\epsilon}
\newcommand{\prole}{\role{p}}
\newcommand{\qrole}{\role{q}}
\newcommand{\rrole}{\role{r}}
\newcommand{\srole}{\role{s}}

\newcommand{\blt}{\begin{array}[t]{l}}
\newcommand{\bl}{\begin{array}{l}}
\newcommand{\blz}{\begin{array}{@{}l}}
\newcommand{\el}{\end{array}}

\newcommand{\suspend}[2]{\calcwd{suspend}\; #1 \; #2}
\newcommand{\suspendexn}[3]{\calcwd{suspend}\; #1 \; #2 \; #3}
\newcommand{\suspendtimeout}[4]{\calcwd{suspend}\; #1 \; #2 \; #3 \; #4}
\newcommand{\nma}{\alpha}

\newcommand{\sessname}{s}
\newcommand{\apname}{p}

\newcommand{\ap}[2]{#1(#2)}
\newcommand{\apty}[1]{\mkwd{AP}(#1)}

\newcommand{\idle}[1]{\calcwd{idle}(#1)}
\newcommand{\standalone}[1]{#1}
\newcommand{\sessthread}[3]{(#3)^{\roleidx{#1}{#2}}}
\newcommand{\sessthreadd}[3]{(#3)^{\roleidxx{#1}{#2}}}
\newcommand{\hstate}{\sigma}
\newcommand{\storedhandler}[3]{\roleidx{#1}{#2} \mapsto #3 }

\newcommand{\istate}{\rho}
\newcommand{\apstate}{\chi}
\newcommand{\defeq}{\triangleq}
\newcommand{\threadseq}[4]{#1; #2 \mid #3 \vdash #4}
\newcommand{\hstateseq}[4]{#1; #2 \mid #3 \vdash #4}
\newcommand{\istateseq}[4]{#1; #2 \mid #3 \vdash #4}
\newcommand{\swstateseq}[2]{#1 \vdash #2}
\newcommand{\tlthreadctx}{\mathcal{Q}}

\newcommand{\ltslbl}{\gamma}
\newcommand{\synclbl}{\pi}
\newcommand{\lbleval}[1]{\xrightarrow{#1}}
\newcommand{\zaplbleval}[1]{\rsquigarrow{#1}{3}}
\newcommand{\synceval}[1]{\xRightarrow{#1}}

\newcommand{\ite}[3]{\calcwd{if} \: #1 \: \calcwd{then} \: #2 \: \calcwd{else} \: #3}

\newcommand{\set}[1]{\{ #1 \}}
\newcommand{\setlr}[1]{\left\{ #1 \right\}}

\newcommand{\rec}[3]{\calcwd{rec}\: #1(#2).#3}

\newcommand{\sttyfun}[5]{#1 \xrightarrow[#5]{#3, #4} #2}

\newcommand{\then}{\, . \,}

\newcommand{\qproc}[2]{#1 \triangleright #2}
\newcommand{\qcontents}{\delta}
\newcommand{\qentry}[4]{(\role{#1}, \role{#2}, \msg{#3}{#4})}
\newcommand{\qentryy}[4]{(#1, #2, \msg{#3}{#4})}

\newcommand{\valseq}[3]{#1 \vdash #2 : #3}

\newcommand{\fn}[1]{\mkwd{fn}(#1)}

\newcommand{\lblsyncsend}[4]{#1:\role{#2}\uparrow\role{#3}{::}#4}
\newcommand{\lblsyncrecv}[4]{#1:\role{#2}\downarrow\role{#3}{::}#4}

\newcommand{\lblzapmsg}[4]{#1:\role{#2}\lightning\role{#3}{::}#4}
\newcommand{\lblzapdequeue}[3]{#1:\role{#2}\lightning\role{#3}}
\newcommand{\lblzap}[2]{\zap{\roleidx{#1}{#2}}}

\newcommand{\safe}[1]{\mkwd{safe}(#1)}

\newcommand{\inittok}{\iota}
\newcommand{\inittokpol}{\inittok^{\pm}}
\newcommand{\inittokpos}{\inittok^{+}}
\newcommand{\inittokneg}{\inittok^{-}}
\newcommand{\inittoknegprime}{\inittok{'}^{-}}

\newcommand{\maptwo}[2]{#1 \mapsto #2}
\newcommand{\inittokseq}[3]{#1 \vdash #2 : #3}
\newcommand{\jarg}[1]{\{ #1 \} \:\:}
\newcommand{\apseq}[3]{\jarg{#1} #2 \vdash #3}
\newcommand{\qty}{Q}

\newcommand{\equivsynceval}{\Longrightarrow}
\newenvironment{proofcase}[1]{\textbf{Case} #1.\hfill\\}{}

\newcommand{\lblsyncend}[2]{\mkwd{end}(#1, #2)}

\newcommand{\deriv}[1]{\mathbf{#1}}
\newcommand{\derivd}{\deriv{D}}

\newcommand{\dom}[1]{\mkwd{dom}(#1)}

\newcommand{\var}{\mathit}
\newcommand{\langname}{\textsf{Maty}\xspace}
\newcommand{\langnamezap}{{\textsf{Maty}\ensuremath{_\lightning}}\xspace}
\newcommand{\langnamesw}{{\textsf{Maty}\ensuremath{_\rightleftarrows}}\xspace}

\newcommand{\ttrue}{\calcwd{true}}
\newcommand{\ffalse}{\calcwd{false}}
\newcommand{\dfprop}{\mkwd{df}}
\newcommand{\dfpropzap}{\mkwd{df}_{\lightning}}
\newcommand{\df}[1]{\dfprop(#1)}
\newcommand{\dfzap}[1]{\dfpropzap(#1)}

\newcommand{\ctxe}{\mathcal{E}}

\newcommand{\tyenvrt}{\Psi}
\newcommand{\ctx}[1]{\mathcal{#1}}
\newcommand{\confctx}{\ctx{G}}
\newcommand{\ctxg}{\confctx}

\newcommand{\subst}[3]{#1 \{ #2 / #3 \}}
\newcommand{\envactive}[1]{\mkwd{ongoing}(#1)}

\newcommand{\stget}{\calcwd{get}\xspace}

\newcommand{\swactor}[5][a]{\langle #1\actorsep #2\actorsep #3\actorsep #4\actorsep #5 \rangle}
\newcommand{\zapactor}[5][a]{\langle #1\actorsep #2\actorsep #3\actorsep #4\actorsep #5 \rangle}
\newcommand{\swstate}{\theta}

\newcommand{\sname}[1]{\underline{\mkwd{#1}}}
\newcommand{\siglookup}[3]{\Sigma(\sname{#1}) = (#2, #3)}

\newcommand{\sendentry}[2]{(#1, #2)}
\newcommand{\suspendrecvzero}{\calcwd{suspend}_{\mkwd{?}}\xspace}
\newcommand{\suspendrecv}[2]{\calcwd{suspend}_{\mkwd{?}}\: #1 \; #2}
\newcommand{\suspendsend}[3]{\calcwd{suspend}_{\mkwd{!}}\: \sname{#1} \; #2 \; #3}
\newcommand{\become}[2]{\calcwd{become} \: \sname{#1} \: #2}

\newcommand{\ithen}[1]{\calcwd{if} \: #1 \: \calcwd{then}}

\newcommand{\qqquad}{\qquad \quad}
\newcommand{\qqqquad}{\qquad \qquad}
\newcommand{\qqqqquad}{\qquad \qquad \quad}

\newcommand{\qqqqqqquad}{\qquad \qquad \qquad \quad}

\newcommand{\setseq}{\widetilde}

\newcommand{\metapair}[2]{(#1, #2)}
\newcommand{\suspentry}{D}
\newcommand{\monitor}[2]{\calcwd{monitor} \: #1 \: #2}
\newcommand{\monstate}{\omega}
\newcommand{\zap}[1]{\lightning{#1}}
\newcommand{\raiseexn}{\calcwd{raise}}

\newcommand{\messages}[3]{\mkwd{messages}(#1, #2, #3)}

\newcommand{\txtrole}[1]{\textbf{#1}}

\newcommand{\globalsendsingle}[4]{\role{#1} \rightarrow \role{#2} : #3(#4)}

\newcommand{\globalbranch}[2]{\role{#1} \rightarrow \role{#2} : \{}
\newcommand{\globalend}{\localend}

\newcommand{\redlbl}{l}
\newcommand{\activesessions}[1]{\mkwd{ongoingSessions}(#1)}

\newcommand{\rtenvzap}{\Phi}
\newcommand{\zapped}{\lightning\xspace}
\stackMath
\usepackage{ifthen}

\newcommand{\squig}{{\scriptstyle\sim\mkern-3.9mu}}

\newcommand{\rsquigend}{{\scriptstyle\rule{.1ex}{0ex}\rhd}}
\newcounter{sqindex}
\newcommand\squigs[1]{\setcounter{sqindex}{0}\whiledo {\value{sqindex}< #1}{\addtocounter{sqindex}{1}\squig}}
\newcommand\rsquigarrow[2]{\mathbin{\stackon[2pt]{\squigs{#2}\rsquigend}{\scriptscriptstyle\text{#1\,}}}}

\newcommand{\zapenvred}{\Rrightarrow}

\newcommand{\zapenv}[1]{\zap{#1}}

\newcommand{\RED}[1]{\textcolor{red}{#1}}

\newcommand{\BLUE}[1]{\textcolor{blue}{#1}}

\newcommand{\CODE}[1]{{\small\texttt{#1}}}

\newcommand{\CODEss}[1]{{\scriptsize\texttt{#1}}}

\definecolor{shade}{RGB}{215,215,215}
\newcommand{\shade}[1]{\setlength{\fboxsep}{0pt}\colorbox{shade}{\ensuremath{#1}}}

\newcommand{\zaproles}[2]{\mkwd{zap}(#1, #2)}
\newcommand{\roles}[1]{\mkwd{roles}(#1)}

\newcommand{\kp}[1]{\textbf{KP#1}}

\newcommand{\compliant}[1]{\mkwd{comp}(#1)}
\newcommand{\comp}{\compliant}
\newcommand{\compzap}[1]{\mkwd{comp}_{\lightning}(#1)}

\newcommand{\proto}{P}

\newcommand{\snames}[1]{\mkwd{snames}(#1)}

\newcommand{\mterm}[1]{\mkwd{term}_{\mathsf{M}}}
\newcommand{\tterm}[1]{\mkwd{term}_{\mathsf{T}}}
\newcommand{\cterm}[1]{\mkwd{term}_{\mathsf{C}}}

\newcommand{\leave}[1]{\calcwd{leave} \: #1}

\newcommand{\smallmath}[1]{\parbox{\linewidth}{\footnotesize
    \[
      #1
    \]
  }}

\newcommand{\smallermath}[1]{\parbox{\linewidth}{\scriptsize
    \[
      #1
    \]
  }}

\usepackage{etoc}

\begin{document}
\etocdepthtag.toc{mtchapter}
\etocsettagdepth{mtchapter}{subsection}
\etocsettagdepth{mtappendix}{none}

\title{Speak Now}
\subtitle{Safe Actor Programming with Multiparty Session Types (Extended Version)}

\author{Simon Fowler}
\orcid{0000-0001-5143-5475}
\affiliation{\institution{University of Glasgow}
  \city{Glasgow}
  \country{United Kingdom}
}
\email{simon.fowler@glasgow.ac.uk}

\author{Raymond Hu}
\orcid{0000-0003-4361-6772}
\affiliation{\institution{Queen Mary University of London}
  \city{London}
  \country{United Kingdom}
}
\email{r.hu@qmul.ac.uk}

\ccsdesc[500]{Software and its engineering~Concurrent programming languages}

\keywords{multiparty session types, actor languages, deadlock-freedom}

\begin{abstract}
Actor languages such as Erlang and Elixir are widely used for implementing
    scalable and reliable distributed applications, but the informally-specified
    nature of actor communication patterns leaves systems vulnerable to costly
    errors such as communication mismatches and deadlocks.
\emph{Multiparty session types} (MPSTs) rule out communication errors early
    in the development process, but until now, the many-sender, single-receiver
    nature of actor communication has made it difficult for actor languages to
    benefit from session types.

    This paper introduces \langname, the first actor language design
    supporting both \emph{static} multiparty session typing and the full power
    of actors taking part in \emph{multiple sessions}. \langname therefore
    combines the error prevention mechanism of session types with the
    scalability and fault tolerance of actor languages.
Our main insight is to
    enforce session typing through a flow-sensitive effect system,
    combined with an event-driven programming style and first-class message
    handlers. 
Using MPSTs allows us to guarantee communication safety: a process will
    never send or receive an unexpected message, nor will a session get
    stuck because an actor is waiting for a message that will never be sent.
    We extend \langname to support Erlang-style supervision and cascading
    failure, and show that this preserves \langname's strong metatheory.
    We implement \langname in Scala using an API generation approach, and
    demonstrate the expressiveness of our model by implementing
    a representative sample of the widely-used Savina actor benchmark suite; an
    industry-supplied factory scenario; and a chat server.
\end{abstract}

\maketitle

\section{Introduction}\label{sec:intro}
Modern digital infrastructure depends on distributed software. Unfortunately, writing
distributed software is difficult: developers must reason about a host of issues
such as deadlocks, failures, and adherence to complex communication protocols.
Actor languages such as Erlang and Elixir, and frameworks like Akka,
are popular for writing scalable, resilient systems; 
Erlang in particular powers the servers of WhatsApp, which has billions of users
worldwide.
Actor languages support lightweight processes that communicate through
asynchronous explicit message passing rather than shared memory, and support
robust failure recovery strategies like \emph{supervision hierarchies}.

Nevertheless, actor languages are not a silver bullet: it is still
possible---\emph{easy}, even---to introduce subtle communication errors that are
difficult to detect, debug, and fix. Examples include waiting for a message that
will never arrive, sending a message that cannot be handled, or sending an
incorrect payload.
\emph{Multiparty session types}
(MPSTs)~\cite{HondaYC08,DBLP:conf/birthday/BettiniCDGV08} are \emph{types for
protocols} and allow us to reason about structured interactions between
communicating participants.  If each participant typechecks against its
session type, then the system is statically guaranteed to correctly implement
the associated protocol, in turn catching communication errors before a
program is run.

MPSTs therefore offer a tantalising promise for actor languages:
by combining the fault-tolerance and ease-of-distribution of actor languages
with the correctness guarantees given by MPSTs, users can fearlessly write
robust and scalable distributed code, confident in the absence of communication
mismatches and deadlocks.
Unfortunately, there is a spanner in the works: MPSTs have been primarily
studied for \emph{channel-based} languages, which have a significantly different
communication model, and current session-typing approaches for
actor languages are severely limited in expressiveness: they either overlay a
session-based communication discipline, only describe interactions between two
participants, or only permit actors to be involved in a single session at a time. 
Other behavioural type systems for actors struggle to capture \emph{structured}
interactions and \emph{handle failure} effectively.

In this paper we present \langname, the first actor language that supports
statically-checked multiparty session types and failure handling, combining the
error prevention mechanism of session types and the scalability and fault
tolerance of actor languages. Our key insight is to adopt an \emph{event-driven
programming style} and enforce session typing through a \emph{flow-sensitive
effect system}.

\subsection{Actor Languages}\label{sec:intro:actor-languages}

Actor languages and frameworks are inspired by the \emph{actor
model}~\cite{HewittBS73,Agha90}, where an actor reacts to incoming messages by
spawning new actors, sending a finite number of messages to other actors, and
changing the way it reacts to future messages. Figure~\ref{fig:intro:id-server}
shows an Akka implementation of an ID server that generates a fresh ID for
every client request.

\begin{wrapfigure}[6]{l}{0.45\textwidth}
    \vspace{-1em}
\begin{lstlisting}[language=scala]
def idServer(count: Int): 
    Behavior[IDRequest] = {
  Behaviors.receive { (context, message) =>
    message.replyTo ! IDResponse(count)
    idServer(count + 1)
  }
}
\end{lstlisting}
\vspace{-2em}
\caption{ID Server}
\label{fig:intro:id-server}
\end{wrapfigure}

The \verb+idServer+ function records the current request count as its
state, and responds to an incoming \verb+IDRequest+ by sending
the current \verb+count+ before recursing with an incremented request
counter.

It is straightforward to specify the client-server \emph{protocol} for this
example as a session \emph{type} between these two roles, but there are key
problems implementing and verifying even this simple example in standard MPST
frameworks.
First, actor programming is inherently \emph{reactive}: computation is driven by
the reception of a new message, and actors must be able to respond to requests
from a \emph{statically-unknown} number of clients.
Second, each response depends on some common \emph{state}.

\begin{wrapfigure}[16]{l}{0.45\textwidth}
\vspace{-1.5em}
\begin{lstlisting}
def idServer(count: Int, locked: Boolean):
    Behavior[IDServerRequest] = {
  Behaviors.receive { (context, message) =>
    message match {
      case IDRequest(replyTo) =>
        if (locked) {
            replyTo ! Unavailable()
            idServer(count, locked)
        } else {
            replyTo ! IDResponse(count)
            idServer(count + 1, locked)
        }
      case LockRequest(replyTo) =>
        if (locked) {
            replyTo ! Unavailable()
            idServer(count, locked)
        } else {
            replyTo ! Locked(context.self)
            idServer(count, true)
        }
      case Unlock() =>
        idServer(count, false)
    }
  }
}
\end{lstlisting}
\vspace{-2em}
\caption{ID Server extended with locking}
\label{fig:intro:locking-id-server}
\end{wrapfigure}
 
Classical MPSTs are instead based on session $\pi$-calculus, which is
effectively a model of \emph{proactive multithreading} as opposed to reactive
event handling.
A standard MPST server process relies on \emph{replication} to spawn a
separate ($\pi$-calculus) process to handle each client session concurrently.
For reference, common notations/patterns include
$\mathtt{Server} \mathbin{=} a(x) . (P_\mathtt{thread}(x) \;|\; \mathtt{Server})$
or
$\mathtt{Server} \mathbin{=}{} !\,a(x) . P_\mathtt{thread}(x)$.
This model has no direct support for coordinating a \emph{dynamically varying}
number of separate client-handler processes, and---crucially---key
safety properties of standard MPSTs such as deadlock-freedom \textbf{only hold
when each process engages in a single session and each session can be
conducted fully independently from the others} (i.e., an embarrassingly parallel
situation).
Introducing any method to synchronise shared state between these processes, be
it through an intricate web of additional internal sessions
or some out-of-band (i.e., non-session-typed) method, means deadlock-freedom is
no longer guaranteed.
Besides safety concerns, the $\pi$-calculus based model makes it
difficult to express important patterns such as a \emph{single} process waiting
to reactively receive from senders across multiple sessions, since inputs are
normally blocking operations.

\paragraph*{A locking ID server.}
Figure~\ref{fig:intro:locking-id-server} shows a simple extension of our ID
server, where a participant can choose to \emph{lock} the server to
prevent it from generating fresh IDs until the lock is released.

In this example, replies depend on whether the ID server is locked.
Upon receiving an \verb+IDRequest+ message, if the server is locked, then
it will respond with \verb+Unavailable+; otherwise, it will reply as
before.
If an unlocked server receives a
\verb+LockRequest+ message, it responds with \verb+Locked+ and
sets the \verb+locked+ flag. A subsequent \verb+Unlock+ message resets the
\verb+locked+ flag.

This small extension to the example reveals some intricacies: once a client has
received a lock, it is in a \emph{different state of the protocol} to the
remaining clients. First, there is no straightforward way of guaranteeing that
the client ever sends an \verb+Unlock+ message, nor that the
\verb+Unlock+ message was sent by the same actor that acquired the lock.
Second, the server must \emph{always} be able to handle an \verb+Unlock+
message, \emph{even when it is already unlocked}---permitting an invalid state.
Both of these issues can be straightforwardly solved using session types in
\langname.

\begin{figure}[t]
    \centering
    \begin{subfigure}{0.25\textwidth}
        \centering
        \includegraphics[width=0.74\textwidth]{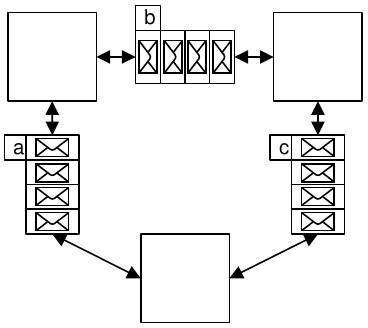}
        \caption{Channels}
        \label{fig:chans-actors:channels}
    \end{subfigure}
    \qquad\quad
    \begin{subfigure}{0.25\textwidth}
        \centering
        \includegraphics[width=0.74\textwidth]{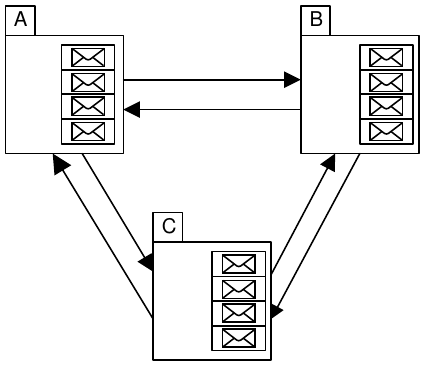}
        \caption{Actors}
        \label{fig:chans-actors:actors}
    \end{subfigure}
\caption{Channel- and actor-based languages~\cite{FowlerLW17}}
    \label{fig:chans-actors}
    \vspace{-1em}
\end{figure}

\subsection{Channels vs.\ Actors}
Session types were originally developed for channel-based languages like Go and
Concurrent ML (Figure~\ref{fig:chans-actors:channels}) and support anonymous
processes that communicate over \emph{channel endpoints} using either
\emph{synchronous} or \emph{asynchronous} communication.  In actor languages
(Figure~\ref{fig:chans-actors:actors}) such as Erlang or Elixir, \emph{named}
processes send messages directly to each others' mailboxes.
The difference in communication models has significant consequences for
distribution and typing.  We can easily give a channel endpoint a precise type,
e.g., $\mkwd{Chan}(\tyint)$ or a \emph{session type} such as
${!}\tyint.{!}\tyint.{?}\tybool.\localend$ to state that the channel should be
used to send two integers and receive a Boolean.
Efficiently implementing channels in a distributed setting requires us to store
buffered data at the same location that it is processed, but difficulties arise
when sending channel endpoints as part of a message (known as \emph{distributed
delegation}).  Furthermore, implementing even basic channel idioms such as
choosing between multiple channels requires complex distributed
algorithms~\cite{Chaudhuri09}. In short, channel-based languages are \emph{easy
to type} but \emph{difficult to distribute}.

In contrast, actor languages are much easier to distribute, since every message
will always be stored at the process that will handle it. But typing an actor is
harder, requiring large variant types, and behavioural typing is difficult since
we can only \emph{send} to process IDs and \emph{receive} from mailboxes.  Thus,
actors are \emph{easy to distribute} but \emph{difficult to type}.

\subsection{Key Principles}
\label{sec:intro:key-principles}

For session types to be useful for real-world actor programs, we argue that a
programming model and session type discipline must satisfy the following
\emph{key principles}:

\begin{description}
    \item[(KP1) Reactivity]
        Following the actor model, frameworks like Akka, and Erlang behaviours
        like \verb+gen_server+, computation should be triggered by incoming
        messages.
    \item[(KP2) No Explicit Channels]
        Channel-based languages impose a significantly different programming
        style, so the programming model should \emph{not} expose explicit
        channels to a developer.
    \item[(KP3) Multiple Sessions]
        Actors must be able to simultaneously take part in an unbounded and
        statically-unknown number of sessions, in order to support server
        applications. It must be possible for different participants to be at
        different states of a protocol.
    \item[(KP4) Interaction Between Sessions]
        Much like our ID server example, interactions in one session should be
        able to affect the behaviour of an actor in other sessions.
    \item[(KP5) Failure Handling and Recovery]
        The programming model and type discipline should support failure
        recovery via supervision hierarchies.
\end{description}

\emph{No previous work that applies session types to actor languages satisfies
all of the key principles above.}
\citet{MostrousV11} investigated session typing for Core Erlang: their
approach emulated session-typed channels by tagging messages with unique
references.  Their approach was unimplemented, not reactive, exposed a
channel-based discipline, and does not support failure, violating \kp{1, 2, 5}.
\citet{FrancalanzaT23} implemented a binary session typing system for Elixir,
but their approach is limited to typing interactions between isolated pairs of
processes and is therefore severely limited in expressiveness, violating
\kp{1, 3, 4, 5}.
\citet{HarveyFDG21} used multiparty session types in an actor language to
support safe runtime adaptation, but each actor can only take part in a \emph{single
session} at a time. It is therefore difficult to write server applications and
so the language \emph{does not support general-purpose actor programming},
violating \kp{1, 3, 4}.

\citet{NeykovaY16} and~\citet{Fowler16} implement programming frameworks closer
to following our key principles: each actor is programmed in a reactive style
and can be involved in multiple sessions, but both works use
\emph{dynamic verification of actors using session types as a notation for
generating runtime monitors}. They do not consider any formalism, session type
system, nor metatheoretical guarantees, and so there is a significant gap
between their conceptual framework and a concrete static programming language
design.

In contrast, \langname supports all of our key principles by reacting to incoming
messages rather than having an explicit receive operation~(\kp{1}); enforcing
session typing through a flow-sensitive effect system rather than explicit
channel handles~(\kp{2}); using the reactive design to support interleaved
handling of messages from different sessions~(\kp{3}); supporting interaction
between sessions using state and self-messages, with our implementation
also supporting an explicit session switching construct~(\kp{4}); and
supporting failures and recovery via supervision hierarchies~(\kp{5}).

\subsection{Contributions}
Concretely, we make four specific contributions:

\begin{enumerate}
    \item We introduce \langname, the first actor language fully supporting 
        multiparty session types (\S\ref{sec:formalism}).
    \item
        We show that \langname enjoys a
        strong metatheory~(\S\ref{sec:metatheory}):
        type preservation (Theorem~\ref{thm:preservation}) guarantees that an actor will
        never send or receive an unexpected message;
        progress (Theorem~\ref{thm:progress}) guarantees that processes never get stuck
        due to waiting for a message that is never sent because of an unsafe
        protocol design or faulty implementation; and global
        progress (Corollary~\ref{cor:global-progress}) guarantees that, in a
        system where all event handlers terminate, communication is always
        eventually possible in any ongoing session.
    \item We show how to extend \langname with support for Erlang-style failure
        handling and process supervision~(\S\ref{sec:failure}), and prove
        that this maintains \langname's strong metatheory.
    \item We detail our implementation of \langname using an API generation
        approach in Scala (\S\ref{sec:implementation}), and demonstrate our
        implementation on a representative selection of benchmarks from the
        Savina actor benchmark suite; a real-world case study from the
        factory domain; and a chat server.
\end{enumerate}

Section~\ref{sec:related} discusses related work, and
Section~\ref{sec:conclusion} concludes.
 \begin{figure}[t]
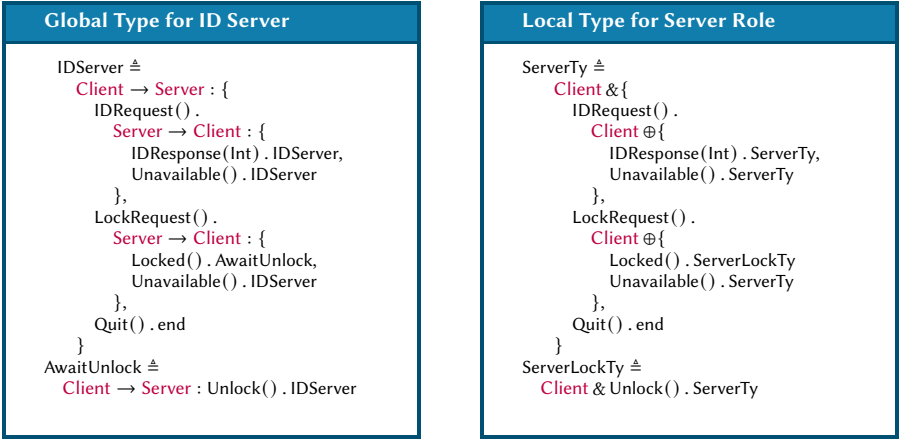

\begin{minipage}{0.4\textwidth}
    {
    \begin{stbox}[title=Global Type for ID Server,height=5.75cm]
        \scriptsize
        \vspace{-1.2em}
        \[
        \hspace{-1.25em}
        \bl
        {
            \bl
            \mkwd{IDServer} \defeq \\
            \quad \globalbranch{Client}{Server} \\
            \qquad \msgg{\mkwd{IDRequest}}{} \then \\
            \qqquad \globalbranch{Server}{Client} \\
            \qqqquad \msgg{IDResponse}{\tyint} \then \mkwd{IDServer}, \\
            \qqqquad \msgg{Unavailable}{} \then \mkwd{IDServer} \\
            \qqquad \}, \\
\qquad \msgg{\mkwd{LockRequest}}{} \then \\
            \qqquad \globalbranch{Server}{Client} \\
            \qqqquad \msgg{Locked}{} \then  \mkwd{AwaitUnlock}, \\
\qqqquad \msgg{Unavailable}{} \then \mkwd{IDServer} \\
            \qqquad \}, \\
            \qquad \msgg{\mkwd{Quit}}{} \then \globalend \\
            \quad \}
            \el
        }\\
            \mkwd{AwaitUnlock} \defeq \\
            \quad \globalsendsingle{Client}{Server}{\mkwd{Unlock}}{} \then
            \mkwd{IDServer}
        \el
        \]
\vspace{-0.5em}
    \end{stbox}
}
\end{minipage}
\qquad
\begin{minipage}{0.4\textwidth}
    {
    \begin{stbox}[title=Local Type for Server Role,height=5.75cm]
        \scriptsize
        \vspace{-1.2em}
        \[
        \hspace{-1.25em}
            \bl
            \mkwd{ServerTy} \defeq \\
            \quad
            {\bl
            \localofferone{Client} \\
            \quad \msgg{\mkwd{IDRequest}}{} \then \\
            \qquad \localselectone{Client} \\
            \qqquad \msgg{IDResponse}{\tyint} \then \mkwd{ServerTy}, \\
            \qqquad \msgg{Unavailable}{} \then \mkwd{ServerTy} \\
            \qquad \}, \\
\quad \msgg{\mkwd{LockRequest}}{} \then \\
            \qquad \localselectone{Client} \\
            \qqquad \msgg{Locked}{} \then \mkwd{ServerLockTy} \\
            \qqquad \msgg{Unavailable}{} \then \mkwd{ServerTy} \\
            \qquad \}, \\
            \quad \msgg{\mkwd{Quit}}{} \then \globalend \\
            \}
            \el
            } \\
            \mkwd{ServerLockTy} \defeq \\
            \quad \localoffersingle{Client}{\mkwd{Unlock}}{} \then \mkwd{ServerTy}
            \el
        \]
\vspace{-0.5em}
    \end{stbox}
}
\end{minipage}
\vspace{-1em}
\caption{Session types for the ID server example.}
\label{fig:tour:id-server-sts}
\end{figure}

 \section{A Tour of \langname}\label{sec:tour}

In this section we introduce \langname by example, first by considering how to
write our ID server, and then by considering a larger online shop example.

\subsection{The Basics: ID Server}

\paragraph*{Session types.}
Figure~\ref{fig:tour:id-server-sts} shows the session types for the ID server
example.
The \emph{global type} describes the interactions between the ID server and a
client.  For simplicity, we assume a standard encoding of mutually recursive
types and use mutually recursive definitions in our examples. The client starts
by sending one of \mkwd{IDRequest}, \mkwd{LockRequest}, or \mkwd{Quit} to the
server. On receiving \mkwd{IDRequest}, the server replies with \mkwd{IDResponse}
if it is unlocked, or \mkwd{Unavailable} if it is locked; in both cases, the
protocol then repeats.  On receiving \mkwd{LockRequest}, the server replies with
\mkwd{Locked} (if it locks successfully), and the client must then send
\mkwd{Unlock} before repeating. If already locked, the server responds with
\mkwd{Unavailable}.  The protocol ends when the server receives a \mkwd{Quit}
message.

A global type can be \emph{projected} to \emph{local types} that describe the
protocol from the perspective of each participant.  The local type on the right
details the protocol from the server's viewpoint: the $\&$ operator denotes
offering a choice, and the $\oplus$ operator denotes making a selection.  The
(omitted) \mkwd{ClientTy} type is similar, but implements the \emph{dual}
actions: where the server offers a choice, the client makes a selection, and
vice-versa.
We define a \emph{protocol} $P$ as a mapping from role names to local session
types. In our example we define 
$\mkwd{IDServerProtocol} \defeq \set{\role{Client}: \mkwd{ClientTy},
\role{Server} : \mkwd{ServerTy}}$.

\paragraph{Programming model.}
The \langname programming model is as follows:
\begin{itemize}
    \item \langname is faithful to the actor model, which has a single thread of
        execution per actor. This allows access to shared state \emph{without}
        needing concurrency control mechanisms like mutexes.
    \item An actor registers with an \emph{access point} to register to take part
    in a session.
    \item Once a session is established, the actor can send messages according
        to its session type. The actor maintains some \emph{actor-level state}
        and its thread must either return an updated state (if it has completed
        its part in the protocol), or suspend by installing a \emph{message
        handler} (if it is ready to receive a message).
        Suspension acts as a \emph{yield point} to the event loop, and occurs
        at \textbf{precisely the same point as in mainstream actor languages.}
    \item The event loop can then invoke other installed handlers for any
        messages in its mailbox---\textbf{this is the key mechanism that
        allows \langname to support multiple sessions}.
\end{itemize}

\begin{figure}
        \footnotesize
    \begin{minipage}{0.475\textwidth}
\[
\begin{array}{l}
\mkwd{idServer} : (\apty{\mkwd{IDServerProtocol}} \times (\tyint \times
\tybool)) \\
    \quad \rightarrow (\tyint \times \tybool) \\
    \mkwd{idServer} = \lambda (\var{ap}, \var{state}) .\; \mkwd{registerAgain}
    \app \var{ap}; \; \var{state}
\vspace{0.5em}
\\
\mkwd{registerAgain} : (\apty{\mkwd{IDServerProtocol}}) \rightarrow \one \\
\mkwd{registerAgain} = \lambda \var{ap} . \\
\quad \calcwd{register}~\var{ap}~\role{Server} \\
\qquad (\lambda \var{st} .\: \mkwd{registerAgain} \app
    \var{ap};~\suspend{\mkwd{requestHandler}}{\var{st}})
\vspace{0.5em}
\\
\mkwd{unlockHandler} : \handlerty{\mkwd{ServerLockTy}}{(\typair{\tyint}{\tybool})} \\
\mkwd{unlockHandler} = \\
\quad \handlerone{Client}{\var{st}} \: \\
\qquad \mkwd{Unlock()} \mapsto \\
\qqquad \calcwd{let}~(\var{curID}, \var{locked}) = \var{st}~\calcwd{in} \\
\qqquad \suspend{\mkwd{requestHandler}}{(\var{curID}, \ffalse)} \\
\quad \}
\vspace{0.5em}
\\
\mkwd{main} : \one \\
\mkwd{main} = \\
\quad
\letintwo{\var{idServerAP}}{\newap{\mkwd{IDServerProtocol}}} \\
\quad \calcwd{spawn}~(\mkwd{idServer} \app (\mkwd{idServerAP}, (0, \ffalse)); \\
\quad \calcwd{spawn}~(\mkwd{client} \app \mkwd{idServerAP})
\end{array}
\]
\end{minipage}
\hfill
\begin{minipage}{0.475\textwidth}
\[
\begin{array}{l}
\mkwd{requestHandler} : \handlerty{\mkwd{ServerTy}}{(\tyint \times \tybool)} \\
\mkwd{requestHandler} = \\
\quad \handlerone{Client}{\var{st}} \\
\qquad \mkwd{IDRequest}() \mapsto \\
\qqquad \calcwd{let}~(\var{curID}, \var{locked}) = \var{st}~\calcwd{in} \\
\qqquad \calcwd{if}~\mkwd{locked}~\calcwd{then} \\
\qqqquad \role{Client}~!~\mkwd{Unavailable}(); \\
\qqqquad \suspend{\mkwd{requestHandler}}{\var{st}} \\
\qqquad \calcwd{else} \\
\qqqquad \role{Client}~!~\mkwd{IDResponse}~(\var{curID}); \\
\qqqquad \suspend{\mkwd{requestHandler}}{(\var{curID} + 1, \var{locked})}, \\
\qquad \mkwd{LockRequest}() \mapsto \\
\qqquad \calcwd{let}~(\var{curID}, \var{locked}) = \var{st}~\calcwd{in} \\
\qqquad \calcwd{if}~\var{locked}~\calcwd{then} \\
\qqqquad \role{Client}~!~\mkwd{Unavailable}(); \\
\qqqquad \suspend{\mkwd{requestHandler}}{\var{st}} \\
\qqquad \calcwd{else} \\
\qqqquad \role{Client}~!~\mkwd{Locked}(); \\
\qqqquad \suspend{\mkwd{unlockHandler}}{(\var{curID}, \ttrue)}, \\
\qquad \mkwd{Quit}() \mapsto \var{st} \\
\quad \}
\end{array}
\]
\end{minipage}
\vspace{-1em}
    \caption{\langname implementation of ID Server}
    \label{fig:intro:maty-idserver}
\end{figure}
 \paragraph{Implementing the ID server.}
Figure~\ref{fig:intro:maty-idserver} shows an implementation of the ID server in
\langname; we allow ourselves the use of mutually-recursive definitions, taking
advantage of the usual encoding into anonymous recursive functions. Although we
use an effect system that annotates function arrows, we omit effect annotations
where they are not necessary.

The server maintains actor-level state of type $(\tyint \times \tybool)$,
containing the current ID and a flag recording whether the server is locked. The
\mkwd{idServer} function takes an \emph{access point}~\cite{GayV10} and initial
state as an argument, and registers for the \role{Server} role.
An access point can be thought of as a ``matchmaking service'': actors
\emph{register} to play a role in a session, and the access point establishes a
session once actors have registered for each role.
The $\calcwd{register}$ construct takes three arguments: an access point, a role,
and a callback to be invoked when the session is established. 
\emph{Once the callback is invoked, the actor can perform session communication
actions for the given role}: in this case, the actor can communicate according
to the \mkwd{ServerTy} type, namely receiving the initial item request from a
client.
The callback first recursively registers to be involved in future sessions, and
then \emph{suspends} awaiting a message from a client, by installing
$\mkwd{requestHandler}$. 

A \emph{message handler} (or simply \emph{handler}) is a first-class construct
that describes how an actor handles an incoming message: each handler takes the
role to receive from; a variable to which to bind the current actor state; and a
series of branches that detail how each message should be handled.
An actor \emph{installs} a message
handler for the current session by invoking the $\calcwd{suspend}$ construct,
which reverts the actor back to being idle with an updated state, and indicates
that the given handler should be invoked when a message is received from the
\role{Client}.

The \mkwd{requestHandler} has type
$\handlerty{\mkwd{ServerTy}}{(\typair{\tyint}{\tybool})}$: handlers are
parameterised by an input session type and the type of the actor's state.
The handler has three branches, one for each possible incoming message.
\langname uses a flow-sensitive effect system~\cite{FosterTA02,
Atkey09, Gordon17}
to enforce session typing using pre- and post-conditions on expressions.
In the \mkwd{IDRequest} branch, the pre-condition
$\localselect{Client}{ \msgg{IDResponse}{\tyint}\then\mkwd{ServerTy},
\msgg{Unavailable}{}\then\mkwd{ServerTy}}$ means that the actor can \emph{only
send \mkwd{IDResponse} or \mkwd{Unavailable} messages}; all other communication
actions are rejected statically. After either message is sent, the
session type advances to \mkwd{ServerTy}, allowing the handler to suspend
recursively. The \mkwd{LockRequest} branch works similarly;
since \calcwd{suspend} aborts the current evaluation context, both branches can
be given type $(\typair{\tyint}{\tybool})$ with post-condition $\localend$ to
match the \mkwd{Quit} branch.
The \mkwd{unlockHandler} handles \mkwd{Unlock} by updating the state,
reinstalling \mkwd{requestHandler}, and suspending.
Implementing a client is similar (details omitted).
Finally, \lstinline+main+ sets up the access point (associating \role{Client}
with \mkwd{ClientTy} and \role{Server} with \mkwd{ServerTy}), then spawns
$\mkwd{idServer}$ with a pair of arguments $\var{idServerAP}$ (to allow the server to
register for sessions) and $(0, \ffalse)$, its initial state.
The \lstinline+main+ function also spawns the client, passing the same access
point.

\begin{figure}
    {\small \header{Scenario Description}}
    \begin{minipage}{0.425\textwidth}
    \includegraphics[width=\linewidth]{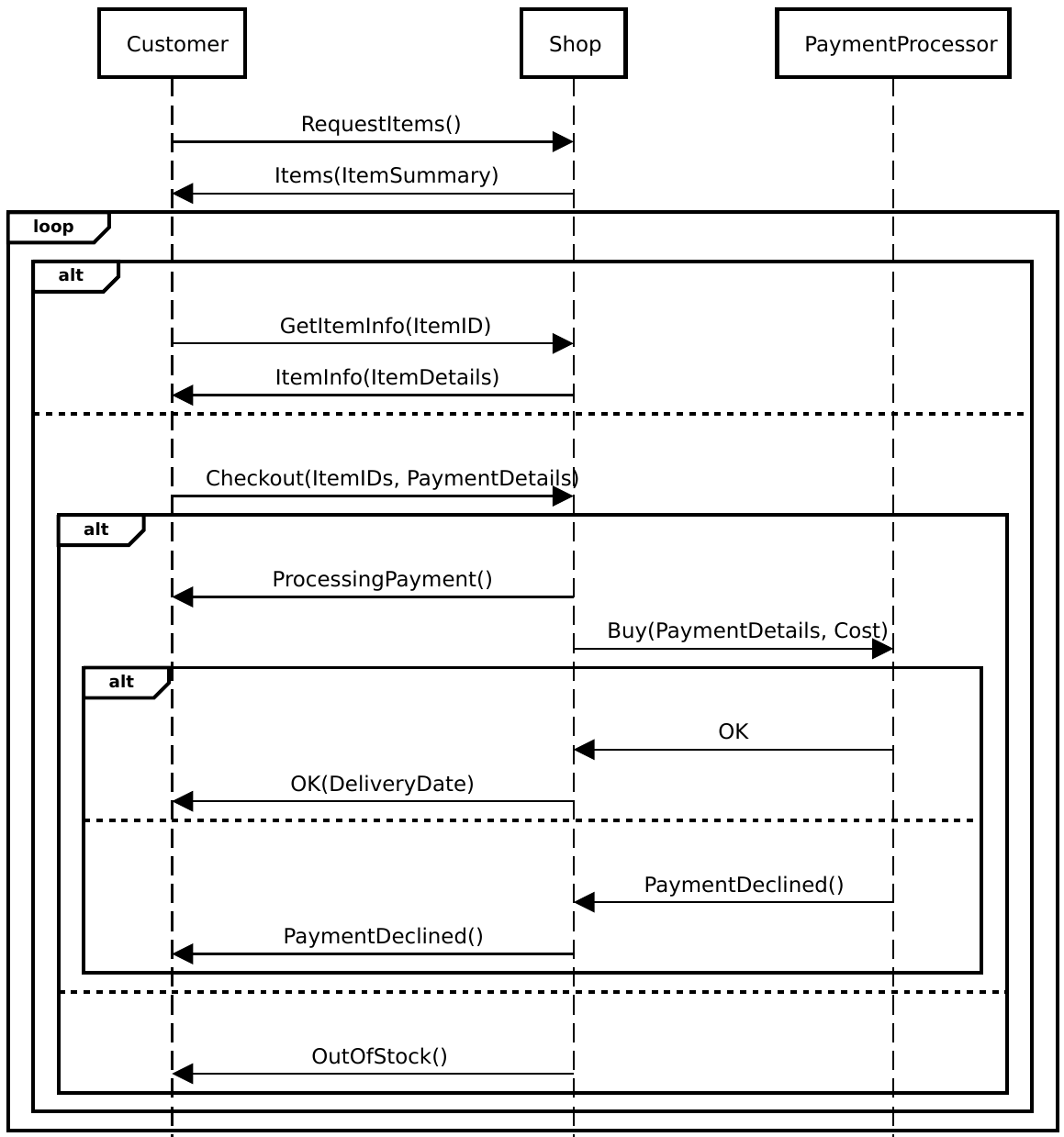}
\end{minipage}
\hfill
\small
\begin{minipage}{0.55\textwidth}
\begin{itemize}
  \setlength{\itemindent}{-1.75em}
    \item A \txtrole{Shop} can serve many \txtrole{Customer}s at once.
    \item The \txtrole{Customer} begins by requesting a list of items from the \txtrole{Shop}, which sends back a list of pairs of an item's identifier and name.
    \item The \txtrole{Customer} can then repeatedly either request full details (including description and cost) of an item, or proceed to checkout.
    \item To check out, the \txtrole{Customer} sends their payment details and a list of item IDs to the \txtrole{Shop}.
    \item If any items are out of stock, then the \txtrole{Shop} notifies the
        customer who can then try again. Otherwise, the \txtrole{Shop} notifies
        the \txtrole{Customer} that it is processing the payment, and forwards
        the payment details and total cost to the \txtrole{Payment Processor}.
    \item The \txtrole{Payment Processor} responds to the \txtrole{Shop}
        with whether the payment was successful.
    \item The \txtrole{Shop} relays the result to the \txtrole{Customer},
        with a delivery date if the purchase was successful.
\end{itemize}
\end{minipage}

\vspace{0.5em}
{\small \header{Local Types for $\role{Shop}$ role}}
\vspace{-1.1em}
{\footnotesize
\begin{minipage}{0.475\textwidth}
\[
\bl
\mkwd{ShopTy} \defeq \\
    \quad
    \bl
        \localoffersingle{Customer}{\mkwd{requestItems}}{} \then \\
        \localselectsingle{Customer}{\mkwd{items}}{[(\mkwd{ItemID} \times \mkwd{ItemName})]} \then \\
        \mkwd{ReceiveCommand}
    \el
\\ \\
\mkwd{PaymentResponse} \defeq \\
\quad \localofferone{PaymentProcessor} \\
        \qquad
        \blz
        \msg{\mkwd{ok}}{} \then \\
        \quad 
        \localselectsingle{Customer}{\mkwd{ok}}{\mkwd{DeliveryDate}} \then \\
        \quad 
        \mkwd{ReceiveCommand}, \\
        \msg{\mkwd{paymentDeclined}}{} \then \\
        \quad \localselectsingle{Customer}{\mkwd{paymentDeclined}}{} \then \\
        \quad
        \mkwd{ReceiveCommand} \\
        \el \\
        \quad \}
\el
\]
\end{minipage}
\hfill
\begin{minipage}{0.475\textwidth}
\[
\bl
    \mkwd{ReceiveCommand} \defeq \\
    \quad
        \blz
            \localofferone{Customer} \\
            \quad
            \blz
                \msg{\mkwd{getItemInfo}}{\mkwd{ItemID}} \then  \\
                \quad
                \localselectsingle{Customer}{\mkwd{itemInfo}}{\mkwd{Description}}
                \then \\
                \quad \mkwd{ReceiveCommand},
                \\
                \msg{\mkwd{checkout}}{([\mkwd{ItemID}] \times \mkwd{PaymentDetails})} \then \\
                \quad
                    \blz
                        \localselectone{Customer} \\
                        \quad
                        \msg{\mkwd{paymentProcessing}}{} \then \\
                        \qquad
                        \blz
                        \localselectnone{PaymentProcessor} \\
                        \quad \msg{\mkwd{buy}}{(\mkwd{PaymentDetails} \times \mkwd{Price})} \then \\
                        \quad \mkwd{PaymentResponse}, \\
                    \el \\
                    \quad \msg{\mkwd{outOfStock}}{} \then \\
                        \qquad \localselectsingle{Customer}{\mkwd{outOfStock}}{}
                        \then \\
                        \qquad \mkwd{ReceiveCommand} \\
                    \} \\
                \el
            \el \\
            \}
        \el
    \el
\]
\end{minipage}
}
\vspace{-1.5em}
\caption{Online Shop Scenario}
    \label{fig:intro:seqdiag}
\end{figure}

 \subsection{A Larger Example: A Shop}

Our ID server example illustrates key parts of \langname, but involves only two
roles. We now consider a larger online shop example
(Figure~\ref{fig:intro:seqdiag}) to serve as a running example. Multiple clients
interact with a single shop, which connects to an external payment processor.
Figure~\ref{fig:intro:seqdiag} shows the local types for \role{Shop}; we omit
$\mkwd{ClientTy}$ and $\mkwd{PPTy}$ for \role{Client} and
\role{PaymentProcessor}, which follow the same pattern. The global type closely
follows the sequence diagram.

\paragraph{Shop message handlers.}
Figure~\ref{fig:intro:shop-impl} shows the shop's message handlers. After
spawning, the shop suspends with \mkwd{itemReqHandler}, awaiting a
\mkwd{requestItems} message. On receipt, it retrieves the current stock from its
state, sends a summary to the customer, and installs \mkwd{custReqHandler}.

The \mkwd{custReqHandler} handles the \mkwd{getItemInfo} and \mkwd{checkout}
messages. For \mkwd{getItemInfo}, the shop sends the item details and suspends
recursively. For \mkwd{checkout}, it checks availability: if all items are
in stock, it notifies the customer, updates the stock, sends \mkwd{buy} to the
payment processor, and installs \mkwd{paymentHandler}; otherwise, it sends
\mkwd{outOfStock} and reinstalls \mkwd{custReqHandler}.

The \mkwd{paymentHandler} waits for the processor's reply: if it receives
\mkwd{ok}, it sends the delivery date; if it instead receives
\mkwd{paymentDeclined}, it restores the previous stock. Both branches reinstall
\mkwd{custReqHandler} to handle future requests.

\begin{figure}[t]
{\footnotesize
    \begin{minipage}{0.45\textwidth}
\[
    \bl
    \mkwd{itemReqHandler} : \handlerty{\mkwd{ShopTy}}{\tylist{\mkwd{Item}}} \\
    \mkwd{itemReqHandler} \defeq \\
        \quad
        \bl
        \handlerone{Customer}{\var{stock}} \\
        \quad \msg{\mkwd{requestItems}}{} \mapsto \\
\qquad
        \send{Customer}{\mkwd{itemSummary}}{\mkwd{summary}(\var{stock})}; \\
        \qquad \suspend{\mkwd{custReqHandler}}{\var{stock}} \:\: \}
        \vspace{0.5em}
        \el \\ 
\mkwd{paymentHandler} : [\mkwd{ItemID}] \to \\
    \qquad \handlerty{\mkwd{PaymentResponse}}{\tylist{\mkwd{Item}}} \\
    \mkwd{paymentHandler} \defeq \lambda \var{itemIDs} . \\
    \quad
        \bl
        \handlerone{PaymentProcessor}{\var{stock}} \\
        \quad \msg{\mkwd{ok}}{} \mapsto \\
        \qquad
        \bl
            \send{Customer}{\mkwd{ok}}{\mkwd{deliveryDate}(\var{itemIDs})}; \\
            \suspend{\mkwd{custReqHandler}}{\var{stock}}
            \el \\
        \quad \msg{\mkwd{paymentDeclined}}{} \mapsto \\
        \qquad
        \bl
            \send{Customer}{\mkwd{paymentDeclined}}{}; \\
            \letintwo{\var{newStock}}{\mkwd{increaseStock}(\var{itemIDs}, \var{stock})} \\
            \suspend{\mkwd{custReqHandler}}{\var{newStock}}
        \: \: \}
        \el
        \el
    \el
\]
\end{minipage}
\hfill
\begin{minipage}{0.5\textwidth}
\[
    \bl
    \mkwd{custReqHandler} : \handlerty{\mkwd{ReceiveCommand}}{\tylist{\mkwd{Item}}} \\
    \mkwd{custReqHandler} \defeq \\
    \quad
        \bl
        \handlerone{Customer}{\var{stock}} \\
        \quad \msg{\mkwd{getItemInfo}}{\var{itemID}} \mapsto \\
        \qquad \send{Customer}{\mkwd{itemInfo}}{\mkwd{lookupItem}(\var{itemID},
        \var{stock})}; \\
        \qquad \suspend{\mkwd{custReqHandler}}{\var{stock}} \\
        \quad \msg{\mkwd{checkout}}{(\var{itemIDs}, \var{details})} \mapsto \\
        \qquad \ithen{\mkwd{inStock}(\var{itemIDs}, \var{stock})} \\
        \qqquad \send{Customer}{\mkwd{paymentProcessing}}{}; \\
        \qqquad \letintwo{\var{total}}{\mkwd{cost}(\var{itemIDs}, \var{stock})} \\
        \qqquad \letintwo{\var{newStock}}{\mkwd{decreaseStock}(\var{itemIDs},
        \var{stock})} \\
        \qqquad \send{PaymentProcessor}{\mkwd{buy}}{(\var{details}, \var{total})}; \\
        \qqquad
        \suspend{(\mkwd{paymentHandler} \app \var{itemIDs})}{\var{newStock}} \\
        \qquad \calcwd{else} \\
        \qqquad \send{Customer}{\mkwd{outOfStock}}{}; \\
        \qqquad \suspend{\mkwd{custReqHandler}}{\var{stock}} \\
        \}
        \el 
    \el
\]
\end{minipage}
}
\caption{Implementation of Shop message handlers in \langname}
\label{fig:intro:shop-impl}
\end{figure}
 
\paragraph*{Tying the example together.}

Finally, we can show how to establish a session using the Shop actors.
Let $\mkwd{CustomerProtocol} = \set{\role{Shop} : \mkwd{ShopTy},
\role{Client} : \mkwd{ClientTy},
\role{PaymentProcessor} : \mkwd{PPTy}}$.

\smallmath{
\begin{minipage}{0.375\textwidth}
\[
    \hspace{-2em}
    \bl
    \mkwd{main} \defeq \\
    \: \letintwo{\var{custAP}}{\newap{\mkwd{CustomerProtocol}}} \\
\: \spawn{(\mkwd{shop} \app (\var{custAP}, \mkwd{initialStock}))}; \\
    \: \spawn{(\mkwd{paymentProcessor} \app \var{custAP})}; \\
    \: \spawn{(\mkwd{customer} \app \var{custAP})} \\
\el
\]
\end{minipage}
\hfill
\begin{minipage}{0.6\textwidth}
\[
    \bl
\mkwd{shop} \defeq \lambda (\var{custAP}, \var{stock}) . \;
        \mkwd{registerAgain} \app \var{custAP}; \; \var{stock}
\\\\
        \mkwd{registerAgain} \defeq \lambda \var{custAP} . \\
        \quad
            \registertwo{\var{custAP}}{\role{Shop}} \\
            \qquad (\fun{\var{st}}{\mkwd{registerAgain} \: \var{custAP};
            \suspend{\mkwd{itemReqHandler}}{\var{st}}} )
    \el
\]
\end{minipage}
}

The \mkwd{shop} definition takes the access point and then proceeds to
\emph{register} to take part in a session to interact with customers.
After each session has been established, the session type for the shop states
that it needs to receive a message from a client, so the shop suspends with
$\mkwd{itemReqHandler}$.

 \section{\langname: A Core Actor Language with Multiparty Session Types}
\label{sec:formalism}
In this section we introduce \langname formally, giving its syntax, typing rules, and semantics.

\subsection{Syntax}
\begin{figure}
 {\footnotesize
\header{Syntax of terms}
\[
    \begin{array}{lrcl}
        \text{Roles} & \prole, \qrole \\
        \text{Variables} & x, y, z, f \\
    \text{Values} & V, W & ::= & x \midspace \fun{x}{M} \midspace \rec{f}{x}{M} \midspace c 
        \midspace (\vala, \valb) \midspace \handler{p}{\var{st}}{\seq{H}} \\
        \text{Handler clauses} & H & ::= & \msg{\ell}{x} \mapsto M \\
        \text{Computations} & M, N & ::= &
            \efflet{x}{M}{N} \midspace \effreturn{V} \midspace V \app W \\
                            & & \midspace & \ite{V}{M}{N} \midspace \letin{(x, y)}{\vala}{\tma} \\
                            & & \midspace &  \spawn{M} \midspace \send{p}{\ell}{V}  \midspace \suspend{\vala}{\valb} \\
                            & & \midspace &  \newap{\proto} \midspace \register{\vala}{\prole}{\valb} \\
    \end{array}
\]

 \header{Syntax of types and type environments}
 \[
     \begin{array}{lrcl}
         \text{Output session types} & \stout & ::= & \localselect{p}{\msg{\ell_i}{\tya_i}.S_i}_{i \in I} \\
         \text{Input session types} & \stin & ::= & \localoffer{p}{\msg{\ell_i}{\tya_i}.S_i}_{i \in I} \\
         \text{Session types} & S, T & ::= & \stout \midspace \stin
              \midspace
             \recty{X}{S} \midspace
             \recvar \midspace
             \localend
             \\
         \text{Protocols} & \proto & ::= & \set{\prole_i : \sta_i}_{i \in I} \\
         \text{Types} & \tya, \tyb, \tyc & ::= & \basety \midspace \sttyfun{\tya}{\tyb}{S}{T}{\tyc} \midspace
         (\tya \times \tyb) \midspace 
         \apty{(\prole_i:  \sta_i)_{i \in I}} \midspace \handlerty{\stin}{\tyc} \\
         \text{Base types} & \basety & ::= & \one \midspace \tybool \midspace \tyint \midspace \cdots \\
         \text{Type environments} & \tyenv & ::= & \cdot \midspace \tyenv, x : \tya
     \end{array}
 \]
 }
     \caption{Syntax of terms, types, and type environments}
     \label{fig:syntax}
 \end{figure}
 
Figure~\ref{fig:typing} shows the syntax of \langname.
We let $\prole, \qrole$ range over roles, and $x, y, z, f$
range over variables.
We stratify the calculus into values $\vala, \valb$ and computations $\tma,
\tmb$ in the style of \emph{fine-grain call-by-value}~\cite{LPT03}, with
different typing judgements for each.

\paragraph*{Session types.}
Although global types are convenient for describing protocols, we instead
follow~\citet{ScalasY19} and base our formalism around local types
(\emph{projection} of global types onto roles is standard~\cite{HondaYC08,
ScalasDHY17}; the local types resulting from projecting a global type satisfy
the properties that we will see in~\secref{sec:metatheory}~\cite{ScalasY19}).
\emph{Selection} session types
$\localselect{\prole}{\msg{\ell_i}{\tya_i}\then \sta_i}_{i \in I}$ indicate
that a process can choose to send a message with label $\ell_j$ and payload type
$\tya_j$ to role $\prole$, and continue as session type $\sta_j$ (assuming $j \in I$).
\emph{Branching} session types
$\localoffer{\prole}{\msg{\ell_i}{\tya_i} \then \sta_i}_{i \in I}$ indicate that
a process must \emph{receive} a message. We let $\stout$ range over selection
(or \emph{output}) session types, and let $\stin$ range over branching (or
\emph{input}) session types.
Session type $\recty{\recvar}{\sta}$ indicates a recursive session type 
that binds variable $\recvar$ in $\sta$; we take an equi-recursive
view of session types and identify each recursive session type with its
unfolding. Finally, $\localend$ denotes a session type that has finished.

\paragraph*{Protocols.}
A \emph{protocol} $\proto$ is a collection of roles and their associated session
types. Protocols are used when defining access points, and in
Section~\ref{sec:metatheory} when describing behavioural properties.

\paragraph*{Types.}
Base types $\basety$ are standard. Since our type system enforces session typing
by pre- and post-conditions,
a function type $\sttyfun{\tya}{\tyb}{\sta}{\stb}{\tyc}$ states that the function takes
an argument of type $\tya$ where the current session type is $\sta$, and
produces a result of type $\tyb$ with resulting session type $\stb$, to be run
on an actor with state of type $\tyc$.
An access point has type $\apty{(\prole_i: \sta_i)_{i \in 1..n}}$, mapping each role to a
local type. Finally, a message handler has type $\handlerty{\stin}{\tya}$ where
$\stin$ is an \emph{input} session type and $\tya$ is the type of the actor state.

\subsection{Typing Rules}
\begin{figure}
{\footnotesize
\headersig{Value typing}{$\vseq[]{\tyenv}{V}{\tya}$}
\vspace{-1em}

\begin{mathpar}
    \inferrule
    [TV-Var]
    { x : \tya \in \tyenv }
    { \vseq{\tyenv}{x}{\tya}}

    \inferrule
    [TV-Const]
    { c \text{ has base type } \basety }
    { \vseq{\tyenv}{c}{\basety}}

    \inferrule
    [TV-Lam]
    {
        \mtseqst{\tyenv, x : \tya}{\tyc}{S}{M}{\tyb}{T}
    }
    { \vseq{\tyenv}{\fun{x}{M}}{\sttyfun{\tya}{\tyb}{S}{T}{\tyc}} }

    \inferrule
    [TV-Rec]
    {
    \mtseqst
        {\tyenv, x: A, f : \sttyfun{\tya}{\tyb}{S}{T}{\tyc}}
        {\tyc}
        {S}
        {M}
        {\tyb}
        {T}
    }
    {\vseq{\tyenv}{\rec{f}{x}{M}}{\sttyfun{\tya}{\tyb}{S}{T}{\tyc}} }

    \inferrule
    [TV-Pair]
    {
        \vseq{\tyenv}{\vala}{\tya}
        \\
        \vseq{\tyenv}{\valb}{\tyb}
    }
    { \vseq{\tyenv}{(\vala, \valb)}{(\tya \times \tyb)} }

    \inferrule
    [TV-Handler]
    {
        (\mtseqst{\tyenv, x_i : \tya_i, \var{st} : \tyc}{\tyc}{S_i}{M_i}{\tyc}{\localend} )_{i \in I}
    }
    { \vseq
        {\tyenv}
        {\handler{p}{\var{st}}{\msg{\ell_i}{x_i} \mapsto {M_i}}_{i \in I}}
        {\handlerty{\localoffer{p}{\msg{\ell_i}{\tya_i}.S_i}_{i \in I}}{\tyc}}
    }
\end{mathpar}

\headersig{Computation typing}{$\mtseqst[]{\tyenv}{\tyc}{S}{M}{\tya}{T}$}
\vspace{-1em}

\begin{mathpar} 
    \inferrule
    [T-Let]
    {
        \mtseqst{\tyenv}{\tyc}{S_1}{M}{\tya}{S_2} \\\\
        \mtseqst{\tyenv, x : \tya}{\tyc}{S_2}{N}{\tyb}{S_3}
    }
    { \mtseqst{\tyenv}{\tyc}{S_1}{\efflet{x}{M}{N}}{\tyb}{S_3} }

    \inferrule
    [T-Return]
    { \vseq{\tyenv}{V}{\tya} }
    { \mtseqst{\tyenv}{\tyc}{S}{\effreturn{V}}{\tya}{S} }

    \inferrule
    [T-App]
    {\vseq{\tyenv}{V}{\sttyfun{\tya}{\tyb}{S}{T}{\tyc}} \\
     \vseq{\tyenv}{W}{\tya}}
     {\mtseqst{\tyenv}{\tyc}{S}{V \app W}{\tyb}{T} }

    \inferrule
    [T-LetPair]
    {
        \vseq{\tyenv}{\vala}{(\tya_1 \times \tya_2)}
        \\\\
        \mtseqst{\tyenv, x : \tya_1, y: \tya_2}{\tyc}{\sta_1}{\tma}{\tyb}{\sta_2}
    }
    { \mtseqst{\tyenv}{\tyc}{\sta_1}{\letin{(x, y)}{\vala}{\tma}}{\tyb}{\sta_2} }
\quad
    \inferrule
    [T-If]
    {
        \vseq{\tyenv}{V}{\tybool} \\\\
        \mtseqst{\tyenv}{\tyc}{S_1}{M}{\tya}{S_2} \\\\
        \mtseqst{\tyenv}{\tyc}{S_1}{N}{\tya}{S_2}
    }
    {\mtseqst{\tyenv}{\tyc}{S_1}{\ite{V}{M}{N}}{A}{S_2}}
\quad
    \inferrule
    [T-Spawn]
    {
        \mtseqst{\tyenv}{\tya}{\localend}{\tma}{\tya}{\localend} \\
    }
    { \mtseqst{\tyenv}{\tyc}{\sta}{\spawn{\tma}}{\one}{\sta} }

    \inferrule
    [T-Send]
    {
        j \in I \\
        \vseq{\tyenv}{V}{\tya_j}
    }
    {
        \mtseqst
        {\tyenv}
        {\tyc}
        {\localselect{p}{\msg{\ell_i}{\tya_i}.{S_i}}_{i \in I}}
        {\send{p}{\ell_j}{V}}
        {\one}
        {S_j}
    }

    \inferrule
    [T-Suspend]
    {
        \vseq{\tyenv}{\vala}{\handlerty{\stin}{\tyc} } \\
        \vseq{\tyenv}{\valb}{\tyc}
    }
    { \mtseqst{\tyenv}{\tyc}{\stin}{\suspend{\vala}{\valb}}{\tya}{S'} }

    \inferrule
    [T-NewAP]
    {  \compliant{(\role{p}_i : T_i)_{i \in I})}  }
    { \mtseqst
        {\tyenv}
        {\tyc}
        { S }
        {\newap{(\role{p}_i : T_i )_{i \in I}}}
        {\apty{(\role{p}_i : T_i )_{i \in I}}}
        { S }
    }
\quad
\inferrule
    [T-Register]
    {
        j \in I \\
        \vseq{\tyenv}{V}{\apty{(\role{p}_i : T_i)_{i \in I}}} \\
        \valseq{\tyenv}{\valb}{\sttyfun{\tya}{\tya}{\stb_j}{\localend}{\tya}}
    }
    {
        \mtseqst
        { \tyenv }
        { \tya }
        { \sta }
        { \register{\vala}{\prole_j}{\valb} }
        { \one }
        { \sta }
    }
\end{mathpar}
}

    \caption{\langname Typing Rules}
    \label{fig:typing}
\end{figure}

 \paragraph*{Values.}
Fig.~\ref{fig:typing} gives the typing rules for \langname.
Value typing $\vseq{\tyenv}{\vala}{\tya}$ states
that value $\vala$ has type $\tya$ under environment $\tyenv$.
Session typing is enforced by flow-sensitive effect typing
(following~\citet{HarveyFDG21}) and so we do not need linear types.
Typing for variables and constants is standard (we assume constants
include the unit value $()$ of type $\one$), and typing rules for
functions and recursive functions are adapted to include session types.
Message handlers specify how to handle incoming messages:
\textsc{TV-Handler} states that
$\handler{\prole}{\var{st}}{\msg{\ell_i}{x_i} \mapsto \tma_i}_i$ is typable with
type $\handlerty{\localoffer{\prole}{\msg{\ell_i}{\tya_i}\then\sta_i}_i}{\tyc}$
if each $\tma_i$ is typable with precondition $\sta_i$ where the environment is
extended with $x_i$ of type $\tya_i$ and $\var{st}$ of type $\tyc$, and all
branches have postcondition $\localend$.

\paragraph*{Computations.}
The computation typing judgement has the form
$\mtseqst{\tyenv}{\tyc}{\sta}{\tma}{\tya}{\stb}$, read as ``under type
environment $\tyenv$ and evaluating in an actor with state of type $\tyc$,
given session precondition $\sta$, term $\tma$ has type
$\tya$ and postcondition $\stb$''.
A let-binding $\efflet{x}{\tma}{\tmb}$ evaluates $\tma$ and binds its result to
$x$ in $\tmb$, with the session postcondition from typing $\tma$ used as the
precondition when typing $\tmb$ (\textsc{T-Let}); note that this is the
only evaluation context in the system. The $\effreturn{\vala}$ expression
is a trivial computation returning value $\vala$ and has type $\tya$ if $\vala$
also has type $\tya$ (\textsc{T-Return}).  
A function application $\vala \app
\valb$ is typable by \textsc{T-App} provided that the precondition in the
function type matches the current precondition, and advances the postcondition
to that of the function type.
Rule \textsc{T-LetPair} types a pair deconstruction by binding both pair
elements in the continuation $\tma$.
Rule \textsc{T-If} types a conditional if its
condition is of type $\tybool$ and both continuations have the same return type
and postcondition;
this design is in keeping with analogous session type
systems~\cite{HarveyFDG21,FrancalanzaT23}, but by
treating $\calcwd{suspend}$ as a control
operator (with an arbitrary return type and postcondition) we can maintain
expressiveness by allowing each branch
to finish at a different session type.

The $\spawn{\tma}$ construct spawns a new actor that evaluates term $\tma$;
rule \textsc{T-Spawn} states that if the spawned actor supports state type
$\tyc$, then $\tma$ must also have type $\tyc$ to return the initial state.
It must also have pre- and postconditions $\localend$ because the spawned
computation is not yet in a session and so cannot communicate.
Rule \textsc{T-Send} types a send computation $\send{\prole}{\ell}{\vala}$ if
$\ell$ is contained within the selection session precondition, and if $\vala$
has the corresponding type; the postcondition is the session continuation for
the specified branch.
There is \emph{no $\calcwd{receive}$ construct}, since receiving messages is
handled by the event loop. Instead, when an actor wishes to receive a message,
it must \emph{suspend} itself with updated state $\valb$ and install a message handler using
$\suspend{\vala}{\valb}$. The \textsc{T-Suspend} rule states that
$\suspend{\vala}{\valb}$ is
typable if the handler is compatible with the current session type precondition
and state type; since the computation does not return, it can be given an
arbitrary return type and postcondition.

Sessions are initiated using \emph{access points}: we create an access point
for a session with roles and types $(\prole_i : \sta_i)_i$
using $\newap{(\prole_i : \sta_i)_i}$, which must be
annotated with the set of roles and local types to be involved in the
session (\textsc{T-NewAP}).
The rule ensures that the protocol supported by the access point is
\emph{compliant}; will describe this further in~\secref{sec:metatheory}, but at
a high level, if a protocol is compliant then it is free of communication
mismatches and deadlocks.

An actor can \emph{register} to
take part in a session as role $\prole$ on access point $\vala$ using
$\register{\vala}{\prole}{\valb}$; function $\valb$ is a callback to be
invoked once the session is established. Rule \textsc{T-Register} ensures that
the access point must contain a session type $\stb$ associated with
role $\prole$, and since the initiation callback will be evaluated when the
session is established, $\tma$ must be typable under session type $\stb$.
Since neither $\calcwd{newAP}$ nor $\calcwd{register}$ perform any
communication, the session types are unaltered.

The typing rules are not syntax-directed due to the rule for
\textsc{T-Suspend}, however they can be made algorithmic through type
annotations or standard type inference techniques.

\begin{figure}[t]
    {\footnotesize
    \header{Runtime syntax}
    \vspace{-1em}
    \begin{minipage}{0.35\textwidth}
    \[
        \begin{array}{lrcl}
            \text{Actor names} & a, b \\
            \text{Session names} & s \\
            \text{AP names} & \apname \\
            \text{Init.\ tokens} & \inittok \\
            \text{Runtime names} & \nma & ::= & \aname \midspace \sessname \midspace \apname \midspace \inittok
                                 \\
            \text{Values} & \hspace{-2em} \valc, \vala, \valb & ::= & \cdots \midspace \apname \\
            \text{Type env.} & \tyenv & ::= & \cdots  \\
                                     & & \midspace & \tyenv, \apname :
                                     \apty{(\prole_i : \sta_i)_i} \\
           \text{Reduction labels} & \redlbl & ::= & s \midspace \tau
        \end{array}
    \]
    \end{minipage}
    \hfill
    \begin{minipage}{0.55\textwidth}
        \[
            \begin{array}{lrcl}
                \text{Configurations} & \config{C}, \config{D} & ::= &
                    (\nu \nma) \config{C} \midspace
                    \config{C} \parallel \config{D} \\
                                      & & \midspace &
                                      \actor{\threadt}{\hstate}{\istate} \midspace
                    \ap{\apname}{\apstate} \midspace \qproc{s}{\qcontents} \\
                \text{Message queues} & \qcontents & ::= & \epsilon \midspace \qentry{p}{q}{\ell}{\vala} \cdot \qcontents \\
                \text{Stored handlers} & \hstate & ::= & \epsilon \midspace \hstate,
                \storedhandler{s}{p}{V} \\
                \text{Initialisation states} & \istate & ::= & \epsilon
                \midspace \istate, \maptwo{\inittok}{\vala} \\
                \text{Thread states} & \threadt & ::= &
                \idle{\vala} \midspace \sessthread{s}{p}{M} \midspace \standalone{M} \\
                \text{Access point states} & \apstate & ::= & (\prole_i \mapsto
                \setseq{\inittok_i})_{i} \\
\text{Evaluation contexts} & \ctxe & ::= & [~] \midspace \efflet{x}{\ctxe}{M} \\
                \text{Thread contexts} & \threadctx & ::= & \ctxe \midspace \sessthread{s}{p}{\ctxe} \\
                \text{Top-level contexts} & \tlthreadctx & ::= &
                    [~] \midspace \sessthread{s}{p}{[~]}
            \end{array}
        \]
    \end{minipage}

\headersig{Structural congruence (configurations)}{$\config{C} \equiv \config{D}$}
\begin{mathpar}
    \inferrule
    {}
    { \config{C} \parallel \config{D} \equiv \config{D} \parallel \config{C} }

    \inferrule
    {}
    { \config{C} \parallel (\config{D} \parallel \config{D}') \equiv
      (\config{C} \parallel \config{D}) \parallel \config{D}'
    }

    \inferrule
    {}
    {(\nu \nma_1)(\nu \nma_2)\config{C} \equiv (\nu \nma_2)(\nu \nma_1)\config{C}}

    \inferrule
    {}
    { (\nu \sessname)(\qproc{\sessname}{\epsilon}) \parallel \config{C} \equiv \config{C} }

    \inferrule
    { \nma \not\in \fn{\config{C}}}
    { \config{C} \parallel (\nu \nma) \config{D} \equiv (\nu \nma)( \config{C} \parallel \config{D}) }
\quad 
    \inferrule
    {
        \role{p}_1 \ne \role{p}_2 \vee
        \role{q}_1 \ne \role{q}_2
    }
    {
        \qproc{s}{ \sigma_1 {\cdot} \qentryy{\role{p}_1}{\role{q}_1}{\ell_1}{\vala_1} 
        {\cdot} \qentryy{\role{p}_2}{\role{q}_2}{\ell_2}{\vala_2} {\cdot} \sigma_2}
        \equiv
    \qproc{s}{ \sigma_1 {\cdot} \qentryy{\role{p}_2}{\role{q}_2}{\ell_2}{\vala_2}
    {\cdot} \qentryy{\role{p}_1}{\role{q}_1}{\ell_1}{\vala_1} {\cdot} \sigma_2}
    }
\end{mathpar}
}
    \vspace{-1.1em}
    \caption{Operational semantics (1)}
    \label{fig:semantics}
    \vspace{-1em}
\end{figure}

\subsection{Operational Semantics}
\paragraph*{Runtime syntax.}
Modelling the concurrent behaviour of \langname processes, requires additional
runtime syntax (Fig.~\ref{fig:semantics}).
Runtime names are identifiers for runtime entities: actor names
$\aname$ identify actors; session names $\sessname$ identify established
sessions; access point names $\apname$ identify access points; and
\emph{initialisation tokens} $\inittok$ associate entries in an
access point with registered initialisation continuations.

We model communication and concurrency through a language of
\emph{configurations} (reminiscent of $\pi$-calculus processes).
A \emph{name restriction} $(\nu \nma) \config{C}$ binds runtime name $\nma$ in
configuration $\config{C}$, and the right-associative parallel composition
$\config{C} \parallel \config{D}$ denotes configurations $\config{C}$ and
$\config{D}$ running in parallel.

An actor is represented as a 4-tuple $\actor{\threadt}{\hstate}{\istate}$, where
$\threadt$ is a thread that can either be idle with state $\vala$
($\idle{\vala}$);
a term $\tma$ that is not involved in a session;
or $\sessthread{s}{\prole}{\tma}$ denoting
that the actor is evaluating term $\tma$ playing role $\prole$ in session $s$. 
An actor is \emph{active} if its thread is $\tma$ or
$\sessthread{s}{\prole}{\tma}$ (for some $s$, $\prole$, and $\tma$), and
\emph{idle} otherwise.
A handler state $\hstate$ maps
endpoints to handlers, which are invoked when an incoming message is
received and the actor is idle. The initialisation state $\istate$
maps initialisation tokens to callbacks to be invoked whenever a
session is established.
Our reduction rules (Figure~\ref{fig:semantics-2}) make use of indexing notation as syntactic sugar
for parallel composition: for example,
$\actor[a_i]{\threadt_i}{\hstate_i}{\istate_i}_{i \in 1..n}$ is syntactic sugar
for the configuration
$\actor[a_1]{\threadt_1}{\hstate_1}{\istate_1} \parallel \cdots
\parallel \actor[a_n]{\threadt_n}{\hstate_n}{\istate_n}$.

An access point $\ap{\apname}{\apstate}$ has name $\apname$ and state
$\apstate$, where the state maps roles to sets of initialisation tokens for
actors that have registered to take part in the session.
Finally, each session $s$ is associated with a queue $\qproc{s}{\qcontents}$,
where $\qcontents$ is a list of entries $\qentry{\prole}{\qrole}{\ell}{\vala}$
denoting a message $\msg{\ell}{\vala}$ sent from $\prole$ to $\qrole$.

\paragraph*{Initial configurations.}
A program $\tma$ is run by placing it in an \emph{initial configuration}
$(\nu a)(\actor{\tma}{\epsilon}{\epsilon})$.

\paragraph*{Structural congruence and term reduction.}
Structural congruence is the smallest congruence relation defined by the axioms
in Figure~\ref{fig:semantics}. As with the $\pi$-calculus, parallel
composition is associative and commutative, and we have the usual scope
extrusion rule; we write $\fn{\config{C}}$ to refer to the set of free names in a
configuration $\config{C}$. We also include a structural congruence rule on
queues that allows us to reorder unrelated messages; notably this rule maintains
message ordering between pairs of participants.  Consequently, the session-level
queue representation is isomorphic to a set of queues between each pair of
roles.  Term reduction $\tma \teval \tmb$ is standard $\beta$-reduction
(omitted).

\begin{figure}[t]
    {\footnotesize
\headersig{Configuration reduction}{$\config{C} \cevalann{\redlbl} \config{D}$}
    \begin{mathpar}
    \inferrule
    [E-Send]
    { }
    {
        \actor
        { \sessthread{s}{p}{\ctxe[\send{q}{\ell}{\vala}]} }
        { \hstate }
        { \istate }
        \parallel
        \qproc{s}{\qcontents}
\cevalann{s} \\\\
\actor
        { \sessthread{s}{p}{\ctxe[\effreturn{()}]} }
        { \hstate }
        { \istate }
        \parallel
        \qproc{s}{\qcontents {\cdot} \qentry{p}{q}{\ell}{\vala}}
    }
\quad
\inferrule
    [E-React]
    {
        (\ell(x) \mapsto \tma) \in \seq{H}
    }
    {
        \actor
        { \idle{\valb} }
        { \hstate[\storedhandler{s}{p}{\handler{q}{\var{st}}{\seq{H}}}] }
        { \istate }
        \parallel
        \qproc{s}{\qentry{q}{p}{\ell}{\vala} {\cdot} \qcontents} \\\\
        \cevalann{s}
\actor
        { \sessthread{s}{\prole}{\tma \{ \vala / x, \valb / \var{st} \}} }
        { \hstate }
        { \istate }
        \parallel
        \qproc{s}{\qcontents}
    }
\\
\inferrule
    [E-Suspend]
    { }
    { \actor
        {\sessthread{s}{\prole}{\ctxe[\suspend{\vala}{\valb}]}}
        {\hstate}
        {\istate}
\cevalann{\tau}\\\\
\actor
      {\idle{\valb}}
        {\hstate[\storedhandler{s}{\prole}{V}]}
        {\istate}
    }
\quad
\inferrule
    [E-Spawn]
    { }
    { \actor
        { \threadctx[\spawn{\tma}] }
        { \hstate }
        { \istate }
\cevalann{\tau}\\\\
(\nu b)(
      \actor
        { \threadctx[\effreturn{()}] }
        { \hstate }
        { \istate }
        {\,\parallel\,}
      \actor
        [b]
        { \tma }
        { \epsilon }
        { \epsilon }
        )
    }
\quad
\inferrule
    [E-Reset]
    { }
    {
        \actor
          { \tlthreadctx[\effreturn{\vala}] }
          { \hstate }
          { \istate }
\cevalann{\tau}\\\\
\actor
        { \idle{\vala} }
          { \hstate }
          { \istate }
    }

\inferrule
    [E-NewAP]
    { \apname \text{ fresh}}
    {
      \actor
        { \threadctx[\newap{(\role{p}_i : \sta_i )_{i \in I}}] }
        { \hstate }
        { \istate }
\cevalann{\tau}\\\\
(\nu \apname)(
          \actor
            { \threadctx[\effreturn{\apname}] }
            { \hstate }
            { \istate }
\parallel
\ap{\apname}{(\prole_i \mapsto \emptyset)_{i \in I}}
        )
    }
\quad
\inferrule
    [E-Register]
    { \inittok \text{ fresh} }
    { \actor
        {\threadctx[\register{\apname}{\prole}{\vala}]}
        {\hstate}
        {\istate}
        \parallel \ap{\apname}{\apstate[\prole \mapsto \setseq{\inittok'}]}
        \cevalann{\tau}\\\\
        (\nu \inittok)
        (\actor
            {\threadctx[\effreturn{()}]}
            {\hstate}
            {\istate[\maptwo{\inittok}{\vala}]}
            \parallel
        \ap{\apname}{\apstate[\prole \mapsto \setseq{\inittok'} \cup \set{\inittok}]})
    }

    \inferrule
    [E-Init]
    { \sessname \text{ fresh} }
    {
        (\nu \inittok_{\prole_i})_{i \in 1..n}(
        \ap
            {\apname}
            {(\prole_i \mapsto \setseq{\inittok'_{\prole_i}} \cup \set{\inittok_{\prole_i}})_{i \in 1..n} }
        \parallel
        \actor
            [a_i]
            {\idle{\valb_i}}
            {\hstate_i}
            {{\istate_i[\inittok_{\prole_i} \mapsto \vala_i]}}_{i \in 1..n}
            )
\cevalann{\tau}\\\\
(\nu \sessname)(
            \ap{\apname}{(\prole_i \mapsto \setseq{\inittok'_{\prole_i}})_{i \in 1..n}}
            \parallel
            \qproc{s}{\epsilon}
            \parallel
            \actor
                [a_i]
                {\sessthreadd{\sessname}{\prole_i}{\vala_i \app \valb_i}}
                {\hstate_i}
                {\istate_i}_{i \in 1..n}
                )
    }

    \inferrule
    [E-Par]
    { \config{C} \cevalann{l} \config{C}' }
    { \config{C} \parallel \config{D} \cevalann{l} \config{C}' \parallel \config{D} }

\inferrule
    [E-Lift]
    { M \teval N }
    { \actor
        {\threadctx[M]}
        {\hstate}
        {\istate}
\cevalann{\tau}
\actor
        {\threadctx[N]}
        {\hstate}
        {\istate}
    }

    \inferrule
    [E-Nu]
    { \config{C} \cevalann{l} \config{D} }
    { (\nu \nma) \config{C} \cevalann{l - \nma} (\nu \nma) \config{D} }

    \inferrule
    [E-Struct]
    { \config{C} \equiv \config{C}' \\
        \config{C}' \cevalann{l} \config{D}' \\
        \config{D}' \equiv \config{D}
    }
    { \config{C} \cevalann{l} \config{D} }
\end{mathpar}
\[
    \text{where } l - \nma = \tau \text{ if } l = \nma, \text{ and } l \text{ otherwise}
\]

    }
    \caption{Operational semantics (2)}
    \label{fig:semantics-2}
\end{figure}

 \paragraph*{Communication and concurrency.}
It is convenient for our metatheory to annotate each communication reduction
with the name of the session in which the communication occurs, although we
sometimes omit the label where it is not relevant.
Rule
\textsc{E-Send} describes an actor playing role $\prole$ in session $s$ sending
a message $\msg{\ell}{\vala}$ to role $\qrole$: the message is appended to the
session queue and the operation reduces to $\effreturn{()}$.
The \textsc{E-React} rule captures the event-driven nature of the system: if an
actor is idle with state $\vala$, and has a stored handler for
$\roleidx{s}{\prole}$, and there exists a matching message in the
session queue, then the message is dequeued and the message handler is
evaluated with the message payload and state.
If an actor is currently evaluating a computation in the context of a session
$\roleidx{s}{\prole}$, rule \textsc{E-Suspend} evaluates $\suspend{\vala}{\valb}$ by
installing handler $\vala$ for $\roleidx{s}{\prole}$ and returning the actor to
the $\idle{\valb}$ state.
Rule \textsc{E-Spawn} spawns a fresh actor with empty handler and initialisation
state, and \textsc{E-Reset} returns an actor to the $\idle{\vala}$ state once it
has finished evaluating to an updated state $\vala$.

\paragraph*{Session initialisation.}
Rule \textsc{E-NewAp} creates an access point with a fresh name $\apname$ and
empty mappings for each role.
Rule \textsc{E-Register} evaluates
$\register{\apname}{\prole}{\vala}$ 
by creating an initialisation token $\inittok$, storing a mapping from
$\inittok$ to the callback $\vala$ in the requesting actor's initialisation
state, and appending $\inittok$ to the participant set for $\prole$ in
$\apname$.
Finally, \textsc{E-Init} establishes
a session when idle participants are registered for all roles: the rule discards
all initialisation tokens, creates a session name restriction and empty session
queue, and invokes all initialisation callbacks.

\begin{example}
    Consider a simple Ping-Pong example.  We can describe the protocol as:

    \smallmath{
        \mkwd{PingPong} = \setlr{
        \bl
        \role{Pinger}:
            \localselectsingle{Ponger}{\mkwd{Ping}}{\one} \then 
            \localoffersingle{Ponger}{\mkwd{Pong}}{\one} \then \localend,\\
        \role{Ponger}:
            \localoffersingle{Pinger}{\mkwd{Ping}}{\one} \then 
            \localselectsingle{Pinger}{\mkwd{Pong}}{\one} \then \localend
            \el
        }}

        We will abbreviate \role{Pinger} as \role{Pi} and \role{Ponger} as
        \role{Po}.
        The \mkwd{main} function and the initialisation functions for the Pinger
        and Ponger are described as:

        \smallmath{
            \mkwd{main} \defeq
                \letin{\var{ap}}{\newap{\mkwd{PingPong}}} \spawn{\mkwd{pinger}
                \app \var{ap}}; \; {\spawn{\mkwd{ponger} \app \var{ap}}}
                \vspace{-1em}
        }
        
    \begin{minipage}[t]{0.475\textwidth}
        {\scriptsize
        \[
            \blt
        \mkwd{ponger} \defeq \lambda \var{ap} \then
                \register{\var{ap}}{\role{Po}}{\mkwd{pongerCallback}} \\
            \mkwd{pongerCallback} \defeq \lambda () . \\
            {\blt
                            \suspend
                                {
                                    (\handler
                                        {Pi}
                                        {\var{st}}
                                        {\mkwd{Ping} \mapsto
                                        \send{Pi}{\mkwd{Pong}}{()}})
                                }
                                {()}
            \el
            }
            \el
        \] 
    }
    \end{minipage}
    \hfill
    \begin{minipage}[t]{0.475\textwidth}
        {\scriptsize
        \[
            \bl
            \mkwd{pinger} \defeq \lambda \var{ap} \then \register{\var{ap}}{\role{Pi}}{\mkwd{pingerCallback}} \\
\mkwd{pingerCallback} \defeq \lambda () . \\
                {
                \blt
                    \send{Po}{\mkwd{Ping}}{()}; \\
                    \suspend
                        {
                            (\handler
                                {Po}
                                {\var{st}}
                                {\mkwd{Pong} \mapsto ()})
                        }
                        {()}
                \el
                } 
                \el
        \]
        }
    \end{minipage}

    With these defined, we place the \mkwd{main} function in an initial
    configuration, which creates a new access point $p$ (\textsc{E-NewAP}) and
    spawns the Pinger and Ponger actors (\textsc{E-Spawn}):

    \smallermath{
        (\nu a)(\actor{\mkwd{main}}{\epsilon}{\epsilon})
            \:\: \ceval^+ \:\:
            (\nu \var{ping})(\nu \var{pong})(\nu p)(\nu a)
\left(
\bl
    \actor{\idle{()}}{\epsilon}{\epsilon}
    \parallel \actor[\var{ping}]{\mkwd{pinger}}{\epsilon}{\epsilon} 
    \parallel \actor[\var{pong}]{\mkwd{ponger}}{\epsilon}{\epsilon} \\
    \parallel \ap{p}{\role{Pi} \mapsto \emptyset, \role{Po} \mapsto \emptyset}
\el
\right)
}

    At this point, both of the actors can register with the access point
    (\textsc{E-Register}). By registering, the access points generate
    \emph{initialisation tokens} $\iota_1, \iota_2$, which are stored both in
    the access point and also as keys in the actors' initialisation states. 
    The actors then revert to being idle (\textsc{E-Reset}).

\smallermath{
        \ceval^+
        (\nu \inittok_1)(\nu \inittok_2)(\nu \var{ping})(\nu \var{pong})(\nu p)(\nu a)
\left(
\bl
    \actor{\idle{()}}{\epsilon}{\epsilon} \\
    \parallel \actor[\var{ping}]{\idle{()}}{\epsilon}{\inittok_1 \mapsto
    \mkwd{pingerCallback}} \\
    \parallel \actor[\var{pong}]{\idle{()}}{\epsilon}{\inittok_2 \mapsto
    \mkwd{pongerCallback}} \\
    \parallel \ap{p}{\role{Pi} \mapsto \set{\inittok_1},
        \role{Po} \mapsto \set{\inittok_2}}
\el
\right)
}

    Since the access point now has idle registered actors for each role, it
    establishes a session and removes the initialisation tokens
    (\textsc{E-Init}).
    Both actors evaluate their initialisation callbacks in the context of the
    newly-created session:

\smallermath{
        \ceval^+
        (\nu s)(\nu \var{ping})(\nu \var{pong})(\nu p)(\nu a)
\left(
\bl
    \actor{\idle{()}}{\epsilon}{\epsilon} \\
    \parallel \actor[\var{ping}]{\sessthread{s}{Pi}{\mkwd{pingerCallback} \app ()}}{\epsilon}{\epsilon}
    \parallel \actor[\var{pong}]{\sessthread{s}{Po}{\mkwd{pongerCallback} \app ()}}{\epsilon}{\epsilon} \\
    \parallel \ap{p}{\role{Pi}\mapsto \emptyset, \role{Po} \mapsto \emptyset}
    \parallel \qproc{s}{\epsilon}
\el
\right)
    }    

    Following the behaviour in the callbacks, the Ponger suspends awaiting a
    message (\textsc{E-Suspend}), and the Pinger sends a message to the Ponger,
    which is stored in the session queue (\textsc{E-Send}).

    \smallermath{
        \ceval^+
        (\nu s)(\nu \var{ping})(\nu \var{pong})(\nu p)(\nu a)
\left(
\bl
    \actor{\idle{()}}{\epsilon}{\epsilon}
    \parallel \actor[\var{ping}]{ 
        \sessthread{s}{Pi}{
        \suspend
            { (\handler
                    {Po}
                    {\var{st}}
                    {\mkwd{Pong} \mapsto ()}) 
            }
            {()}}
    }{\epsilon}{\epsilon} \\
    \parallel 
        \actor
            [\var{pong}]
            {\idle{()}}
            {\roleidx{s}{Po} \mapsto
            \handler{Pi}{\var{st}}{\mkwd{Ping} \mapsto
        \send{Pi}{\mkwd{Pong}}{()}}}
            {\epsilon} \\
    \parallel \ap{p}{\role{Pi}\mapsto \emptyset, \role{Po} \mapsto \emptyset}
    \parallel \qproc{s}{\qentry{Pi}{Po}{\mkwd{Ping}}{()}}
\el
\right)
    }

    The Pinger can now suspend, awaiting a message from the Ponger
    (\textsc{E-Suspend}).  Since there is a queued message for the idle Ponger,
    we can re-activate the suspended handler (\textsc{E-React}):

    \smallermath{
        \ceval^+
        (\nu s)(\nu \var{ping})(\nu \var{pong})(\nu p)(\nu a)
\left(
\bl
    \actor{\idle{()}}{\epsilon}{\epsilon} \\
    \parallel
        \actor
            [\var{ping}]
            {\idle{()}}
            {\roleidx{s}{Pi} \mapsto \handler
                    {Po}
                    {\var{st}}
                    {\mkwd{Pong} \mapsto ()}}
            {\epsilon} \\
    \parallel 
        \actor
            [\var{pong}]
            { \sessthread{s}{Po}{\send{Pi}{\mkwd{Pong}}{()}}}
            {\epsilon}
            {\epsilon} \\
    \parallel \ap{p}{\role{Pi}\mapsto \emptyset, \role{Po} \mapsto \emptyset}
    \parallel \qproc{s}{\epsilon}
\el
\right)
}

Finally, the Ponger can send a $\mkwd{Pong}$ back to the Pinger, which
activates the stored handler:

    \smallermath{
        \ceval^+
        (\nu s)(\nu \var{ping})(\nu \var{pong})(\nu p)(\nu a)
\left(
\bl
    \actor{\idle{()}}{\epsilon}{\epsilon}
    \parallel
        \actor
            [\var{ping}]
            {\sessthread{s}{Pi}{()}}
            {\epsilon}
            {\epsilon} 
    \parallel 
        \actor
            [\var{pong}]
            { \sessthread{s}{Po}{()}}
            {\epsilon}
            {\epsilon} \\
    \parallel \ap{p}{\role{Pi} \mapsto \emptyset, \role{Po} \mapsto \emptyset}
    \parallel \qproc{s}{\epsilon}
\el
\right)
}
    Both actors have now finished the session and therefore revert to
    being idle (\textsc{E-Reset}).\end{example}

 \section{Metatheory}\label{sec:metatheory}
In order to prove properties about \langname, we define an
extrinsic~\cite{Reynolds00} type system for \langname configurations. This is
only used for our proofs; we do not need to implement it in a typechecker.

Following~\citet{ScalasY19} we begin by showing a type semantics for sets of
local types. Using this semantics we can ensure that collections of local types
are \emph{compliant}, meaning that communicated messages are always compatible
and that communication is deadlock-free, and use this to prove type
preservation, progress, and global progress for \langname configurations.

\begin{figure}

{\footnotesize
\header{Runtime types, environments, and labels}
    \[
        \begin{array}{lrcl}
        \text{Polarised initialisation tokens} & \inittokpol & ::= & \inittokpos \midspace \inittokneg \\
\text{Queue types} & \qty & ::= & \epsilon \midspace \qentry{p}{q}{\ell}{\tya} \cdot \qty \\
        \text{Runtime type environments} & \rtenv & ::= & \cdot \midspace
        \rtenv, \aname \midspace \rtenv, \apname \midspace
        \rtenv, \inittokpol : S \midspace \rtenv, \roleidx{s}{p}: S \midspace \rtenv, s : \qty
        \\
        \text{Labels} & \ltslbl & ::= & \lblsyncsend{s}{\role{p}}{\role{q}}{\ell} \midspace
        \lblsyncrecv{s}{\role{p}}{\role{q}}{\ell} \midspace
        \lblsyncend{s}{\prole}
        \end{array}
    \]

\headersig{Structural congruence (queue types)}{$\qty \equiv \qty'$}
\begin{mathpar}
    \inferrule
    { \role{p}_1 \ne \role{p}_2 \vee \role{q}_1 \ne \role{q}_2 }
    { Q_1 \cdot \qentryy{\role{p}_1}{\role{q}_1}{\ell_1}{\tya_1} \cdot \qentryy{\role{p}_2}{\role{q}_2}{\ell_2}{\tya_2} \cdot Q_2
        \equiv
        Q_1 \cdot \qentryy{\role{p}_2}{\role{q}_2}{\ell_2}{\tya_2} \cdot \qentryy{\role{p}_1}{\role{q}_1}{\ell_1}{\tya_1} \cdot Q_2
    }
\end{mathpar}

\headersig{Runtime type environment reduction}{$\rtenv\lbleval{\ltslbl} \rtenv'$}
\[
    \begin{array}{lrcl}
        \textsc{Lbl-Send} &
        \rtenv, \roleidx{s}{\prole} : \localselect{q}{\msg{\ell_i}{\tya_i}.\sta_i}_{i \in I}, s : \qty
        & \lbleval{\lblsyncsend{s}{p}{q}{\ell_j}} &
        \rtenv, \roleidx{s}{\prole} : \sta_j, s : Q \cdot \qentry{p}{q}{\ell_j}{\tya_j}  \quad (j \in I)
\\
\textsc{Lbl-Recv} &
        \rtenv, \roleidx{s}{p} : \localoffer{q}{\msg{\ell_i}{\tya_i}.\sta_i}_{i \in I},
         s : \qentry{q}{p}{\ell_j}{\tya_j} \cdot Q
                               &
         \lbleval{\lblsyncrecv{s}{q}{p}{\ell_j}}
                               &
        \rtenv, \roleidx{s}{p} : \sta_j, s : Q
        \quad (\text{if } j \in I)
\\
\textsc{Lbl-End} &
        \rtenv, \roleidx{s}{\prole} : \localend & \lbleval{\lblsyncend{s}{\prole}} & \rtenv
\\
\textsc{Lbl-Rec} &
        \rtenv, \roleidx{s}{\prole} : \recty{X}{S} & \lbleval{\ltslbl} & \rtenv' \quad (\text{if }
        \rtenv, \roleidx{s}{\prole} : S \{ \recty{X}{S} / X \} \lbleval{\ltslbl} \rtenv')
    \end{array}
\]
}
    \caption{Labelled transition system on runtime type environments}
    \label{fig:lts}
\end{figure}

\paragraph*{Relations.} We write $\mathcal{R}^?$, $\mathcal{R}^+$, and
$\mathcal{R}^*$ for the reflexive, transitive, and reflexive-transitive closures
of a relation $\mathcal{R}$ respectively. We write $\mathcal{R}_1 \mathcal{R}_2$
for the composition of relations $\mathcal{R}_1$ and $\mathcal{R}_2$.

\paragraph*{Runtime types and environments.}
Runtime environments are used to type configurations and to define behavioural
properties on sets of local types.
Unlike type environments $\tyenv$, runtime type environments $\rtenv$ are
\emph{linear} to ensure safe use of session endpoints, and also to
ensure that there is precisely one instance of each actor and access point.
Runtime type environments can contain actor names $a$; access point names $\apname$;
\emph{polarised} initialisation tokens $\inittokpol : S$ (since each
initialisation token is used twice: once in the access point and one inside an
actor's initialisation state); session endpoints
$\roleidx{s}{\prole} : S$; and finally session queue types $s : \qty$. Queue
types mirror the structure of queue entries and consist of a series of triples
$\qentry{\prole}{\qrole}{\ell}{\tya}$. We include structural congruence on queue
types to match structural congruence on queues, and extend this to runtime
environments.

\paragraph*{Labelled transition system on environments.}
Figure~\ref{fig:lts} shows the LTS on runtime type environments.
The \textsc{Lbl-Send} reduction gives the behaviour of an output session type
interacting with a queue: supposing we send a message with some label
$\ell_j$ from $\prole$ to $\qrole$, we advance the session type for
$\prole$ to the continuation $\sta_j$ and add the message to the end of the
queue. The \textsc{Lbl-Recv} rule handles receiving and works similarly, instead
\emph{consuming} the message from the queue. Rule \textsc{Lbl-End} allows
us to discard a session endpoint from the environment if it does not support
any further communication, and \textsc{Lbl-Rec} allows reduction of
recursive session types by considering their unrolling.
We write $\rtenv \equivsynceval \rtenv'$ if $\rtenv \equiv \lbleval{\ltslbl}
\equiv \rtenv'$ for some synchronisation label $\ltslbl$, and conversely write
$\rtenv \not\equivsynceval{}$ if there exists no $\rtenv'$ such that $\rtenv
\equivsynceval \rtenv'$.

\paragraph*{Protocol Properties}

In order to prove type preservation and progress properties on \langname
configurations, we need to ensure each protocol in the system is
\emph{compliant}, meaning that it is safe and deadlock-free.
\emph{Safety} is the minimum we can expect from a protocol in order for us to
prove type preservation: a safe runtime type environment ensures that
communication does not introduce type errors. Intuitively, safety ensures that a
message received from a queue is of the expected type, thereby ruling out
communication mismatches; safety properties must also hold under unfoldings of
recursive session types and safety must be preserved by environment reduction.
Deadlock-freedom on runtime type environments requires that every message that
is sent in a protocol can eventually be received, and that a participant will never
wait for a message that will never arrive.

\begin{definition}[Compliance]
    A runtime environment $\rtenv$ is \emph{compliant}, written
    $\compliant{\rtenv}$, if it is \emph{safe} and \emph{deadlock-free}:

    \begin{description}[leftmargin=0.1em]
        \item[Safe]
            An environment $\rtenv$ is \emph{safe}, written $\safe{\rtenv}$, if:
            \begin{itemize}[leftmargin=*]
                \item $\rtenv = \rtenv', \roleidx{s}{\prole} :
                    \localoffer{q}{\msg{\ell_i}{\tya_i}.\sta_i}_{i \in I}, s : \qty$
                    with $\qty \equiv \qentry{\qrole}{\prole}{\ell_j}{\tyb_j} \cdot \qty'$
                    implies $j \in I$ and $\tyb_j = \tya_j$; and
                \item $\rtenv = \rtenv', \roleidx{s}{\prole} : \recty{\recvar}{\sta}$
                    implies $\safe{\rtenv', \roleidx{s}{\prole} : \sta \{
                    \recty{\recvar}{\sta} / \recvar \}}$; and
                \item $\safe{\rtenv}$ and $\rtenv \equivsynceval \rtenv'$ implies
                    $\safe{\rtenv'}$.
            \end{itemize}
        \item[Deadlock-free]
            An environment $\rtenv$ is \emph{deadlock-free}, written
            $\df{\rtenv}$, if
            ${\rtenv} \mathop{\equivsynceval{}^{\!\!\!\!*}} {\rtenv'}
            {\not\equivsynceval{}}\!\!$ implies
                $\rtenv' = s : \emptyq$.
    \end{description}
A \emph{protocol} $\set{\prole_i : \sta_i}_{i \in 1..n}$ is compliant if
    $\compliant{\roleidxx{s}{\prole_1} : \sta_1, \ldots, \roleidxx{s}{\prole_n}
    : \sta_n}$ for an arbitrary $s$.
\end{definition}

Checking compliance for an \emph{asynchronous} protocol is undecidable in
general~\cite{ScalasY19}, but various sound and tractable mechanisms can ensure
it in practice. For example, syntactic projections from global types produce
safe and deadlock-free sets of local types~\cite{ScalasY19}.  Furthermore,
multiparty compatibility~\cite{DenielouY13} allows safety to be verified by
bounded model checking; this is the core approach implemented in
Scribble~\cite{HuY17}, used by our implementation.

We have therefore designed our type system to be agnostic to any specific
implementation method for validating compliance, as common in recent MPST language
design papers~(e.g.,~\cite{HarveyFDG21, LagaillardieNY22}).

\begin{figure}
    {\footnotesize 
    \headersig{Configuration typing rules}{$\cseq[]{\tyenv}{\rtenv}{\config{C}}$}
    \begin{mathpar}
    \inferrule
    [T-APName]
    { \cseq
        {\tyenv, \apname : \apty{(\prole_i : \sta_i)_{i \in I}}}
        {\rtenv, \apname}
        {\config{C}}
    }
    { \cseq{\tyenv}{\rtenv}{(\nu \apname)\config{C}} }
\quad
    \inferrule
    [T-InitName]
    { \cseq{\tyenv}{\rtenv, \inittokpos : S, \inittokneg : S}{\config{C}} }
    { \cseq{\tyenv}{\rtenv}{(\nu \inittok) \config{C}} }
\quad
    \inferrule
    [T-SessionName]
    { \rtenv' = \{ \roleidxx{s}{\role{p}_i} : S_{\role{p}_{i}} \}_{i \in 1..n},
        s : \qty \\\\
        \compliant{\rtenv'} \\\\
        s \not\in \snames{\rtenv} \\
      \cseq{\tyenv}{\rtenv, \rtenv'}{\config{C}} \\
    }
    { \cseq{\tyenv}{\rtenv}{(\nu s)\config{C}} }
\quad
    \inferrule
    [T-ActorName]
    { \cseq{\tyenv}{\rtenv, a}{\config{C}}}
    { \cseq{\tyenv}{\rtenv}{(\nu \aname)\config{C}}}

    \inferrule
    [T-Par]
    {
        \cseq{\tyenv}{\rtenv_1}{\config{C}} \\
        \cseq{\tyenv}{\rtenv_2}{\config{D}}
    }
    {\cseq{\tyenv}{\rtenv_1, \rtenv_2}{\config{C} \parallel \config{D}}}

    \inferrule
    [T-AP]
    {
        \apname : \apty{(\prole_i: \sta_i)_{i \in I}} \in \tyenv \\\\
        \apseq{\set{\prole_i : \sta_i}_{i \in I}}{\rtenv}{\apstate} \\
        \compliant{\set{\prole_i: S_i}_{i \in I}} \\
    }
    { \cseq
        {\tyenv}
        {\rtenv, \apname}
        {\ap{\apname}{\apstate}}
    }

    \inferrule
    [T-Actor]
    {
        \threadseq{\tyenv}{\rtenv_1}{\tya}{\threadt} \\\\
        \hstateseq{\tyenv}{\rtenv_2}{\tya}{\hstate} \\
        \istateseq{\tyenv}{\rtenv_3}{\tya}{\istate} \\
    }
    { \cseq
        { \tyenv }
        { \rtenv_1, \rtenv_2, \rtenv_3, \aname }
        { \actor
            {\threadt}
            {\hstate}
            {\istate}
        }
    }

    \inferrule*
    [left=T-EmptyQueue]
    { } { \cseq{\tyenv}{s : \epsilon}{\qproc{s}{\epsilon}} }

    \inferrule*
    [left=T-ConsQueue]
    { \vseq{\tyenv}{V}{\tya} \\ \cseq{\tyenv}{s : Q}{\qproc{s}{\sigma}} }
    { \cseq{\tyenv}{s : (\qentry{p}{q}{\ell}{\tya} \cdot Q)}{\qproc{s}{\qentry{p}{q}{\ell}{\vala} \cdot \sigma}} }
\end{mathpar}
\vspace{1em}

\begin{minipage}{0.49\textwidth}
\headersig{Access point typing}{$\apseq{(\prole_i : \sta_i)_{i}}{\rtenv}{\apstate}$}
\begin{mathpar}
    \inferrule*
    [left=TA-Empty]
    { }
    { \apseq{(\prole_i : \sta_i)_{i \in 1..n}}{\cdot}{\cdot}}

    \inferrule
    [TA-Entry]
    { j \in I \\ \apseq{(\prole_i : \sta_i)_{i \in I}}{\rtenv}{\apstate} }
    { \apseq
        {(\prole_i : \sta_i)_{i \in I}}
        {\rtenv, \setseq{\inittokneg : \sta_j}}
        {\apstate[\prole_j \mapsto \setseq{\inittok}]}
    }
\end{mathpar}
\end{minipage}
\hfill
\begin{minipage}{0.49\textwidth}
\headersig{Thread state typing}{$\threadseq{\tyenv}{\rtenv}{\tya}{\threadt}$}
\begin{mathpar}
    \inferrule*
    [left=TT-Idle]
    { \vseq{\tyenv}{\vala}{\tya} }
    { \threadseq{\tyenv}{\cdot}{\tya}{\idle{\vala}}}

    \inferrule
    [TT-Sess]
    {
        \mtseqst{\tyenv}{\tya}{S}{M}{\tya}{\localend}
    }
    { \threadseq{\tyenv}{\roleidx{s}{\prole} : S}{\tya}{\sessthread{s}{p}{M}}}
\quad
\inferrule
    [TT-NoSess]
    { \mtseqst{\tyenv}{\tya}{\localend}{M}{\tya}{\localend} }
    { \threadseq{\tyenv}{\cdot}{\tya}{M} }
\end{mathpar}
\end{minipage}
\vspace{1em}

\begin{minipage}{0.425\textwidth}
    \headersig{Handler state typing}{$\hstateseq{\tyenv}{\rtenv}{\tya}{\hstate}$}
\begin{mathpar}
    \inferrule
    [TH-Empty]
    { }
    { \hstateseq{\tyenv}{\cdot}{\tya}{\epsilon} }
\quad
    \inferrule
    [TH-Handler]
    {
        \vseq{\tyenv}{\vala}{\handlerty{\stin}{\tya}} \\
\hstateseq{\tyenv}{\rtenv}{\tya}{\hstate}
    }
    { \hstateseq
       {\tyenv}
       {\rtenv, \roleidx{s}{\prole}: \stin}
       {\tya}
       {\hstate[\storedhandler{s}{\prole}{\vala}]}
   }
\end{mathpar}
\end{minipage}
\hfill
\begin{minipage}{0.525\textwidth}
    \headersig{Initialisation state typing}{$\istateseq{\tyenv}{\rtenv}{\tya}{\istate}$}
\begin{mathpar}
    \inferrule
    [TI-Empty]
    { }
    { \istateseq{\tyenv}{\cdot}{\tya}{\epsilon} }
\quad
\inferrule
    [TI-Callback]
    {
\valseq{\tyenv}{\vala}{\sttyfun{\tya}{\tya}{\sta}{\localend}{\tya}} \\
        \istateseq{\tyenv}{\rtenv}{\tya}{\istate}
    }
    { \istateseq{\tyenv}{\rtenv, \inittokpos : \sta}{\tya}{\istate[\maptwo{\inittok}{\vala}]} }
\end{mathpar}
\end{minipage}
\[
    \begin{array}{rcl}
        \snames{\rtenv} & = & \set{ s \mid s : Q \in \rtenv \vee
            \exists \prole . (\roleidx{s}{\prole} \in \dom{\rtenv}) }
    \end{array}
\]
    }
    \vspace{-1.5em}
    \caption{Typing of Configurations}
    \label{fig:rt-typing}
    \vspace{-1.4em}
\end{figure}

\subsection{Configuration Typing}

Figure~\ref{fig:rt-typing} shows the typing rules for \langname configurations.
The configuration typing judgement $\cseq{\tyenv}{\rtenv}{\config{C}}$ can be
read, ``under type environment $\tyenv$ and runtime type environment $\rtenv$,
configuration $\config{C}$ is well typed''. 
We have three rules for name restrictions: read bottom-up, \textsc{T-APName}
adds $\apname$ to both environments, and
rule \textsc{T-InitName} adds tokens of both polarities to the runtime type
environment.
Rule \textsc{T-SessionName} is key to the generalised multiparty session typing
approach introduced by~\citet{ScalasY19}: to type a name restriction
$(\nu s) \config{C}$,
the type environment $\rtenv'$
consists of a set of session endpoints $\set{\roleidxx{s}{\prole_i}}_i$
with session types $\sta_{\prole_i}$, along with a session queue $s : \qty$.
We call $\rtenv'$ the \emph{session environment for $s$},
which must be compliant.
The condition $s \not\in \snames{\rtenv}$ ensures that no other endpoint or
queue with session name $s$ may be present in the initial environment.

Rule \textsc{T-Par} types two parallel subconfigurations under disjoint runtime
environments. Rule \textsc{T-AP} types an access point: it requires that the
access point reference is included in $\tyenv$ and through the auxiliary
judgement $\apseq{(\prole_i : \sta_i)_i}{\rtenv}{\apstate}$ ensures that each
initialisation token in the access point has a compatible type. We also
require that the protocol supported by the access point is compliant.

Rule \textsc{T-Actor} types an actor $\actor{\threadt}{\hstate}{\istate}$
using three auxiliary judgements.
The thread state typing judgement
$\threadseq{\tyenv}{\rtenv}{\tyc}{\threadt}$ 
ensures that a thread either performs all pending communication
actions, or it suspends.
The handler typing judgement $\hstateseq{\tyenv}{\rtenv}{\tyc}{\hstate}$
ensures that the stored handlers match the types in the runtime environments,
and the initialisation state typing judgement 
$\istateseq{\tyenv}{\rtenv}{\tyc}{\istate}$ ensures that all initialisation
callbacks match the session type of the initialisation token.
Finally, \textsc{T-EmptyQueue} and \textsc{T-ConsQueue} ensure that queued
messages match the queue type.

\subsection{Properties}\label{sec:metatheory:properties}

With configuration typing defined, we can begin to describe the properties
enjoyed by $\langname$.

\subsubsection{Preservation}

Typing is preserved by reduction; consequently we know that communication actions
must match those specified by the session type. Proofs can be found in the extended version.

\begin{restatable}[Preservation]{theorem}{preservation}\label{thm:preservation}
    Typability is preserved by structural congruence and reduction.\hfill
    \begin{itemize}
        \item[($\equiv$)] If $\cseq{\tyenv}{\rtenv}{\config{C}}$ and $\config{C} \equiv \config{D}$
            then there exists some $\rtenv' \equiv \rtenv$ such that $\cseq{\tyenv}{\rtenv'}{\config{D}}$.
        \item[($\rightarrow$)] If $\cseq{\tyenv}{\rtenv}{\config{C}}$ with
            $\safe{\rtenv}$ and $\config{C} {\rightarrow} \config{D}$, then
            there exists some $\rtenv'$ such that
            $\rtenv \equivsynceval^?  \rtenv'$
            where $\safe{\rtenv'}$ and $\cseq{\tyenv}{\rtenv'}{\config{D}}$.
    \end{itemize}
\end{restatable}

\begin{remark}[Session Fidelity]
    Traditionally, \emph{session fidelity} is presented as a property that all
    communication in a system conforms to its associated session
    type, i.e., that if a process performs a communication action then there
    is a corresponding (meta-theoretical) type reduction~\cite{HondaYC08, CoppoDYP16}.
Fidelity is often an implicit corollary of type preservation in works on
    functional session types~(e.g., \cite{GayV10, LindleyM15, FowlerLMD19,
    HarveyFDG21}).
Alternatively, session fidelity in~\cite{ScalasY19} (and derived works)
    refer to session fidelity as the property that at least one typing context
    reduction can be reflected by a process.  
We follow the former definition and account for fidelity through our
    preservation theorem.
\end{remark}

\subsubsection{Progress}
In general, just because two protocols are individually deadlock-free
\emph{does not} mean that the system as a whole is deadlock-free, due to the
possibility of inter-session deadlocks. For example, consider the following two
trivially deadlock-free protocols:

\raisebox{0.2em}{
    \smash{
    \smallmath{
    \begin{minipage}{0.45\textwidth}
        \[
            \bl
            \{
            \prole: \localoffer{\qrole}{\msg{\ell_1}{\one} \then \localend }, \:\:
        \qrole: \localselect{\prole}{\msg{\ell_1}{\one} \then \localend} \}
            \el
        \]
    \end{minipage}
    \hfill
    \begin{minipage}{0.45\textwidth}
        \[
            \bl
            \{
            \rrole: \localoffer{\srole}{\msg{\ell_2}{\one} \then \localend }, \:\:
            \srole: \localselect{\rrole}{\msg{\ell_2}{\one} \then \localend}
            \}
            \el
        \]
    \end{minipage}
    }
}
}

Even with an asynchronous semantics, a standard multiparty process calculus
would admit the following deadlocking process, since each send is blocked by a
receive:

    \raisebox{0em}{
        \smash{
            {\footnotesize
                \begin{mathpar}
                    s_1[\prole][\qrole] \& \msg{\ell_1}{x} \then s_2[\rrole][\srole] \oplus
                    \msg{\ell_2}{y} \then \mathbf{0}
                    \;
                    \parallel
                    \;
                    s_2[\srole][\rrole] \& \msg{\ell_2}{a} \then s_1[\qrole][\prole] \oplus
                    \msg{\ell_1}{b} \then \mathbf{0}
                \end{mathpar}
            }
        }
    }

There are various approaches to ruling out inter-session deadlocks: some
approaches restrict each subprocess to only play a single role in a single
session~(e.g.,~\cite{ScalasY19}); this would rule out the above example but is
too restrictive for our setting. Other approaches (e.g.,~\cite{CoppoDYP16})
overlay additional interaction type systems to rule out inter-process deadlocks,
again at the cost of expressiveness and type system complexity. Finally,
logically-inspired approaches to multiparty session typing
(e.g.,~\cite{CarboneLMSW16}) treat sessions as monolithic processes $(\nu s)(P_1
\parallel \cdots \parallel P_n)$ that mean that such cycles cannot arise.
Programming with such processes requires ``multi-fork'' style session
initiations that combine channel- and process creation, and therefore are
inapplicable to our programming model.

\langname \emph{does not} suffer from inter-process deadlocks because of our
event-driven programming style where
although an actor is involved in many sessions at a time, only one is
\emph{running} at once, and code within handlers does not block.
Since an actor yields and installs a handler whenever it needs to receive a
message, the actor can then schedule any handler that has a waiting message.

To show this intuition formally, we start by classifying a \emph{canonical form}
for configurations.

\begin{definition}[Canonical form]
    A configuration $\config{C}$ is in \emph{canonical form} if it can be written:

    \footnotesize{
        \raisebox{-0.5em}{
        \smash{
            \begin{mathpar}
            (\nu \tilde{\inittok})
            (\nu \apname_{i \in 1..l})(\nu \sessname_{j \in 1..m})
            (\nu \aname_{k \in 1..n})(
            \ap{\apname_i}{\apstate_i}_{i \in 1..l} \parallel 
            (\qproc{s_j}{\qcontents_j})_{j \in 1..m} \parallel
            \actor[a_k]{\threadt_k}{\hstate_k}{\istate_k}_{k \in 1..n}
            )
            \end{mathpar}
        }
    }
    }
\end{definition}

Every well typed configuration can be written in canonical form; the result
follows from the structural congruence rules and Theorem~\ref{thm:preservation}.

Compliance requires the session \emph{types} in every session to satisfy
progress. Due to our event-driven design, the property transfers to
\emph{configurations}: a non-reducing closed configuration cannot be blocked on
any session communication and so cannot contain any sessions.

\emph{Progress} states that since (by compliance) all protocols are
deadlock-free, a configuration can either reduce, or it contains no
sessions and no further sessions can be established.
\begin{restatable}[Progress]{theorem}{progress}\label{thm:progress}
If $\cseq{\cdot}{\cdot}{\config{C}}$, then either there
exists some $\config{D}$ such that $\config{C} \ceval \config{D}$, or
$\config{C}$ is structurally congruent to the following canonical form:

\smash{
{\small
\begin{mathpar}
    (\nu \tilde{\inittok})(\nu \apname_{i \in 1..m})(\nu \aname_{j \in 1..n})
    (
    \ap{\apname_1}{\apstate_1}_{i \in 1..m} \parallel
    \actor[a_j]{\idle{\vala_j}}{\epsilon}{\istate_j}_{j \in 1..n}
    )
\end{mathpar}
}
}
\end{restatable}

\subsubsection{Global Progress}
\emph{Global progress} extends Theorem~\ref{thm:progress} to show that if all event
handlers terminate (a usual assumption in event-driven systems), then
communication is eventually possible for any session with pending communication
actions.
The technical development shows that thread reduction in one actor does not
block another (Lemma~\ref{lem:gp:independence}); that if all threads terminate,
the configuration can reach an idle state
(Corollary~\ref{cor:gp:reduction-idle}); and that from any idle configuration,
we can invoke a handler for any session with pending communication actions
(Lemma~\ref{lem:gp:idle-reduces}).

We begin by defining configuration contexts
${\small \confctx ::= [~] \midspace (\nu \nma)\confctx \midspace \confctx \parallel
\config{C}} \midspace \config{C} \parallel \confctx$ that allow us to focus on a subconfiguration.
We say that $\actor{\tma}{\hstate}{\istate}$ is an \emph{actor subconfiguration}
of a configuration $\config{C}$ if $\config{C} =
\confctx[\actor{\tma}{\hstate}{\istate}]$ for some configuration context
$\confctx$.

\begin{definition}[Ongoing environment / session]
We say that a runtime type environment $\rtenv$ is \emph{ongoing}, written
$\envactive{\rtenv}$, if it contains at least one entry of the form
$\roleidx{s}{\prole} : \sta$ where $\sta \ne \localend$.
A \emph{session} is ongoing if its session environment is ongoing.
\end{definition}

Given a derivation $\cseq{\tyenv}{\rtenv}{\config{C}}$
we can annotate each name restriction $(\nu s) \mathcal{D}$ appearing in
$\mathcal{C}$ with its session environment (i.e., to $(\nu s: \rtenv')
\config{D}$ if session $s$ has session environment $\rtenv'$).
We define $\activesessions{\config{C}}$ as the set of session names in
$\config{C}$ with ongoing session environments, with the remaining cases defined
recursively:

\vspace{-1em}
    {\scriptsize
    \begin{mathpar}
        \activesessions{(\nu s: \rtenv)\config{C}} =
        {
        \begin{cases}
            \set{s} \cup \activesessions{\config{C}} & \text{if } \envactive{\rtenv} \\
            \activesessions{\config{C}} & \text{otherwise}
        \end{cases}
        }
    \end{mathpar}
    }

We now classify \emph{thread-terminating} actors: an actor is thread-terminating
if it will eventually either suspend with a handler or fully evaluate to a
value. A \emph{thread reduction} for an actor $a$ is a configuration reduction
that affects the ongoing thread of $a$.

\begin{definition}[Thread Reduction]
    A reduction
    $\config{C} \ceval \config{D}$ 
    where $\config{C} = \confctx_1[\actor{\threadctx[\tma]}{\hstate}{\istate}]$
    and $\config{D} = \confctx_2[\actor{\threadctx[\tmb]}{\hstate'}{\istate'}]$ is a
    \emph{thread reduction for $a$} if 
    $\actor{\threadctx[\tma]}{\hstate}{\istate}$ is a subconfiguration of the
    fired redex of $\config{C}$, and
    $\actor{\threadctx[\tmb]}{\hstate'}{\istate'}$ is a
    subconfiguration of its contractum.
\end{definition}

\begin{definition}[Thread-Terminating]
    An actor subconfiguration $\actor{\threadt}{\hstate}{\istate}$ of a
    configuration $\config{C} = \confctx[\actor{\threadt}{\hstate}{\istate}]$
    is \emph{thread-terminating} if either $\threadt = \idle{\vala}$ for some
    $\vala$, or
    $\threadt = \threadctx[\tma]$ such that there exists no infinite thread
    reduction for $a$ from $\config{C}$.
\end{definition}

We deliberately design our metatheory to be agnostic to the precise method
used to ensure termination. Concretely, to ensure that actors are always
thread-terminating, we could for example use straightforward
type system restrictions like requiring all callbacks and handlers to be total or use
primitive recursion. We could also use effect-based analyses (e.g.\ those used
for ensuring safe database programming~\cite{LindleyC12}).
We conjecture we could also adapt type systems designed to enforce \emph{fair
termination}~\cite{DBLP:journals/jlap/CicconeDP24, DBLP:conf/ecoop/PadovaniZ25};
we discuss this further in Section~\ref{sec:related}.
The additional power of exceptions described in Section~\ref{sec:failure} would also
allow the smooth integration of run-time termination analyses.
All of the example callbacks and handlers discussed in the paper would preserve
thread-termination.

Next, we show that reduction in one actor will never inhibit reduction in
another. The result follows because all communication is asynchronous and (in
part due to Theorem~\ref{thm:progress}),
given a well-typed configuration,
all constructs occurring in an actor's thread can always reduce immediately.

\begin{restatable}[Independence of Thread Reductions]{lemma}{independence}\label{lem:gp:independence}
    If $\cseq{\cdot}{\cdot}{\config{C}}$ where
    $\config{C} = \confctx_1[\actor{\threadctx[\tma]}{\hstate}{\istate}]$
    and $\config{C} \ceval
    \confctx_2[\actor{\threadctx[\tmb]}{\hstate'}{\istate'}]$ is a thread
    reduction for $a$, then
    for every $\config{D}$ and $\confctx_3$ such that
    $\config{C} \ceval \config{D}$ and
    $\config{D} = \confctx_3[\actor{\threadctx[\tma]}{\hstate}{\istate}]$
    it follows that $\config{D} \ceval
    \confctx_4[\actor{\threadctx[\tmb]}{\hstate'}{\istate'}]$ for some
    $\confctx_4$.
\end{restatable}

\begin{definition}[Idle Configuration]
    An actor subconfiguration $\actor[a]{\threadt}{\hstate}{\istate}$ of a
    configuration $\config{C}$ is \emph{idle} if $\threadt = \idle{\vala}$ for
    some $\vala$.
    Configuration $\config{C}$ is idle if all of its actor subconfigurations
    are idle.
\end{definition}

It follows by typing and from Lemma~\ref{lem:gp:independence} that every
thread-terminating actor subconfiguration of a configuration $\config{C}$
eventually evaluates to either $\effreturn{\vala}$ or
$\suspend{\vala}{\valb}$ and that
(via \textsc{E-Suspend} or \textsc{E-Return}) $\config{C}$ 
reverts to being idle.

\begin{corollary}\label{cor:gp:reduction-idle}
    If $\cseq{\cdot}{\cdot}{\config{C}}$ and $\config{C}$ is
    thread-terminating, then $\config{C} \cevalstar \config{D}$ where
    $\config{D}$ is idle.
\end{corollary}

Finally, due to session typing and compliance, every ongoing session in a
well-typed idle configuration can reduce. The result is as a special case
of Theorem~\ref{thm:progress}.

\begin{restatable}{lemma}{idlereduces}\label{lem:gp:idle-reduces}
If $\cseq{{\cdot}}{{\cdot}\!}{\!\config{C}}$ where $\config{C}$ is idle, then for
    each $s {\in} \activesessions{\config{C}}$,
    $\config{C} {\equiv} (\nu s) \config{D}$ and $\config{D} \cevalann{s}$. 
\end{restatable}

Since (by Theorem~\ref{thm:preservation}) we can always use the structural
congruence rules to hoist a session name restriction to the topmost level,
global progress follows as an immediate corollary.

\begin{restatable}[Global Progress]{corollary}{globalprogress}\label{cor:global-progress}
    If $\cseq{\cdot}{\cdot}{\config{C}}$ where $\config{C}$ is
    thread-terminating, then for every $s \in \activesessions{\config{C}}$,
    $\config{C} \equiv (\nu s)\config{D}$ for some $\config{D}$, and
    $\config{D}\cevalannstar{\tau}\cevalann{s}$.
\end{restatable}

 \section{Failure Handling and Supervision}\label{sec:failure}

\begin{figure}[t]
    {\footnotesize
    \header{Syntax}
    \vspace{-1.2em}
    \begin{minipage}[t]{0.4\textwidth}
    \[
        \begin{array}{lrcl}
        \text{Types} & \tya, \tyb & ::= & \cdots \midspace \typid \\
        \text{Values} & \vala, \valb & ::= & \cdots \midspace \aname \\
        \text{Computations} & \tma, \tmb & ::= & \cdots \midspace
        \suspendexn{\valc}{\vala}{\valb} \\
                            & & \midspace & \monitor{\vala}{\valb} \midspace \raiseexn 
        \end{array}
    \]
\end{minipage}
\hfill
\begin{minipage}[t]{0.5\textwidth}
    \[
        \begin{array}{lrcl}
        \text{Monitored processes} & \monstate & ::= &
        \setseq{\metapair{\aname}{\vala}} \\
\text{Configurations} & \config{C}, \config{D} & ::= & \cdots \midspace
\zapactor{\threadt}{\hstate}{\istate}{\monstate} \\ 
                      & & \midspace & \zap{\aname} \midspace \zap{\roleidx{s}{p}} \midspace \zap{\inittok}
        \end{array}
    \]
\end{minipage}

    \headertwo
        {Modified typing rules for computations}
        {\framebox{$\mtseqst{\tyenv}{\tyc}{\sta}{\tma}{\tya}{\stb}$}}
         \vspace{-1em}

    \begin{mathpar}
        \inferrule
        [T-Spawn]
        { \mtseqst{\tyenv}{\tya}{\localend}{\tma}{\tya}{\localend} }
        { \mtseqst{\tyenv}{\tyc}{\sta}{\spawn{\tma}}{\typid}{\sta} }

        \inferrule
        [T-Suspend]
        {
            \vseq{\tyenv}{\valc}{\handlerty{\stin}{\tyc}} \\
            \vseq{\tyenv}{\vala}{\tyc} \\
            \vseq{\tyenv}{\valb}{\sttyfun{\tyc}{\tyc}{\localend}{\localend}{\tyc}}
}
        { \mtseqst{\tyenv}{\tyc}{\stin}{\suspendexn{\valc}{\vala}{\valb}}{\tya}{\stb} }

        \inferrule*
        [left=T-Monitor]
        {
            \vseq{\tyenv}{\vala}{\typid} \\
            \vseq{\tyenv}{\valb}{\sttyfun{\tyc}{\tyc}{\localend}{\localend}{\tyc}}
        }
        { \mtseqst{\tyenv}{\tyc}{\sta}{\monitor{\vala}{\valb}}{\one}{\sta} }

        \inferrule*
        [left=T-Raise]
        { }
        { \mtseqst{\tyenv}{\tyc}{\sta}{\raiseexn}{\tya}{\stb} }
    \end{mathpar}
    \vspace{-1.5em}

    \headersig{Modified configuration reduction rules}{$\config{C} \cevalann{l} \config{D}$}
         \[
             \begin{array}{ll}
                \textsc{E-React} \qqqqqqquad &
                \zapactor
                { \idle{\valb} }
                {
                \hstate[\storedhandler{s}{p}{\metapair{\handler{q}{\var{st}}{\seq{H}}}{\valc}}] }
                { \istate }
                { \monstate }
                \parallel
                \qproc{s}{\qentry{q}{p}{\ell}{\vala} {\cdot} \qcontents} \\
                                      &
                \qqqqquad \cevalann{s}
                \zapactor
                { \sessthread{s}{\prole}{\tma \{ \vala / x, \valb / \var{st} \}} }
                { \hstate }
                { \istate }
                { \monstate }
                \parallel
                \qproc{s}{\qcontents}
                                      \quad \text{if } (\msg{\ell}{\vala} \mapsto \tma) \in \seq{H}
            \end{array}
        \]
\[
         \begin{array}{lrcl}
             \textsc{E-Spawn} &
             \zapactor{\threadctx[\spawn{\tma}]}{\hstate}{\istate}{\monstate} & \cevaltau &
             (\nu
             b)(\zapactor{\threadctx[\effreturn{b}]}{\hstate}{\istate}{\monstate}
             \parallel \zapactor[b]{\tma}{\epsilon}{\epsilon}{\emptyset})
             \\
             \textsc{E-Suspend} &
             \zapactor
             {\sessthread{s}{p}{\ctxe[\suspendexn{\valc}{\vala}{\valb}]}}
             {\hstate}
             {\istate}
             {\monstate}
                                  & \cevaltau &
             \zapactor
             {\idle{\vala}}
             {\hstate[\roleidx{s}{p} \mapsto \metapair{\valc}{\valb}]}
             {\istate}
             {\monstate}
                                  \\
                          \textsc{E-Monitor} &
                            \zapactor
                            {\threadctx[\monitor{b}{\vala}]}
                            {\hstate}
                            {\istate}
                            {\monstate}
                                & \cevaltau &
                                \zapactor
                                {\threadctx[\effreturn{()}]}
                                {\hstate}
                                {\istate}
                                {\monstate \cup \set{\metapair{b}{\vala}}} \\
\textsc{E-InvokeM} & 
             \zapactor
             {\idle{\valb}}
             {\hstate}
             {\istate}
             {\monstate \cup \set{\metapair{b}{\vala}}} \parallel \zap{b}
                                & \cevaltau &
            \zapactor
            {(\vala \app \valb)}
            {\hstate}
            {\istate}
            {\monstate} \parallel \zap{b}  \\
\textsc{E-Raise} &
             \zapactor
             {\ctxe[\raiseexn]}
             {\hstate}
             {\istate}
             {\monstate} & \cevaltau &
             \zap{\aname} \parallel \zap{\hstate} \parallel \zap{\istate} \\
\textsc{E-RaiseS} & 
             \zapactor
             {\sessthread{s}{p}{\ctxe[\raiseexn]}}
             {\hstate}
             {\istate}
             {\monstate} & \cevaltau &
             \zap{\aname} \parallel \zap{\roleidx{s}{p}} \parallel \zap{\hstate} \parallel \zap{\istate} \\
\textsc{E-CancelMsg} & \qproc{s}{\qentry{p}{q}{\ell}{\vala} \cdot \qcontents} \parallel \zap{\roleidx{s}{q}}
                                  & \cevaltau & \qproc{s}{\qcontents} \parallel \zap{\roleidx{s}{q}} \\
\textsc{E-CancelAP} & (\nu \inittok)(\ap{p}{\apstate[\prole \mapsto \setseq{\inittok'} \cup \set{\inittok}]} \parallel \zap{\inittok})
                                 & \cevaltau &
                                    \ap{p}{\apstate[\prole \mapsto \setseq{\inittok'}]}
         \end{array}
         \]
        \[
        \begin{array}{ll}
             \textsc{E-CancelH}   \qqqqqqquad & \zapactor
                                        {\idle{\valb}}
                                        {\hstate[\roleidx{s}{p} \mapsto
                                        \metapair{\handler{q}{\var{st}}{\seq{H}}}{\vala}] }
                                        {\istate}
                                        {\monstate}
                                        \parallel
                                        \qproc{s}{\qcontents}
                                        \parallel \zap{\roleidx{s}{q}}
                                        \\
                                  & \qqqqquad \cevaltau \:
                                    \zapactor
                                        {(\vala \app \valb)}
                                        {\hstate}
                                        {\istate}
                                        {\monstate}
                                        \parallel
                                        \qproc{s}{\qcontents}
                                        \parallel
                                        \zap{\roleidx{s}{q}} \parallel \zap{\roleidx{s}{p}}
\quad \text{if }
                                  \messages{\qrole}{\prole}{\qcontents} =
                                  \emptyset
         \end{array}
     \]
     \[
         \text{where }
\messages{\prole}{\qrole}{\qcontents} =
       \set{\msg{\ell}{\vala} \mid \qentry{r}{s}{\ell}{\vala} \in \qcontents \wedge
       \prole = \rrole \wedge \qrole = \srole }
 \]

     \begin{minipage}{0.32\textwidth}
     \headersig{Structural congruence}{$\config{C} \equiv \config{D}$}
     \[
     \begin{array}{c}
         (\nu s)(\zap{\roleidxx{s}{\prole_i}}_{i \in 1..n} \parallel \qproc{s}{\epsilon}) \parallel \config{C} \equiv \config{C}
\\
         (\nu a)(\zap{a}) \parallel \config{C} \equiv \config{C}
     \end{array}
     \]
     \end{minipage}
\hfill
     \begin{minipage}{0.62\textwidth}
     \header{Syntactic sugar}
     \[
         \begin{array}{rcll}
             \zap{\hstate} & \defeq & \zap{\roleidxx{s_1}{\prole_1}} \parallel \cdots \parallel \zap{\roleidxx{s_n}{\prole_n}}  & (\text{where } \dom{\hstate} = \set{\roleidxx{s_i}{\prole_i}}_{i \in 1..n}) \\
             \zap{\istate} & \defeq & \zap{\inittok_1} \parallel \cdots \parallel \zap{\inittok_n} & (\text{where } \dom{\istate} = \set{\inittok_{i}}_{i \in 1..n})\\
         \end{array}
     \]
    \end{minipage}
    }
    \vspace{-1em}
    \caption{\langnamezap: Modified syntax and reduction rules}
    \vspace{-1em}
 \label{fig:extensions:supervision-failure-1}
\end{figure}

Actor languages support the \emph{let-it-crash} philosophy, where actors fail fast
and are restarted by supervisors.  A crashed actor cannot send messages, so we
need a mechanism to prevent sessions from getting stuck.  We use the
\emph{affine sessions} approach~\cite{MostrousV18, HarveyFDG21,
LagaillardieNY22, FowlerLMD19}, where participants can be marked as cancelled;
receiving from a cancelled participant with an empty queue raises an exception,
triggering a crash and propagating the failure.
Figure~\ref{fig:extensions:supervision-failure-1} shows the additional syntax,
typing rules, and reduction rules needed for supervision and cascading failure;
we call this extension \langnamezap.
We make actors \emph{addressable}, so $\calcwd{spawn}$ returns a process
identifier (PID) of type $\typid$.  The $\monitor{\vala}{\valb}$ construct
installs a callback function $\valb$ to be evaluated should the actor referred
to by $\vala$ crash. 
The $\raiseexn$ construct signifies a user-level error has
occurred, for example
$\ite{\mkwd{fileExists}(\var{path}))}{\mkwd{processFile}(\var{path})}{\raiseexn}$.
Like $\calcwd{suspend}$, $\raiseexn$ can be given an arbitrary type and
post-condition since it does not return.
We also modify $\calcwd{suspend}$ to take an additional callback to be run if
the sender fails; a sensible piece of syntactic sugar is $\suspend{\vala}{\valb}
\defeq \suspendexn{\vala}{\valb}{(\fun{\var{st}}{\raiseexn})}$ to propagate the
failure.

We can make our $\mkwd{shop}$ actor robust by using a \mkwd{shopSup}
actor that restarts it upon failure:

\smallmath{
        \mkwd{shopSup} \defeq \lambda \var{custAP} .
        {\monitor
            {(\spawn{\mkwd{shop}} \app (\var{custAP}, \mkwd{initialStock}))}
            {(\mkwd{shopSup} \app \var{custAP})}}
}
The $\mkwd{shopSup}$ actor spawns a $\mkwd{shop}$ actor and monitors the
resulting PID. Any failure of the $\mkwd{shop}$ actor will be detected by the
\mkwd{shopSup} which will restart the actor and monitor it again. The restarted
\mkwd{shop} actor will re-register with the access points and can then take
part in subsequent sessions.

\paragraph*{Configurations.}
To capture the additional runtime behaviour we modify the actor configuration to
$\zapactor{\threadt}{\hstate}{\istate}{\monstate}$, where $\monstate$ pairs
monitored PIDs with callbacks to be evaluated should the actor crash. We also
introduce three kinds of ``zapper thread'', $\zap{a}$, $\zap{\roleidx{s}{p}}$,
$\zap{\inittok}$ to indicate the cancellation of an actor, role, and
initialisation token respectively.

\paragraph*{Reduction rules by example.}
Consider the supervised Shop example after the
\role{Customer} has sent a \mkwd{Checkout} request and is awaiting a response.
Instead of suspending, \role{Shop} raises an exception.
This scenario
can be represented by the following configuration, where $\var{shop}$,
$\var{cust}$, and $\var{pp}$ are
actors playing the \role{Shop}, \role{Customer}, and \role{PaymentProcessor} in
session $\var{s}$, and $\var{sup}$ is monitoring $\var{shop}$:

\smallermath{
    (\nu\var{sup})(\nu\var{shop})(\nu\var{cust})(\nu\var{pp})(\nu\var{s})
    \left(
        \bl 
        \zapactor
            [\var{shop}]
            { \sessthread{s}{Shop}{\raiseexn}}
            { \epsilon }
            { \epsilon }
            { \epsilon }
            \\
        \parallel
        \zapactor
            [\var{cust}]
            { \idle{()} }
            { \storedhandler{s}{Customer}{\metapair{\mkwd{checkoutHandler}}{\raiseexn}}}
            { \epsilon}
            { \epsilon}
            \\
            \parallel
        \zapactor
            [\var{pp}]
            { \idle{()} }
            { \storedhandler{s}{PaymentProcessor}{\metapair{\mkwd{buyHandler}}{\raiseexn}}}
            { \epsilon}
            { \epsilon}
            \\ 
            \parallel \qproc{s}{\qentry{Customer}{Shop}{\mkwd{checkout}}{([123], 510)}}
            \\
        \parallel
        \zapactor
            [\var{sup}]
            {\idle{()}}
            {\epsilon}
            {\epsilon}
            {\metapair{\var{shop}}{\mkwd{shopSup} \app \var{cAP}}}
        \el
    \right)
}

For brevity we shorten $\role{Shop}$, $\role{Customer}$, and
$\role{PaymentProcessor}$ to $\role{S}$, $\role{C}$, and $\role{PP}$ respectively.
We let context
{\small
    $
    \confctx =
    (\nu\var{sup})(\nu\var{shop})(\nu\var{cust})(\nu\var{pp})(\nu\var{s})([~]
    \parallel
    \zapactor
            [\var{sup}]
            {\idle{()}}
            {\epsilon}
            {\epsilon}
            {\metapair{\var{shop}}{\mkwd{shopSup} \app \var{cAP}}})
$.}

Since the $\var{shop}$ actor is playing role $\roleidx{s}{S}$ and
raising an exception, by \textsc{E-RaiseS} the actor is replaced with zapper
threads $\zap{\var{shop}}$ and $\zap{\roleidx{s}{S}}$.

\smallermath{
    \hspace{-2.2em}
    \begin{tabular}{rcl}
        $
    \confctx
    \left[
\!\!
        \bl 
        \zapactor
            [\var{shop}]
            { \sessthread{s}{S}{\raiseexn}}
            { \epsilon }
            { \epsilon }
            { \epsilon }
            \\
        \parallel
        \zapactor
            [\var{cust}]
            { \idle{()} }
            { \storedhandler{s}{C}{\metapair{\mkwd{checkoutHandler}}{\raiseexn}}}
            { \epsilon}
            { \epsilon}
            \\
            \parallel
        \zapactor
            [\var{pp}]
            { \idle{()} }
            { \storedhandler{s}{PP}{\metapair{\mkwd{buyHandler}}{\raiseexn}}}
            { \epsilon}
            { \epsilon}
            \\ 
            \parallel \qproc{s}{\qentry{C}{S}{\mkwd{checkout}}{([123], 510)}}
        \el
\!\!
        \right]
        $
& \!\!\!\!\!\!${\ceval}$\!\!\!\!\!\! &
$
    \confctx
    \left[
\!\!
        \bl 
        \zap{\var{shop}}
        \parallel \zap{\roleidx{s}{S}}
        \\
        \parallel
        \zapactor
            [\var{cust}]
            { \idle{()} }
            { \storedhandler{s}{C}{\metapair{\mkwd{checkoutHandler}}{\raiseexn}}}
            { \epsilon}
            { \epsilon}
            \\
            \parallel
        \zapactor
            [\var{pp}]
            { \idle{()} }
            { \storedhandler{s}{PP}{\metapair{\mkwd{buyHandler}}{\raiseexn}}}
            { \epsilon}
            { \epsilon}
            \\ 
            \parallel \qproc{s}{\qentry{C}{S}{\mkwd{checkout}}{([123], 510)}}
        \el
        \!\!
        \right]
        $
\end{tabular}
}

Next, since $\roleidx{s}{S}$ has been cancelled, the $\mkwd{checkout}$ message
can never be received and so is removed from the queue (\textsc{E-CancelMsg}).
Similarly since both $\role{C}$ and $\role{PP}$ are waiting for messages from
cancelled role \role{S}, they both evaluate their failure computations,
$\raiseexn$ (\textsc{E-CancelH}).
In turn this results in the cancellation of the $\var{cust}$ and $\var{pp}$
actors, and the $\roleidx{s}{C}$ and $\roleidx{s}{PP}$ endpoints
(\textsc{E-RaiseS}).

\smallermath{
    \hspace{-1em}
\begin{tabular}{rrcl}
    $\ceval^+$ & 
$
\confctx
    \left[
        \bl 
        \zap{\var{shop}}
        \parallel \zap{\roleidx{s}{S}}
        \\
        \parallel
        \zapactor
            [\var{cust}]
            { \idle{()} }
            { \sessthread{s}{C}{\raiseexn}}
            { \epsilon}
            { \epsilon}
            \\
            \parallel
        \zapactor
            [\var{pp}]
            { \idle{()} }
            { \sessthread{s}{PP}{\raiseexn}}
            { \epsilon}
            { \epsilon}
            \\ 
            \parallel \qproc{s}{\epsilon}
        \el
        \right]
        $
 & $\ceval^+$ & 
$
\confctx
    \left[
        \bl 
        \zap{\var{shop}}
        \parallel
        \zap{\roleidx{s}{S}}
        \parallel
        \zap{\var{cust}}
        \parallel \zap{\roleidx{s}{C}}
        \parallel
        \zap{\var{pp}}
        \parallel \zap{\roleidx{s}{PP}}
        \parallel \qproc{s}{\epsilon}
        \el
        \right]
        $
\end{tabular}
}

At this point the session has failed and can be garbage collected, leaving the
supervisor actor and the zapper thread for $\var{shop}$. Since the supervisor was
monitoring $\var{shop}$, which has crashed, the monitor callback is invoked
(\textsc{E-InvokeM}) which finally re-spawns and monitors the Shop actor.

\smallermath{
    \hspace{-2.5em}
    \begin{tabular}{rrcl}
        $\!\!\!\!\ceval\!\!\!\!$
        &
        $
        (\nu \var{shop})(\nu\var{sup})\left(
            \bl
            \zap{\var{shop}} \\ \parallel 
        \zapactor
            [\var{sup}]
            {\mkwd{shopSup} \app \var{cAP} \app ()}
            {\epsilon}
            {\epsilon}
            {\epsilon}
            \el
            \right)
            $
        &
        $\!\!\!\!\!{\ceval^+}\!\!\!\!\!$
        &
        $
        (\nu \var{shop}')(\nu\var{sup})
        \left(
            \bl
        \zapactor
            [\var{shop}']
            {\mkwd{shop} \: (\var{cAP}, \mkwd{initialStock})}
            {\epsilon}
            {\epsilon}
            {\epsilon} \\
        \parallel 
        \zapactor
            [\var{sup}]
            {\idle{()}}
            {\epsilon}
            {\epsilon}
            {\metapair{\var{shop}'}{\mkwd{shopSup} \: \var{cAP}}}
            \el
            \right)
            $
    \end{tabular}
}

\begin{figure}[t]
    {\footnotesize
        \header{Runtime syntax}
        \[
        \begin{array}{lrcl}
        \text{Cancellation-aware runtime envs.} & \rtenvzap & ::= & \cdot \midspace
        \rtenvzap, \apname \midspace \rtenvzap, \inittokpol : S \midspace
        \rtenvzap, \roleidx{s}{p}: S \midspace \rtenvzap, \roleidx{s}{p} : \zapped \midspace
        \rtenvzap, s : \qty  \\
        \text{Labels} & \ltslbl & ::= & \cdots \midspace \lblzap{s}{p} \midspace
        \lblzapmsg{s}{p}{q}{\ell} \midspace \lblzapdequeue{s}{p}{q}
        \end{array}
        \]

    \headertwo
        {Modified typing rules for configurations}
        {\framebox{$\cseq{\tyenv}{\rtenvzap}{\config{C}}$}~
            \framebox{$\hstateseq{\tyenv}{\rtenvzap}{\tyc}{\hstate}$}}
         \vspace{-0.75em}

        \begin{mathpar}
            \inferrule
            [T-ActorName]
            { \cseq{\tyenv, a : \typid}{\rtenvzap, a}{\config{C} } }
            { \cseq{\tyenv}{\rtenvzap}{(\nu \aname) \config{C} } }

            \inferrule
            [T-ZapActor]
            { }
            { \cseq{\tyenv}{\aname}{\zap{\aname}} }

            \inferrule
            [T-ZapRole]
            { }
            { \cseq{\tyenv}{\roleidx{s}{p} : \zapped}{\zap{\roleidx{s}{p}}} }

            \inferrule
            [T-ZapTok]
            { }
            { \cseq{\tyenv}{\inittokpos : \sta}{\zap{\inittok}} }

            \inferrule
            [T-Actor]
            {
                \threadseq{\tyenv}{\rtenvzap_1}{\tyc}{\threadt} \\
                \hstateseq{\tyenv}{\rtenvzap_2}{\tyc}{\hstate} \\
                \istateseq{\tyenv}{\rtenvzap_3}{\tyc}{\istate} \\\\
                \forall \metapair{b}{\vala} \in \monstate .\; \vseq{\tyenv}{b}{\typid} \,\wedge\,
                \vseq{\tyenv}{\vala}{\sttyfun{\tyc}{\tyc}{\localend}{\localend}{\tyc}}
            }
            { \cseq
                { \tyenv }
                { \rtenvzap_1, \rtenvzap_2, \rtenvzap_3, \aname }
                { \zapactor
                    {\threadt}
                    {\hstate}
                    {\istate}
                    {\monstate}
                }
            }

            \inferrule
            [TH-Handler]
            {
                \vseq{\tyenv}{\vala}{\handlerty{\stin}{\tyc}} \\\\
\vseq{\tyenv}{\valb}{\sttyfun{\tyc}{\tyc}{\localend}{\localend}{\tyc}} \\
                \hstateseq{\tyenv}{\rtenvzap}{\tyc}{\hstate}
            }
            { \hstateseq
               {\tyenv}
               {\rtenvzap, \roleidx{s}{\prole}: \stin}
               {\tyc}
               {\hstate[\storedhandler{s}{\prole}{\metapair{\vala}{\valb}}]}
            }
        \end{mathpar}

        \headertwo
            {Additional LTS rules}
            {\framebox{$\rtenvzap \lbleval{\ltslbl}
            \rtenvzap'$}~\framebox{$\rtenvzap \zaplbleval{\roleidx{s}{p}} \rtenvzap$}}
        \[
            \hspace{-1em}
        \begin{array}{lrcl}
            \textsc{Lbl-ZapMsg} & \rtenvzap, \roleidx{s}{q} : \zapped,
            s : \qentry{p}{q}{\ell}{\tya} \cdot \qty
                                & \lbleval{\lblzapmsg{s}{p}{q}{\ell}} &
            \rtenvzap, \roleidx{s}{q} : \zapped, s : \qty \\
\textsc{Lbl-ZapRecv} \!\!\!\!\!\! &
        \rtenvzap, \roleidx{s}{p} {:} \localoffer{q}{\msg{\ell_i}{\tya_i}.\sta_i}_{i \in I},
        \roleidx{s}{q} {:} \zapped, s {:} Q
                               &
                               \lbleval{\lblzapdequeue{s}{p}{q}}
                               &
                               \rtenvzap, \roleidx{s}{p} {:} \zapped, \roleidx{s}{q}
                               {:} \zapped, s {:} Q \quad (\text{if } \messages{\qrole}{\prole}{\qty} = \emptyset) \\
            \textsc{Lbl-Zap} &
            \rtenvzap, \roleidx{s}{p} : S & \zaplbleval{\roleidx{s}{p}} &
            \rtenvzap, \roleidx{s}{p} : \zapped
        \end{array}
        \]
 }
 \vspace{-1em}
 \caption{\langnamezap: Modified configuration typing rules and type LTS}
 \label{fig:extensions:supervision-failure-2}
\end{figure}

\subsection{Metatheory}
All metatheoretical results continue to hold.
Figure~\ref{fig:extensions:supervision-failure-2} shows the necessary
modifications to the configuration typing rules and type LTS.
We extend runtime type environments to \emph{cancellation-aware} environments
$\rtenvzap$ that include an additional entry of the form $\roleidx{s}{p} :
\zapped$, denoting that endpoint $\roleidx{s}{p}$ has been cancelled.
We also extend the type LTS to account for failure propagation~\cite{BarwellSY022}. 
Rule \textsc{Lbl-Zap} handles role cancellation (e.g.\  due to
\textsc{E-RaiseS}), but is defined separately since it is not needed for
determining behavioural properties of types.

\subsubsection{Preservation}
We need to modify the preservation theorem to include role cancellation in
environment reduction relation: specifically we write $\zapenvred$ for the
relation $\equivsynceval^{?}\!\!\zaplbleval{}^{*}$.
The safety property is unchanged for cancellation-aware environments.

\begin{restatable}[Preservation ($\ceval$, \langnamezap)]{theorem}{zappres}\label{thm:zap-pres}
If $\cseq{\tyenv}{\rtenvzap}{\config{C}}$ with $\safe{\rtenvzap}$
and $\config{C} \ceval \config{D}$,
then there exists some $\rtenvzap'$ such that $\rtenvzap \zapenvred \rtenvzap'$
and $\safe{\rtenvzap'}$
and $\cseq{\tyenv}{\rtenvzap'}{\config{D}}$.
\end{restatable}

\subsubsection{Progress}
\langnamezap enjoys progress since \textsc{E-CancelMsg} discards messages that
cannot be received, and \textsc{E-CancelMsg} invokes the failure continuation
whenever a message will never be sent due to a failure. Monitoring is
orthogonal. The one change is that zapper threads for actors may remain if the
actor name is free in an existing monitoring or initialisation callback.
We require a slightly-adjusted deadlock-freedom property and canonical form to
account for session failure.

\begin{definition}[Deadlock-freedom and compliance (\langnamezap)]
    A runtime environment $\rtenvzap$ is \emph{deadlock-free}, written
    $\dfzap{\rtenvzap}$,
    if
    $\rtenvzap \mathop{\equivsynceval{}^*} \rtenvzap' \not\equivsynceval{}$
    implies that either $\rtenvzap' = s : \emptyq$ or
    $\rtenvzap' = (\roleidxx{s}{\prole_i} : \zapped)_{i \in I}, s : \emptyq$.

    A runtime environment $\rtenvzap$ is \emph{compliant}, written
    $\compzap{\rtenvzap}$, if $\safe{\rtenvzap}$ and $\dfzap{\rtenvzap}$.
\end{definition}

\begin{definition}A \langnamezap configuration $\config{C}$ is in \emph{canonical form} if it can be written:

\smallmath{
        (\nu \tilde{\inittok})
        (\nu \apname_{i \in 1..l})(\nu \sessname_{j \in 1..m})
        (\nu \aname_{k \in 1..n})(
        \ap{\apname_i}{\apstate_i}_{i \in 1..l} \parallel 
        (\qproc{s_j}{\qcontents_j})_{j \in 1..m} \parallel
        \zapactor[a_k]{\threadt_k}{\hstate_k}{\istate_k}{\monstate_k}_{k \in 1..{n' - 1}} \parallel
        \setseq{\zap{\nma}}
        )
    }
    with $(\zap{a_k})_{k \in n' .. n}$ contained in $\setseq{\zap{\nma}}$.
\end{definition}

\begin{restatable}[Progress (\langnamezap)]{theorem}{zapprog}
If $\cseq{\cdot}{\cdot}{\config{C}}$, then either there
exists some $\config{D}$ such that $\config{C} \ceval \config{D}$, or
$\config{C}$ is structurally congruent to the following canonical form:

\smallmath{
    (\nu \tilde{\inittok})(\nu \apname_{i \in 1..m})(\nu \aname_{j \in 1..n})
    (
    \ap{\apname_1}{\apstate_1}_{i \in 1..m} \parallel
    \zapactor[a_j]{\idle{\valc_j}}{\epsilon}{\istate_j}{\monstate_j}_{j \in 1..{n'-1}} \parallel
    (\zap{a_j})_{j \in {n'..n}}
    )
}
\end{restatable}

\subsubsection{Global Progress}

A modified version of global progress holds: in every ongoing session, after a
number of reductions, each session will either be cancelled or perform a
communication action.

\begin{restatable}[Global Progress (\langnamezap)]{theorem}{zapgprog}
    If $\cseq{\cdot}{\cdot}{\config{C}}$ where $\config{C}$ is
    thread-terminating, then
    for every $s \in \activesessions{\config{C}}$,
    then there exist $\config{D}$ and $\config{D}_1$ such that
    $\config{C} \equiv (\nu s)\config{D}$ where
    $\config{D}\cevalannstar{\tau}\config{D}_1$ and either
    $\config{D}_1\cevalann{s}$, or $\config{D}_1 \equiv \config{D}_2$ for some $\config{D}_2$
    where $s \not\in \activesessions{\config{D}_2}$.
\end{restatable}

\subsection{Discussion}
The semantics of $\raiseexn$ follows the Erlang ``let it crash'' methodology
that favours crashing upon errors over defensive programming:
the cascading failure approach, where a crashed actor propagates its failure
to others in a session, is common in affine-session-based
works~(e.g.,~\cite{LagaillardieNY22, FowlerLMD19}) and follows the strategy of
previous dynamically-checked Erlang implementations of session typing~\cite{Fowler16}.

Cancellation is also flexible enough to support other failure-handling
strategies: we can for example implement a $\leave{\vala}$ construct
that allows an actor to exit a session and update its state to $\vala$
\emph{without} terminating, using the following reduction rule:

\smallmath{
    \zapactor{\sessthread{s}{p}{\ctxe[\leave{\vala}]}}{\hstate}{\istate}{\monstate}
    \ceval
    \zapactor{\idle{\vala}}{\hstate}{\istate}{\monstate} \parallel
    \zap{\roleidx{s}{p}}
}
\langname's
combination of event-driven concurrency and cancellation also makes handling
timeouts straightforward.
We could for example extend the $\suspendexn{\valc}{\vala}{\valb}$ construct to
$\suspendtimeout{\valc}{\vala}{\mathsf{t}}{\valb}$, where $\mathsf{t}$ is some
deadline and invoke the failure-handling computation if the deadline is missed.
The failure-handling callback could then e.g.\ either retry or raise an
exception.
Indeed,~\citet{HouLY24} show how session cancellation can be used to enable
flexible timed session types and we expect that their results could be
incorporated into our design.
 \section{Implementation and Evaluation}\label{sec:implementation}

\subsection{Implementation}

Based on our formal design, we have implemented a toolchain for
\langname-style event-driven actor programming in Scala.
It adopts the state machine based API generation approach of
Scribble~\cite{HuY16}: \begin{enumerate}[leftmargin=*]
\item
The user specifies \emph{global types} in the Scribble protocol description
language~\cite{YoshidaHNN13}.
\item
Our toolchain internally uses Scribble to validate global types
according to the MPST-based safety conditions, \emph{project} them to local
types for each role, and construct a representation of each local type based on
\emph{communicating finite state machines} (CFSM)~\cite{10.1145/322374.322380}.
\item
From each CFSM, the toolchain generates a typed,
\emph{protocol-and-role-specific} API for the user to implement that role as
an event-driven \langname actor in native Scala.
\end{enumerate}

\paragraph{Typed APIs for \langname actor programming.}
Consider the \role{Shop} role in our running example
(Fig.~\ref{fig:intro:seqdiag}).
Fig.~\ref{fig:apigen} (top) shows the CFSM for \role{Shop} (with abbreviated message
labels) and a summary of the main generated types and operations (omitting the
type annotations for the \CODE{sid} and \CODE{pay} parameters).
The toolchain generates Scala types for each CFSM state: non-blocking states
(sends or suspends) are coloured \BLUE{blue}, whereas blocking states
(inputs) are \RED{red}.

\begin{figure}[t]
    {\footnotesize
    \header{CFSM (left) and API (right) for \role{Shop} Role (left)}
\begin{tabular}{cc}
\adjustbox{valign=t}{
    \scalebox{0.8}{
    \begin{tikzpicture}
[
    >=stealth',
    NODE/.style={draw,circle,minimum width=3mm,inner sep=0pt,fill=lightgray,font=\footnotesize},
    LAB/.style={font=\footnotesize\ttfamily}
]
\node[NODE] (S1) {$1$};
\node[NODE,below=of S1] (S2) {$2$};
\node[NODE,right=of S2] (S3) {$3$};
\node[NODE,above =of S3] (S4) {$4$};
\node[NODE,below=of S3] (S5) {$5$};
\node[NODE,below=of S5] (S6) {$6$};
\node[NODE,below=of S6] (S7) {$7$};
\node[NODE,left=of S7,yshift=13.3mm,xshift=6mm] (S8) {$8$};
\node[NODE,left=of S7] (S9) {$9$};
\draw[->] (S1) -- node[LAB,above,rotate=90]{C?RI} (S2);
\draw[->] (S2) -- node[LAB,above]{C!Is} (S3);
\path[->] (S3) edge[bend left] node[LAB,above,rotate=90,pos=0.6]{C?GI} (S4);
\path[->] (S4) edge[bend left] node[LAB,below,rotate=90,pos=0.4]{C!II} (S3);
\draw[->] (S3) edge[bend left] node[LAB,below,rotate=90]{C?CO} (S5);
\path[->] (S5) edge[bend left] node[LAB,above,rotate=90,pos=0.2]{C!OOS} (S3);
\draw[->] (S5) -- node[LAB,below,rotate=90]{C!PP} (S6);
\draw[->] (S6) -- node[LAB,below,rotate=90]{P!Buy} (S7);
\draw[->] (S7) -- node[LAB,left]{P?OK} (S8);
\draw[->] (S7) -- node[LAB,above]{P?IF} (S9);
\path[->] (S8) edge[bend left] node[LAB,below,rotate=90,pos=0.2]{C!OK} (S3);
\path[->] (S9) edge[bend left] node[LAB,above,rotate=90,pos=0.1]{C!IF} (S3);
\end{tikzpicture}}}
&
\adjustbox{valign=t}{\scriptsize\ttfamily
\begin{tabular}{l|l|ll}
\textnormal{\bf State}
& \textnormal{\bf State types}
& \textnormal{\textbf{\BLUE{Methods}} (\texttt{send}, \texttt{suspend}) or \textbf{\RED{Input cases}} (extends state type trait)}
\\
\hline
1     & \BLUE{S1Suspend}   &                  suspend[D](d: D, f: (D, \RED{S1}) => Done.\textbf{type}): Done.\textbf{type}
\\
      & \RED{S1}          & \textbf{case class} RequestItems(sid, pay, succ: \BLUE{S2}) \textbf{extends} \RED{S1}
\\
2     & \BLUE{S2}          &                  Customer\_sendItems(pay: ItemList): \BLUE{S3Suspend}
\\
3     & \BLUE{S3Suspend}   &                  suspend[D](d: D, f: (D, \RED{S3}) => Done.\textbf{type}): Done.\textbf{type}
\\
      & \RED{S3}          & \textbf{case class} GetItemInfo(sid, pay, succ: \BLUE{S4}) \textbf{extends} \RED{S3}
\\
      &             & \textbf{case class} Checkout(sid, pay, succ: \BLUE{S5}) \textbf{extends} \RED{S3}
\\
4     & \BLUE{S4}          &                  Customer\_sendItemInfo(pay): \BLUE{S3Suspend}
\\
5     & \BLUE{S5}          &                  Customer\_sendProcessingPayment(): \BLUE{S6}
\\
      &             &                  Customer\_sendOutOfStock(): \BLUE{S3Suspend}
\\
6     & \BLUE{S6}          &                  PaymentProcessor\_sendBuy(pay): \BLUE{S7Suspend}
\\
7     & \BLUE{S7Suspend}   &                  suspend[D](d: D, f: (D, \RED{S7}) => Done.\textbf{type}): Done.\textbf{type}
\\
      & \RED{S7}          & \textbf{case class} OK(sid, pay, succ: \BLUE{S8}) \textbf{extends} \RED{S7}
\\
      &             & \textbf{case class} InsufficientFunds(sid, pay, succ: \BLUE{S9}) \textbf{extends} \RED{S7}
\\
8     & \BLUE{S8}          &                  Customer\_sendOK(pay): \BLUE{S3Suspend}
\\
9     & \BLUE{S9}          &                  Customer\_sendInsufficientFunds(pay): \BLUE{S3Suspend}
\\
\end{tabular}
}
\end{tabular}

\vspace{1em}
\header{Implementation of Customer Request Handler for \role{Shop} Role}
\begin{tabular}{p{0.55\textwidth}|p{0.43\textwidth}}
{\begin{minipage}{0.55\textwidth}
\begin{lstlisting}[columns=flexible]
// d can be used for internal, _session-specific_ actor data
def custReqHandler[T: $\CODEss{\RED{S1orS3}}$](d: DataS, s: $\CODEss{\RED{T}}$): Done.type = {
  s match {
    case c: $\CODEss{\RED{S1}}$ => c match {
      // pay is message payload; succ is successor state
      case RequestItems(sid, pay, succ) =>
        succ.Customer_sendItems(d.summary())
            .suspend(d, custReqHandler[$\CODEss{\RED{S3}}$]) }
    case c: $\CODEss{\RED{S3}}$ => c match {
      case GetItemInfo(sid, pay, succ) =>
        succ.Customer_sendItemInfo(d.lookupItem(pay))
            .suspend(d, custReqHandler[$\CODEss{\RED{S3}}$])
      case Checkout(sid, pay, succ) =>
        if (d.inStock(pay)) {
          succ.Customer_sendProcessingPayment()
              .PaymentProcessor_sendBuy(d.total(pay))
              .suspend(d, paymentResponseHandler)
\end{lstlisting}
\end{minipage}}
&
{\begin{minipage}{0.43\textwidth}
\begin{lstlisting}[columns=flexible]

  // ...continuing on from the left column
  } else {
    val sus = succ.Customer_sendOutOfStock()
    // d.staff: LOption[R1] -- this is a..
    // .."frozen" instance of state type R1
    d.staff match {
      // R1 is the Restock protocol state type
      case x: LSome[R1] =>
        ibecome(d, x, restockHndlr)
      case _: LNone =>
        // Error handling
        throw new RuntimeException
    }
    sus.suspend(d, custReqHandler[$\CODEss{\RED{S3}}$])
  }
}}}
\end{lstlisting}
\end{minipage}}
\end{tabular}
}
\caption{API Generation for Customer-Shop-PaymentProcessor protocol}
\vspace{-1em}
\label{fig:apigen}
\end{figure}

Non-blocking state types provide methods for outputs and
suspend actions, with types specific to each state.
The return type corresponds to the successor state type,
enabling chaining of session actions: e.g.,\ state type \CODE{\BLUE{S2}}
has method \CODE{Customer\_sendItems} for the transition \CODE{C!Is}.
The successor state type \CODE{\BLUE{S3Suspend}} includes a \CODE{suspend}
method to install a handler for the input event of state 3, and to
yield control back to the event loop.
The \CODE{Done.type} type ensures that each handler must either complete the
protocol or perform a \CODE{suspend}.
Input state types are traits implemented by case classes generated
for each input message.
The event loop calls the user-specified handler with the corresponding case
class upon each input event, with each case class carrying an instance of the
successor state type.
For example, \CODE{\RED{S3}} (state 3) is implemented by case classes
\CODE{GetItemInfo} and \CODE{Checkout} for its input transitions, which
respectively carry instances of successor states \CODE{\BLUE{S4}} and
\CODE{\BLUE{S5}}.

The API guides the user to build a \langname actor with handlers for every input
event. Fig.~\ref{fig:apigen} (bottom) handles \CODE{\RED{S1}} and can be passed to
\CODE{\BLUE{S1Suspend}}
right after a session starts. It also handles \CODE{\RED{S3}}
(for \CODE{\BLUE{S3Suspend}}), where the shop
receives \CODE{GetItemInfo} or \CODE{Checkout}.
The runtime for our APIs executes sessions over TCP and uses the Java NIO
library to run the actor event loops.
It supports \emph{fully distributed} sessions between remote \langname actors.

\paragraph*{Switching between sessions.}
As well as supporting the core features and failure handling capabilities of
\langname, our implementation also includes the ability to proactively
\emph{switch} between sessions.
Figure~\ref{fig:apigen} (bottom) shows how this functionality can be used to switch into a
long-running \lstinline+Restock+ session when more stock is needed.
For this purpose, the API allows the user to ``freeze'' unused state type
instances as a type \CODE{LOption[S]} and resume them later by an inline
\CODE{ibecome}.
It allows the callback for a session switching behaviour to be performed
inline with the currently active handler.

\paragraph{Discussion.}

Following our formal model, our generated APIs support a conventional style of
actor programming
where non-blocking operations are programmed in direct-style, in contrast to
approaches that invert both input \emph{and} output
actions~\cite{ZhouHNY20,Thiemann23} through the event loop.

Static Scala typing ensures that handlers safely handle all possible input
events at every stage (by exhaustive matching of case classes), and that state
types offer only the permitted operations at each state (by method typing).
Our API design requires \emph{linear} usage of state type objects
(e.g., \CODE{s} and \CODE{succ}) and frozen session instances. Following other
works~\cite{HuY16,DBLP:conf/ecoop/ScalasY16,DBLP:journals/jfp/Padovani17,DBLP:journals/pacmpl/CastroHJNY19,VieringHEZ21},
we check linearity in a hybrid fashion: the \CODE{Done} return types in
Fig.~\ref{fig:apigen} statically require \CODE{suspend} to be invoked at least
once, but our APIs rule out multiple uses \emph{dynamically}.
We exploit our formal support for failure handling (Sec.~\ref{sec:failure}) to
treat dynamic linearity errors as failures and retain safety and progress.

In summary, our toolchain enables Scala programming of \langname actors that
support concurrent handling of multiple sessions, and ensures their safe
execution.
A statically well-typed actor will \emph{never} select an unavailable branch or
send/receive an incompatible payload type, and an actor system will \emph{never}
become stuck due to mismatching I/O actions.
As in the theory, the system without \CODE{ibecome} will enjoy global progress
provided every handler is terminating (e.g., by avoiding general
recursion/infinite iteration). 
Although we make no formal claims about the system \emph{with}
\CODE{ibecome}, we conjecture that it will also enjoy progress
up-to re-invocation of frozen sessions.

\begin{table}[t]
\footnotesize
\caption{Selected case studies, examples from Savina, and key features of their \langname programs.}
\label{tab:examples}
\vspace{-1em}
\renewcommand{\arraystretch}{0.8}
\begin{tabular}{|l|ccc|ccccccccc|}
\hline
& \multicolumn{3}{c|}{MPST(s)}
& \multicolumn{9}{c|}{\langname actor programs}
\\
& $\oplus/\&$ & $\mu$ & C/P
& mSA & mRA & PP & dSp & dTo & mAP & dAP & be & self
\\
\hline
\hline
Shop (Fig.~\ref{fig:intro:shop-impl})
& $\checkmark$ & $\checkmark$ &
& $\checkmark$ & $\checkmark$ & $\checkmark$ &              &                & $\checkmark$ &              &                &
\\
ShopRestock (Fig.~\ref{fig:apigen})
& $\checkmark$ & $\checkmark$ &
& $\checkmark$ & $\checkmark$ & $\checkmark$ &              &                & $\checkmark$ &              & $\checkmark$   &
\\
Robot~\cite{FowlerASGT23}
& $\checkmark$ & &
& $\checkmark$ &              & $\checkmark$ & $\checkmark$ & ($\checkmark$) &              &              &               &
\\
Chat~\cite{Fowler16}
& $\checkmark$ & $\checkmark$ & $\checkmark$
& $\checkmark$ & $\checkmark$ & $\checkmark$ & $\checkmark$ & $\checkmark$   & $\checkmark$ & $\checkmark$ & $\checkmark$   &
\\
\hline
\hline
Ping-self~\cite{ImamS14a}
& $\checkmark$ & $\checkmark$ & $\checkmark$
& $\checkmark$ & $\checkmark$ &              &              &                & $\checkmark$ &              & $\checkmark$   & $\checkmark$
\\
Ping~\cite{ImamS14a}
& $\checkmark$ & $\checkmark$ &
&              &              &              &              &                &              &              &                &
\\
Fib~\cite{savina}
& & &
& $\checkmark$ & $\checkmark$ & $\checkmark$ & $\checkmark$ & $\checkmark$ & $\checkmark$ & $\checkmark$   & $\checkmark$     &
\\
Dining-self~\cite{ImamS14a}
& $\checkmark$ & $\checkmark$ & $\checkmark$
& $\checkmark$ & $\checkmark$ & $\checkmark$ & $\checkmark$ & ($\checkmark$) & $\checkmark$ &              & $\checkmark$   & $\checkmark$
\\
Dining~\cite{ImamS14a}
& $\checkmark$ & $\checkmark$ &
& $\checkmark$ & $\checkmark$ & $\checkmark$ & $\checkmark$ & ($\checkmark$) & $\checkmark$ &              &                &
\\
Sieve~\cite{ImamS14a}
& $\checkmark$ & $\checkmark$ &
& $\checkmark$ & $\checkmark$ & $\checkmark$ & $\checkmark$ & $\checkmark$   & $\checkmark$ & $\checkmark$ & $\checkmark$   &
\\
\hline
\end{tabular}

\medskip

{\scriptsize
\setlength{\tabcolsep}{1pt}
\begin{tabular}{rclrclrclrclrcl}
$\oplus/\&$ &=& Branch type(s)
&
$\mu$ &=& Recursive type(s)\phantom{aa}
&
C/P &=& \multicolumn{4}{l}{Concurrent/Parallel types}
&
mSA &=& Multiple sessions/actor
\\
mRA &=& Multiple roles/actor\phantom{aa}
&
PP &=& \multicolumn{4}{l}{Parameterised number of actors}
&
dSp &=& \multicolumn{4}{l}{Dynamic actor spawning}
\\
dTo &=& Dynamic topology
&
mAP &=& Multiple APs
&
dAP &=& Dynamic AP creation\phantom{a}
&
be &=& \lstinline+ibecome+\phantom{aa}
&
self &=& Self communication
\end{tabular}
}
\vspace{-2em}
\end{table}

\subsection{Evaluation}
Table~\ref{tab:examples} summarises selected examples from the
Savina~\cite{ImamS14a} benchmark suite (lower) and larger case studies (upper);
the extended version contains sequence diagrams for the larger examples.
Notably, key design features of \langname, e.g.\ support for handling multiple
sessions per actor (mSA) and implementing multiple protocols/roles within actors
(mRA), are crucial to expressing many concurrency patterns.

The ``-self'' versions of Ping and Dining are versions faithful to the original
Akka programs that involve internal coordination using \CODE{self}~\CODE{!}~\CODE{msg}
operations, but our APIs can express equivalent behaviour more simply without
needing self-communication.

The {\small$(\checkmark)$} distinguishes simpler forms of dynamic topologies
(dTo) due to a parameterised number of clients dynamically connecting to a
central server, from richer structures such as the parent-children tree
topology dynamically created in Fib and the user-driven
dynamic connections between clients and chat rooms in
Chat; note both the latter involve dynamic access point creation (dAP).

The \textbf{Robot Coordination} use case (from Actyx AG~\cite{actyx}, originally
described in~\cite{FowlerASGT23}) describes multiple Robots accessing a
Warehouse with a single Door, where only one Robot is allowed in the Warehouse
at a time.
\langname allows the Door and Warehouse to safely handle the concurrent
interleavings of events across \emph{any number} (PP, dSP) of separate Robot
sessions (mSA).

The \textbf{Chat Server} use case~\cite{Fowler16} involves an arbitrary number
of Clients (PP) using a Registry to create new chat Rooms. Each Client can
dynamically join and leave any existing Room.
Rooms are created by spawning new Room actors (dSp) with fresh access points
(dAP, mAP), and we allow any Client to establish sessions with the Registry or
any Room asynchronously (dTo).
We decompose the Client-Registry and the Client-Room interactions into separate
protocols (C/P, mAP), and use \CODE{ibecome} (be) in the Room actor to broadcast
chat messages to all Clients currently in that Room.

 \section{Related Work}\label{sec:related}
We have given an overview of previous work on session-typed actors
in~\secref{sec:intro:key-principles}.
Several works investigate event-driven session typing.
\citet{ZhouHNY20} introduce a multiparty session type discipline with
statically-checked refinement types in F$\star$, using callbacks for each send and receive to avoid reasoning about linearity.
\citet{Miu0Y021} and~\citet{Thiemann23} adopt this approach for web applications
and Agda~\cite{Norell08} respectively.
Our approach only yields control to the event loop on actor \emph{receives}, as
in idiomatic actor programming.
\citet{DBLP:conf/ecoop/HuKPYH10} and~\citet{KouzapasYHH16} introduce a binary
session $\pi$-calculus with primitives for event loops; we
instead encode the event loop directly in the semantics.
\citet{VieringHEZ21} support fault-tolerant session-typed distributed
programming with inversion of control on both input and outputs.
They establish a version of global progress for trees of
\emph{subsessions}~\cite{DemangeonH12}; we instead establish global progress for
every session in the system.
These works all focus on process calculi rather than language design.

\citet{DBLP:journals/jlap/CicconeDP24} developed \emph{fair termination} for
\emph{synchronous} multiparty sessions, a strong property that subsumes our
global progress: it implies every \emph{role} fairly terminates, whereas our
coarser-grained property is per \emph{session}.
\citet{DBLP:conf/ecoop/PadovaniZ25} developed fair termination for
asynchronous \emph{binary} sessions and show that fair termination implies
orphan message freedom~\cite{DBLP:journals/lmcs/ChenDSY17}. Our system ensures
orphan message freedom for terminated multiparty sessions (as in
\cite{DBLP:conf/esop/DenielouY12,DBLP:conf/icalp/DenielouY13}), but we do
not aim to restrict \langname to terminating sessions.
We may be able to strengthen our formal results by adapting the fairness
conditions discussed
in~\cite{DBLP:journals/jlap/CicconeDP24,DBLP:conf/ecoop/PadovaniZ25} (developed
for session $\pi$-calculi) to our event-driven actor setting; however, features
such as combining session creation and parallel composition into one term (based
on linear logic) are more restrictive than in our model.

\emph{Mailbox types}~\cite{deLiguoroP18,FowlerASGT23} capture the expected
contents of a mailbox as a commutative regular expression, and ensure
that processes do not receive unexpected messages. Mailbox and session types
both aim to ensure safe communication but address different problems: session
types suit \emph{structured} interactions among known participants, whereas
mailbox types are better when participants are unknown and message ordering is
unimportant.  Mailbox types cannot yet handle failure.

\emph{Internal delegation}~\cite{DBLP:journals/tcs/CastellaniDGH20}
allows channels to migrate within a session, and may provide a way to
relate our model to session $\pi$-calculi via encodings (cf.\
\cite{KouzapasYHH16}).
\citet{DBLP:conf/ppdp/BarbaneraDGY23,DBLP:journals/jlap/BarbaneraBD25a}
emphasise simplified, \emph{compositional} multiparty session models, and
it would be interesting to formally compare their approach (based on parallel
composition and forwarding) to our single-threaded model.

Effpi~\cite{ScalasYB19} uses Scala's dependent function types to allow functions
to be checked against interactions written in a type-level DSL, supporting
verification of properties such as liveness and termination. This is different
to session typing (e.g., supporting parameterised interactions but not
branching), but it is unclear how their actor API would scale to multiple
session-style interactions.

 \section{Conclusion and Future Work}\label{sec:conclusion}
This paper introduces \langname, an actor language that rules out communication
mismatches and deadlocks using \emph{multiparty session types}.  Key to our
approach is a novel combination of a flow-sensitive effect system and
first-class message handlers. \langname scales to Erlang-style
failure handling.
Finally, we have shown an implementation of \langname in Scala using an API
generation approach, and evaluated our implementation on two larger case studies
and a selection of examples from the Savina benchmark suite.
In future it would be interesting to investigate path-dependent types in our
implementation, and to investigate finer-grained models for failure recovery
(e.g.,~\cite{NeykovaY17}).

\clearpage
\section*{Acknowledgements}
We thank the anonymous OOPSLA reviewers for their thorough comments that greatly
improved the quality of the paper. Thanks to Phil Trinder for his helpful
comments on an early draft, and to Matthew Alan Le Brun and Alceste Scalas for
discussions about the asynchronous semantics of multiparty session calculi.
This work was supported by EPSRC Grant EP/T014628/1 (STARDUST).

\section*{Data Availability Statement}
Our implementation is available on Zenodo~\cite{artifact}.

\bibliographystyle{ACM-Reference-Format}
\bibliography{eventactors}


\begin{thebibliography}{60}


\ifx \showCODEN    \undefined \def \showCODEN     #1{\unskip}     \fi
\ifx \showISBNx    \undefined \def \showISBNx     #1{\unskip}     \fi
\ifx \showISBNxiii \undefined \def \showISBNxiii  #1{\unskip}     \fi
\ifx \showISSN     \undefined \def \showISSN      #1{\unskip}     \fi
\ifx \showLCCN     \undefined \def \showLCCN      #1{\unskip}     \fi
\ifx \shownote     \undefined \def \shownote      #1{#1}          \fi
\ifx \showarticletitle \undefined \def \showarticletitle #1{#1}   \fi
\ifx \showURL      \undefined \def \showURL       {\relax}        \fi
\providecommand\bibfield[2]{#2}
\providecommand\bibinfo[2]{#2}
\providecommand\natexlab[1]{#1}
\providecommand\showeprint[2][]{arXiv:#2}

\bibitem[act(2023)]%
        {actyx}
 \bibinfo{year}{2023}\natexlab{}.
\newblock \bibinfo{booktitle}{\emph{Actyx AG}}.
\newblock
\urldef\tempurl%
\url{https://actyx.io}
\showURL{%
\tempurl}


\bibitem[Agha(1990)]%
        {Agha90}
\bibfield{author}{\bibinfo{person}{Gul~A. Agha}.}
  \bibinfo{year}{1990}\natexlab{}.
\newblock \bibinfo{booktitle}{\emph{{ACTORS} - a model of concurrent
  computation in distributed systems}}.
\newblock \bibinfo{publisher}{{MIT} Press}.
\newblock


\bibitem[Atkey(2009)]%
        {Atkey09}
\bibfield{author}{\bibinfo{person}{Robert Atkey}.}
  \bibinfo{year}{2009}\natexlab{}.
\newblock \showarticletitle{Parameterised notions of computation}.
\newblock \bibinfo{journal}{\emph{J. Funct. Program.}} \bibinfo{volume}{19},
  \bibinfo{number}{3-4} (\bibinfo{year}{2009}), \bibinfo{pages}{335--376}.
\newblock


\bibitem[Barbanera et~al\mbox{.}(2025)]%
        {DBLP:journals/jlap/BarbaneraBD25a}
\bibfield{author}{\bibinfo{person}{Franco Barbanera}, \bibinfo{person}{Viviana
  Bono}, {and} \bibinfo{person}{Mariangiola Dezani{-}Ciancaglini}.}
  \bibinfo{year}{2025}\natexlab{}.
\newblock \showarticletitle{Open compliance in multiparty sessions with partial
  typing}.
\newblock \bibinfo{journal}{\emph{J. Log. Algebraic Methods Program.}}
  \bibinfo{volume}{144} (\bibinfo{year}{2025}), \bibinfo{pages}{101046}.
\newblock
\href{https://doi.org/10.1016/J.JLAMP.2025.101046}{doi:\nolinkurl{10.1016/J.JLAMP.2025.101046}}


\bibitem[Barbanera et~al\mbox{.}(2023)]%
        {DBLP:conf/ppdp/BarbaneraDGY23}
\bibfield{author}{\bibinfo{person}{Franco Barbanera},
  \bibinfo{person}{Mariangiola Dezani{-}Ciancaglini}, \bibinfo{person}{Lorenzo
  Gheri}, {and} \bibinfo{person}{Nobuko Yoshida}.}
  \bibinfo{year}{2023}\natexlab{}.
\newblock \showarticletitle{Multicompatibility for Multiparty-Session
  Composition}. In \bibinfo{booktitle}{\emph{International Symposium on
  Principles and Practice of Declarative Programming, {PPDP} 2023, Lisboa,
  Portugal, October 22-23, 2023}}, \bibfield{editor}{\bibinfo{person}{Santiago
  Escobar} {and} \bibinfo{person}{Vasco~T. Vasconcelos}} (Eds.).
  \bibinfo{publisher}{{ACM}}, \bibinfo{pages}{2:1--2:15}.
\newblock
\href{https://doi.org/10.1145/3610612.3610614}{doi:\nolinkurl{10.1145/3610612.3610614}}


\bibitem[Barwell et~al\mbox{.}(2022)]%
        {BarwellSY022}
\bibfield{author}{\bibinfo{person}{Adam~D. Barwell}, \bibinfo{person}{Alceste
  Scalas}, \bibinfo{person}{Nobuko Yoshida}, {and} \bibinfo{person}{Fangyi
  Zhou}.} \bibinfo{year}{2022}\natexlab{}.
\newblock \showarticletitle{Generalised Multiparty Session Types with
  Crash-Stop Failures}. In \bibinfo{booktitle}{\emph{{CONCUR}}}
  \emph{(\bibinfo{series}{LIPIcs}, Vol.~\bibinfo{volume}{243})}.
  \bibinfo{publisher}{Schloss Dagstuhl - Leibniz-Zentrum f{\"{u}}r Informatik},
  \bibinfo{pages}{35:1--35:25}.
\newblock


\bibitem[Bettini et~al\mbox{.}(2008)]%
        {DBLP:conf/birthday/BettiniCDGV08}
\bibfield{author}{\bibinfo{person}{Lorenzo Bettini}, \bibinfo{person}{Sara
  Capecchi}, \bibinfo{person}{Mariangiola Dezani{-}Ciancaglini},
  \bibinfo{person}{Elena Giachino}, {and} \bibinfo{person}{Betti Venneri}.}
  \bibinfo{year}{2008}\natexlab{}.
\newblock \showarticletitle{Session and Union Types for Object Oriented
  Programming}. In \bibinfo{booktitle}{\emph{Concurrency, Graphs and Models,
  Essays Dedicated to Ugo Montanari on the Occasion of His 65th Birthday}}
  \emph{(\bibinfo{series}{Lecture Notes in Computer Science},
  Vol.~\bibinfo{volume}{5065})}, \bibfield{editor}{\bibinfo{person}{Pierpaolo
  Degano}, \bibinfo{person}{Rocco~De Nicola}, {and} \bibinfo{person}{Jos{\'{e}}
  Meseguer}} (Eds.). \bibinfo{publisher}{Springer}, \bibinfo{pages}{659--680}.
\newblock
\href{https://doi.org/10.1007/978-3-540-68679-8\_41}{doi:\nolinkurl{10.1007/978-3-540-68679-8\_41}}


\bibitem[Brand and Zafiropulo(1983)]%
        {10.1145/322374.322380}
\bibfield{author}{\bibinfo{person}{Daniel Brand} {and} \bibinfo{person}{Pitro
  Zafiropulo}.} \bibinfo{year}{1983}\natexlab{}.
\newblock \showarticletitle{On Communicating Finite-State Machines}.
\newblock \bibinfo{journal}{\emph{J. ACM}} \bibinfo{volume}{30},
  \bibinfo{number}{2} (\bibinfo{date}{apr} \bibinfo{year}{1983}),
  \bibinfo{pages}{323?342}.
\newblock
\showISSN{0004-5411}
\href{https://doi.org/10.1145/322374.322380}{doi:\nolinkurl{10.1145/322374.322380}}


\bibitem[Carbone et~al\mbox{.}(2016)]%
        {CarboneLMSW16}
\bibfield{author}{\bibinfo{person}{Marco Carbone}, \bibinfo{person}{Sam
  Lindley}, \bibinfo{person}{Fabrizio Montesi}, \bibinfo{person}{Carsten
  Sch{\"{u}}rmann}, {and} \bibinfo{person}{Philip Wadler}.}
  \bibinfo{year}{2016}\natexlab{}.
\newblock \showarticletitle{Coherence Generalises Duality: {A} Logical
  Explanation of Multiparty Session Types}. In
  \bibinfo{booktitle}{\emph{{CONCUR}}} \emph{(\bibinfo{series}{LIPIcs},
  Vol.~\bibinfo{volume}{59})}. \bibinfo{publisher}{Schloss Dagstuhl -
  Leibniz-Zentrum f{\"{u}}r Informatik}, \bibinfo{pages}{33:1--33:15}.
\newblock


\bibitem[Castellani et~al\mbox{.}(2020)]%
        {DBLP:journals/tcs/CastellaniDGH20}
\bibfield{author}{\bibinfo{person}{Ilaria Castellani},
  \bibinfo{person}{Mariangiola Dezani{-}Ciancaglini}, \bibinfo{person}{Paola
  Giannini}, {and} \bibinfo{person}{Ross Horne}.}
  \bibinfo{year}{2020}\natexlab{}.
\newblock \showarticletitle{Global types with internal delegation}.
\newblock \bibinfo{journal}{\emph{Theor. Comput. Sci.}}  \bibinfo{volume}{807}
  (\bibinfo{year}{2020}), \bibinfo{pages}{128--153}.
\newblock
\href{https://doi.org/10.1016/J.TCS.2019.09.027}{doi:\nolinkurl{10.1016/J.TCS.2019.09.027}}


\bibitem[Castro{-}Perez et~al\mbox{.}(2019)]%
        {DBLP:journals/pacmpl/CastroHJNY19}
\bibfield{author}{\bibinfo{person}{David Castro{-}Perez},
  \bibinfo{person}{Raymond Hu}, \bibinfo{person}{Sung{-}Shik Jongmans},
  \bibinfo{person}{Nicholas Ng}, {and} \bibinfo{person}{Nobuko Yoshida}.}
  \bibinfo{year}{2019}\natexlab{}.
\newblock \showarticletitle{Distributed programming using role-parametric
  session types in go: statically-typed endpoint APIs for
  dynamically-instantiated communication structures}.
\newblock \bibinfo{journal}{\emph{Proc. {ACM} Program. Lang.}}
  \bibinfo{volume}{3}, \bibinfo{number}{{POPL}} (\bibinfo{year}{2019}),
  \bibinfo{pages}{29:1--29:30}.
\newblock
\href{https://doi.org/10.1145/3290342}{doi:\nolinkurl{10.1145/3290342}}


\bibitem[Chaudhuri(2009)]%
        {Chaudhuri09}
\bibfield{author}{\bibinfo{person}{Avik Chaudhuri}.}
  \bibinfo{year}{2009}\natexlab{}.
\newblock \showarticletitle{A {C}oncurrent {ML} Library in {C}oncurrent
  {H}askell}. In \bibinfo{booktitle}{\emph{{ICFP}}}.
  \bibinfo{publisher}{{ACM}}.
\newblock


\bibitem[Chen et~al\mbox{.}(2017)]%
        {DBLP:journals/lmcs/ChenDSY17}
\bibfield{author}{\bibinfo{person}{Tzu{-}Chun Chen},
  \bibinfo{person}{Mariangiola Dezani{-}Ciancaglini}, \bibinfo{person}{Alceste
  Scalas}, {and} \bibinfo{person}{Nobuko Yoshida}.}
  \bibinfo{year}{2017}\natexlab{}.
\newblock \showarticletitle{On the Preciseness of Subtyping in Session Types}.
\newblock \bibinfo{journal}{\emph{Log. Methods Comput. Sci.}}
  \bibinfo{volume}{13}, \bibinfo{number}{2} (\bibinfo{year}{2017}).
\newblock
\href{https://doi.org/10.23638/LMCS-13(2:12)2017}{doi:\nolinkurl{10.23638/LMCS-13(2:12)2017}}


\bibitem[Ciccone et~al\mbox{.}(2024)]%
        {DBLP:journals/jlap/CicconeDP24}
\bibfield{author}{\bibinfo{person}{Luca Ciccone}, \bibinfo{person}{Francesco
  Dagnino}, {and} \bibinfo{person}{Luca Padovani}.}
  \bibinfo{year}{2024}\natexlab{}.
\newblock \showarticletitle{Fair termination of multiparty sessions}.
\newblock \bibinfo{journal}{\emph{J. Log. Algebraic Methods Program.}}
  \bibinfo{volume}{139} (\bibinfo{year}{2024}), \bibinfo{pages}{100964}.
\newblock
\href{https://doi.org/10.1016/J.JLAMP.2024.100964}{doi:\nolinkurl{10.1016/J.JLAMP.2024.100964}}


\bibitem[Coppo et~al\mbox{.}(2016)]%
        {CoppoDYP16}
\bibfield{author}{\bibinfo{person}{Mario Coppo}, \bibinfo{person}{Mariangiola
  Dezani{-}Ciancaglini}, \bibinfo{person}{Nobuko Yoshida}, {and}
  \bibinfo{person}{Luca Padovani}.} \bibinfo{year}{2016}\natexlab{}.
\newblock \showarticletitle{Global progress for dynamically interleaved
  multiparty sessions}.
\newblock \bibinfo{journal}{\emph{Math. Struct. Comput. Sci.}}
  \bibinfo{volume}{26}, \bibinfo{number}{2} (\bibinfo{year}{2016}),
  \bibinfo{pages}{238--302}.
\newblock


\bibitem[de'Liguoro and Padovani(2018)]%
        {deLiguoroP18}
\bibfield{author}{\bibinfo{person}{Ugo de'Liguoro} {and} \bibinfo{person}{Luca
  Padovani}.} \bibinfo{year}{2018}\natexlab{}.
\newblock \showarticletitle{Mailbox Types for Unordered Interactions}. In
  \bibinfo{booktitle}{\emph{{ECOOP}}} \emph{(\bibinfo{series}{LIPIcs},
  Vol.~\bibinfo{volume}{109})}. \bibinfo{publisher}{Schloss Dagstuhl -
  Leibniz-Zentrum f{\"{u}}r Informatik}, \bibinfo{pages}{15:1--15:28}.
\newblock


\bibitem[Demangeon and Honda(2012)]%
        {DemangeonH12}
\bibfield{author}{\bibinfo{person}{Romain Demangeon} {and}
  \bibinfo{person}{Kohei Honda}.} \bibinfo{year}{2012}\natexlab{}.
\newblock \showarticletitle{Nested Protocols in Session Types}. In
  \bibinfo{booktitle}{\emph{{CONCUR}}} \emph{(\bibinfo{series}{Lecture Notes in
  Computer Science}, Vol.~\bibinfo{volume}{7454})}.
  \bibinfo{publisher}{Springer}, \bibinfo{pages}{272--286}.
\newblock


\bibitem[Deni{\'{e}}lou and Yoshida(2012)]%
        {DBLP:conf/esop/DenielouY12}
\bibfield{author}{\bibinfo{person}{Pierre{-}Malo Deni{\'{e}}lou} {and}
  \bibinfo{person}{Nobuko Yoshida}.} \bibinfo{year}{2012}\natexlab{}.
\newblock \showarticletitle{Multiparty Session Types Meet Communicating
  Automata}. In \bibinfo{booktitle}{\emph{Programming Languages and Systems -
  21st European Symposium on Programming, {ESOP} 2012, Held as Part of the
  European Joint Conferences on Theory and Practice of Software, {ETAPS} 2012,
  Tallinn, Estonia, March 24 - April 1, 2012. Proceedings}}
  \emph{(\bibinfo{series}{Lecture Notes in Computer Science},
  Vol.~\bibinfo{volume}{7211})}, \bibfield{editor}{\bibinfo{person}{Helmut
  Seidl}} (Ed.). \bibinfo{publisher}{Springer}, \bibinfo{pages}{194--213}.
\newblock
\href{https://doi.org/10.1007/978-3-642-28869-2\_10}{doi:\nolinkurl{10.1007/978-3-642-28869-2\_10}}


\bibitem[Deni{\'{e}}lou and Yoshida(2013a)]%
        {DenielouY13}
\bibfield{author}{\bibinfo{person}{Pierre{-}Malo Deni{\'{e}}lou} {and}
  \bibinfo{person}{Nobuko Yoshida}.} \bibinfo{year}{2013}\natexlab{a}.
\newblock \showarticletitle{Multiparty Compatibility in Communicating Automata:
  Characterisation and Synthesis of Global Session Types}. In
  \bibinfo{booktitle}{\emph{{ICALP} {(2)}}} \emph{(\bibinfo{series}{Lecture
  Notes in Computer Science}, Vol.~\bibinfo{volume}{7966})}.
  \bibinfo{publisher}{Springer}, \bibinfo{pages}{174--186}.
\newblock


\bibitem[Deni{\'{e}}lou and Yoshida(2013b)]%
        {DBLP:conf/icalp/DenielouY13}
\bibfield{author}{\bibinfo{person}{Pierre{-}Malo Deni{\'{e}}lou} {and}
  \bibinfo{person}{Nobuko Yoshida}.} \bibinfo{year}{2013}\natexlab{b}.
\newblock \showarticletitle{Multiparty Compatibility in Communicating Automata:
  Characterisation and Synthesis of Global Session Types}. In
  \bibinfo{booktitle}{\emph{Automata, Languages, and Programming - 40th
  International Colloquium, {ICALP} 2013, Riga, Latvia, July 8-12, 2013,
  Proceedings, Part {II}}} \emph{(\bibinfo{series}{Lecture Notes in Computer
  Science}, Vol.~\bibinfo{volume}{7966})},
  \bibfield{editor}{\bibinfo{person}{Fedor~V. Fomin}, \bibinfo{person}{Rusins
  Freivalds}, \bibinfo{person}{Marta~Z. Kwiatkowska}, {and}
  \bibinfo{person}{David Peleg}} (Eds.). \bibinfo{publisher}{Springer},
  \bibinfo{pages}{174--186}.
\newblock
\href{https://doi.org/10.1007/978-3-642-39212-2\_18}{doi:\nolinkurl{10.1007/978-3-642-39212-2\_18}}


\bibitem[Foster et~al\mbox{.}(2002)]%
        {FosterTA02}
\bibfield{author}{\bibinfo{person}{Jeffrey~S. Foster}, \bibinfo{person}{Tachio
  Terauchi}, {and} \bibinfo{person}{Alex Aiken}.}
  \bibinfo{year}{2002}\natexlab{}.
\newblock \showarticletitle{Flow-Sensitive Type Qualifiers}. In
  \bibinfo{booktitle}{\emph{{PLDI}}}. \bibinfo{publisher}{{ACM}},
  \bibinfo{pages}{1--12}.
\newblock


\bibitem[Fowler(2016)]%
        {Fowler16}
\bibfield{author}{\bibinfo{person}{Simon Fowler}.}
  \bibinfo{year}{2016}\natexlab{}.
\newblock \showarticletitle{An {E}rlang Implementation of Multiparty Session
  Actors}. In \bibinfo{booktitle}{\emph{{ICE}}}
  \emph{(\bibinfo{series}{{EPTCS}}, Vol.~\bibinfo{volume}{223})}.
  \bibinfo{pages}{36--50}.
\newblock


\bibitem[Fowler et~al\mbox{.}(2023)]%
        {FowlerASGT23}
\bibfield{author}{\bibinfo{person}{Simon Fowler}, \bibinfo{person}{Duncan~Paul
  Attard}, \bibinfo{person}{Franciszek Sowul}, \bibinfo{person}{Simon~J. Gay},
  {and} \bibinfo{person}{Phil Trinder}.} \bibinfo{year}{2023}\natexlab{}.
\newblock \showarticletitle{Special Delivery: Programming with Mailbox Types}.
\newblock \bibinfo{journal}{\emph{Proc. {ACM} Program. Lang.}}
  \bibinfo{volume}{7}, \bibinfo{number}{{ICFP}} (\bibinfo{year}{2023}),
  \bibinfo{pages}{78--107}.
\newblock


\bibitem[Fowler and Hu(2026)]%
        {artifact}
\bibfield{author}{\bibinfo{person}{Simon Fowler} {and} \bibinfo{person}{Raymond
  Hu}.} \bibinfo{year}{2026}\natexlab{}.
\newblock \bibinfo{booktitle}{\emph{Speak Now: Safe Actor Programming with
  Multiparty Session Types (Artifact)}}.
\newblock
\urldef\tempurl%
\url{https://doi.org/10.5281/zenodo.18792000}
\showURL{%
\tempurl}


\bibitem[Fowler et~al\mbox{.}(2019)]%
        {FowlerLMD19}
\bibfield{author}{\bibinfo{person}{Simon Fowler}, \bibinfo{person}{Sam
  Lindley}, \bibinfo{person}{J.~Garrett Morris}, {and}
  \bibinfo{person}{S{\'{a}}ra Decova}.} \bibinfo{year}{2019}\natexlab{}.
\newblock \showarticletitle{Exceptional asynchronous session types: session
  types without tiers}.
\newblock \bibinfo{journal}{\emph{Proc. {ACM} Program. Lang.}}
  \bibinfo{volume}{3}, \bibinfo{number}{{POPL}} (\bibinfo{year}{2019}),
  \bibinfo{pages}{28:1--28:29}.
\newblock


\bibitem[Fowler et~al\mbox{.}(2017)]%
        {FowlerLW17}
\bibfield{author}{\bibinfo{person}{Simon Fowler}, \bibinfo{person}{Sam
  Lindley}, {and} \bibinfo{person}{Philip Wadler}.}
  \bibinfo{year}{2017}\natexlab{}.
\newblock \showarticletitle{Mixing Metaphors: Actors as Channels and Channels
  as Actors}. In \bibinfo{booktitle}{\emph{{ECOOP}}}
  \emph{(\bibinfo{series}{LIPIcs}, Vol.~\bibinfo{volume}{74})}.
  \bibinfo{publisher}{Schloss Dagstuhl - Leibniz-Zentrum f{\"{u}}r Informatik},
  \bibinfo{pages}{11:1--11:28}.
\newblock


\bibitem[Francalanza and Tabone(2023)]%
        {FrancalanzaT23}
\bibfield{author}{\bibinfo{person}{Adrian Francalanza} {and}
  \bibinfo{person}{Gerard Tabone}.} \bibinfo{year}{2023}\natexlab{}.
\newblock \showarticletitle{{ElixirST}: {A} session-based type system for
  {E}lixir modules}.
\newblock \bibinfo{journal}{\emph{J. Log. Algebraic Methods Program.}}
  \bibinfo{volume}{135} (\bibinfo{year}{2023}), \bibinfo{pages}{100891}.
\newblock


\bibitem[Gay and Vasconcelos(2010)]%
        {GayV10}
\bibfield{author}{\bibinfo{person}{Simon~J. Gay} {and}
  \bibinfo{person}{Vasco~Thudichum Vasconcelos}.}
  \bibinfo{year}{2010}\natexlab{}.
\newblock \showarticletitle{Linear type theory for asynchronous session types}.
\newblock \bibinfo{journal}{\emph{J. Funct. Program.}} \bibinfo{volume}{20},
  \bibinfo{number}{1} (\bibinfo{year}{2010}), \bibinfo{pages}{19--50}.
\newblock


\bibitem[Gordon(2017)]%
        {Gordon17}
\bibfield{author}{\bibinfo{person}{Colin~S. Gordon}.}
  \bibinfo{year}{2017}\natexlab{}.
\newblock \showarticletitle{A Generic Approach to Flow-Sensitive Polymorphic
  Effects}. In \bibinfo{booktitle}{\emph{{ECOOP}}}
  \emph{(\bibinfo{series}{LIPIcs}, Vol.~\bibinfo{volume}{74})}.
  \bibinfo{publisher}{Schloss Dagstuhl - Leibniz-Zentrum f{\"{u}}r Informatik},
  \bibinfo{pages}{13:1--13:31}.
\newblock


\bibitem[Harvey et~al\mbox{.}(2021)]%
        {HarveyFDG21}
\bibfield{author}{\bibinfo{person}{Paul Harvey}, \bibinfo{person}{Simon
  Fowler}, \bibinfo{person}{Ornela Dardha}, {and} \bibinfo{person}{Simon~J.
  Gay}.} \bibinfo{year}{2021}\natexlab{}.
\newblock \showarticletitle{Multiparty Session Types for Safe Runtime
  Adaptation in an Actor Language}. In \bibinfo{booktitle}{\emph{{ECOOP}}}
  \emph{(\bibinfo{series}{LIPIcs}, Vol.~\bibinfo{volume}{194})}.
  \bibinfo{publisher}{Schloss Dagstuhl - Leibniz-Zentrum f{\"{u}}r Informatik},
  \bibinfo{pages}{10:1--10:30}.
\newblock


\bibitem[Hewitt et~al\mbox{.}(1973)]%
        {HewittBS73}
\bibfield{author}{\bibinfo{person}{Carl Hewitt}, \bibinfo{person}{Peter~Boehler
  Bishop}, {and} \bibinfo{person}{Richard Steiger}.}
  \bibinfo{year}{1973}\natexlab{}.
\newblock \showarticletitle{A Universal Modular {ACTOR} Formalism for
  Artificial Intelligence}. In \bibinfo{booktitle}{\emph{{IJCAI}}}.
  \bibinfo{publisher}{William Kaufmann}, \bibinfo{pages}{235--245}.
\newblock


\bibitem[Honda et~al\mbox{.}(2008)]%
        {HondaYC08}
\bibfield{author}{\bibinfo{person}{Kohei Honda}, \bibinfo{person}{Nobuko
  Yoshida}, {and} \bibinfo{person}{Marco Carbone}.}
  \bibinfo{year}{2008}\natexlab{}.
\newblock \showarticletitle{Multiparty asynchronous session types}. In
  \bibinfo{booktitle}{\emph{{POPL}}}. \bibinfo{publisher}{{ACM}},
  \bibinfo{pages}{273--284}.
\newblock


\bibitem[Hou et~al\mbox{.}(2024)]%
        {HouLY24}
\bibfield{author}{\bibinfo{person}{Ping Hou}, \bibinfo{person}{Nicolas
  Lagaillardie}, {and} \bibinfo{person}{Nobuko Yoshida}.}
  \bibinfo{year}{2024}\natexlab{}.
\newblock \showarticletitle{Fearless Asynchronous Communications with Timed
  Multiparty Session Protocols}. In \bibinfo{booktitle}{\emph{{ECOOP}}}
  \emph{(\bibinfo{series}{LIPIcs}, Vol.~\bibinfo{volume}{313})}.
  \bibinfo{publisher}{Schloss Dagstuhl - Leibniz-Zentrum f{\"{u}}r Informatik},
  \bibinfo{pages}{19:1--19:30}.
\newblock


\bibitem[Hu et~al\mbox{.}(2010)]%
        {DBLP:conf/ecoop/HuKPYH10}
\bibfield{author}{\bibinfo{person}{Raymond Hu}, \bibinfo{person}{Dimitrios
  Kouzapas}, \bibinfo{person}{Olivier Pernet}, \bibinfo{person}{Nobuko
  Yoshida}, {and} \bibinfo{person}{Kohei Honda}.}
  \bibinfo{year}{2010}\natexlab{}.
\newblock \showarticletitle{Type-Safe Eventful Sessions in Java}. In
  \bibinfo{booktitle}{\emph{{ECOOP} 2010 - Object-Oriented Programming, 24th
  European Conference, Maribor, Slovenia, June 21-25, 2010. Proceedings}}
  \emph{(\bibinfo{series}{Lecture Notes in Computer Science},
  Vol.~\bibinfo{volume}{6183})}, \bibfield{editor}{\bibinfo{person}{Theo
  D'Hondt}} (Ed.). \bibinfo{publisher}{Springer}, \bibinfo{pages}{329--353}.
\newblock
\href{https://doi.org/10.1007/978-3-642-14107-2\_16}{doi:\nolinkurl{10.1007/978-3-642-14107-2\_16}}


\bibitem[Hu and Yoshida(2016)]%
        {HuY16}
\bibfield{author}{\bibinfo{person}{Raymond Hu} {and} \bibinfo{person}{Nobuko
  Yoshida}.} \bibinfo{year}{2016}\natexlab{}.
\newblock \showarticletitle{Hybrid Session Verification Through Endpoint {API}
  Generation}. In \bibinfo{booktitle}{\emph{{FASE}}}
  \emph{(\bibinfo{series}{Lecture Notes in Computer Science},
  Vol.~\bibinfo{volume}{9633})}. \bibinfo{publisher}{Springer},
  \bibinfo{pages}{401--418}.
\newblock


\bibitem[Hu and Yoshida(2017)]%
        {HuY17}
\bibfield{author}{\bibinfo{person}{Raymond Hu} {and} \bibinfo{person}{Nobuko
  Yoshida}.} \bibinfo{year}{2017}\natexlab{}.
\newblock \showarticletitle{Explicit Connection Actions in Multiparty Session
  Types}. In \bibinfo{booktitle}{\emph{{FASE}}} \emph{(\bibinfo{series}{Lecture
  Notes in Computer Science}, Vol.~\bibinfo{volume}{10202})}.
  \bibinfo{publisher}{Springer}, \bibinfo{pages}{116--133}.
\newblock


\bibitem[Imam({[n.\,d.]})]%
        {savina}
\bibfield{author}{\bibinfo{person}{Shams Imam}.}
  \bibinfo{year}{[n.\,d.]}\natexlab{}.
\newblock \bibinfo{title}{{Savina Actor Benchmark Suite}}.
\newblock \bibinfo{howpublished}{\url{https://github.com/shamsimam/savina}}.
\newblock
\newblock
\shownote{Accessed: 2024-11-13}.


\bibitem[Imam and Sarkar(2014)]%
        {ImamS14a}
\bibfield{author}{\bibinfo{person}{Shams~Mahmood Imam} {and}
  \bibinfo{person}{Vivek Sarkar}.} \bibinfo{year}{2014}\natexlab{}.
\newblock \showarticletitle{Savina - An Actor Benchmark Suite: Enabling
  Empirical Evaluation of Actor Libraries}. In
  \bibinfo{booktitle}{\emph{AGERE!@SPLASH}}. \bibinfo{publisher}{{ACM}},
  \bibinfo{pages}{67--80}.
\newblock


\bibitem[Kouzapas et~al\mbox{.}(2016)]%
        {KouzapasYHH16}
\bibfield{author}{\bibinfo{person}{Dimitrios Kouzapas}, \bibinfo{person}{Nobuko
  Yoshida}, \bibinfo{person}{Raymond Hu}, {and} \bibinfo{person}{Kohei Honda}.}
  \bibinfo{year}{2016}\natexlab{}.
\newblock \showarticletitle{On asynchronous eventful session semantics}.
\newblock \bibinfo{journal}{\emph{Math. Struct. Comput. Sci.}}
  \bibinfo{volume}{26}, \bibinfo{number}{2} (\bibinfo{year}{2016}),
  \bibinfo{pages}{303--364}.
\newblock


\bibitem[Lagaillardie et~al\mbox{.}(2022)]%
        {LagaillardieNY22}
\bibfield{author}{\bibinfo{person}{Nicolas Lagaillardie},
  \bibinfo{person}{Rumyana Neykova}, {and} \bibinfo{person}{Nobuko Yoshida}.}
  \bibinfo{year}{2022}\natexlab{}.
\newblock \showarticletitle{Stay Safe Under Panic: Affine Rust Programming with
  Multiparty Session Types}. In \bibinfo{booktitle}{\emph{{ECOOP}}}
  \emph{(\bibinfo{series}{LIPIcs}, Vol.~\bibinfo{volume}{222})}.
  \bibinfo{publisher}{Schloss Dagstuhl - Leibniz-Zentrum f{\"{u}}r Informatik},
  \bibinfo{pages}{4:1--4:29}.
\newblock


\bibitem[Levy et~al\mbox{.}(2003)]%
        {LPT03}
\bibfield{author}{\bibinfo{person}{Paul~Blain Levy}, \bibinfo{person}{John
  Power}, {and} \bibinfo{person}{Hayo Thielecke}.}
  \bibinfo{year}{2003}\natexlab{}.
\newblock \showarticletitle{Modelling environments in call-by-value programming
  languages}.
\newblock \bibinfo{journal}{\emph{Information and Computation}}
  \bibinfo{volume}{185}, \bibinfo{number}{2} (\bibinfo{year}{2003}),
  \bibinfo{pages}{182--210}.
\newblock


\bibitem[Lindley and Cheney(2012)]%
        {LindleyC12}
\bibfield{author}{\bibinfo{person}{Sam Lindley} {and} \bibinfo{person}{James
  Cheney}.} \bibinfo{year}{2012}\natexlab{}.
\newblock \showarticletitle{Row-based effect types for database integration}.
  In \bibinfo{booktitle}{\emph{{TLDI}}}. \bibinfo{publisher}{{ACM}},
  \bibinfo{pages}{91--102}.
\newblock


\bibitem[Lindley and Morris(2015)]%
        {LindleyM15}
\bibfield{author}{\bibinfo{person}{Sam Lindley} {and}
  \bibinfo{person}{J.~Garrett Morris}.} \bibinfo{year}{2015}\natexlab{}.
\newblock \showarticletitle{A Semantics for Propositions as Sessions}. In
  \bibinfo{booktitle}{\emph{{ESOP}}} \emph{(\bibinfo{series}{Lecture Notes in
  Computer Science}, Vol.~\bibinfo{volume}{9032})}.
  \bibinfo{publisher}{Springer}, \bibinfo{pages}{560--584}.
\newblock


\bibitem[Miu et~al\mbox{.}(2021)]%
        {Miu0Y021}
\bibfield{author}{\bibinfo{person}{Anson Miu}, \bibinfo{person}{Francisco
  Ferreira}, \bibinfo{person}{Nobuko Yoshida}, {and} \bibinfo{person}{Fangyi
  Zhou}.} \bibinfo{year}{2021}\natexlab{}.
\newblock \showarticletitle{Communication-safe web programming in {TypeScript}
  with routed multiparty session types}. In \bibinfo{booktitle}{\emph{{CC}}}.
  \bibinfo{publisher}{{ACM}}, \bibinfo{pages}{94--106}.
\newblock


\bibitem[Mostrous and Vasconcelos(2011)]%
        {MostrousV11}
\bibfield{author}{\bibinfo{person}{Dimitris Mostrous} {and}
  \bibinfo{person}{Vasco~Thudichum Vasconcelos}.}
  \bibinfo{year}{2011}\natexlab{}.
\newblock \showarticletitle{Session Typing for a Featherweight {E}rlang}. In
  \bibinfo{booktitle}{\emph{{COORDINATION}}} \emph{(\bibinfo{series}{Lecture
  Notes in Computer Science}, Vol.~\bibinfo{volume}{6721})}.
  \bibinfo{publisher}{Springer}, \bibinfo{pages}{95--109}.
\newblock


\bibitem[Mostrous and Vasconcelos(2018)]%
        {MostrousV18}
\bibfield{author}{\bibinfo{person}{Dimitris Mostrous} {and}
  \bibinfo{person}{Vasco~T. Vasconcelos}.} \bibinfo{year}{2018}\natexlab{}.
\newblock \showarticletitle{Affine Sessions}.
\newblock \bibinfo{journal}{\emph{Log. Methods Comput. Sci.}}
  \bibinfo{volume}{14}, \bibinfo{number}{4} (\bibinfo{year}{2018}).
\newblock


\bibitem[Neykova and Yoshida(2017a)]%
        {NeykovaY17}
\bibfield{author}{\bibinfo{person}{Rumyana Neykova} {and}
  \bibinfo{person}{Nobuko Yoshida}.} \bibinfo{year}{2017}\natexlab{a}.
\newblock \showarticletitle{Let it recover: multiparty protocol-induced
  recovery}. In \bibinfo{booktitle}{\emph{{CC}}}. \bibinfo{publisher}{{ACM}},
  \bibinfo{pages}{98--108}.
\newblock


\bibitem[Neykova and Yoshida(2017b)]%
        {NeykovaY16}
\bibfield{author}{\bibinfo{person}{Rumyana Neykova} {and}
  \bibinfo{person}{Nobuko Yoshida}.} \bibinfo{year}{2017}\natexlab{b}.
\newblock \showarticletitle{Multiparty Session Actors}.
\newblock \bibinfo{journal}{\emph{Log. Methods Comput. Sci.}}
  \bibinfo{volume}{13}, \bibinfo{number}{1} (\bibinfo{year}{2017}).
\newblock


\bibitem[Norell(2008)]%
        {Norell08}
\bibfield{author}{\bibinfo{person}{Ulf Norell}.}
  \bibinfo{year}{2008}\natexlab{}.
\newblock \showarticletitle{Dependently Typed Programming in {A}gda}. In
  \bibinfo{booktitle}{\emph{Advanced Functional Programming}}
  \emph{(\bibinfo{series}{Lecture Notes in Computer Science},
  Vol.~\bibinfo{volume}{5832})}. \bibinfo{publisher}{Springer},
  \bibinfo{pages}{230--266}.
\newblock


\bibitem[Padovani(2017)]%
        {DBLP:journals/jfp/Padovani17}
\bibfield{author}{\bibinfo{person}{Luca Padovani}.}
  \bibinfo{year}{2017}\natexlab{}.
\newblock \showarticletitle{A simple library implementation of binary
  sessions}.
\newblock \bibinfo{journal}{\emph{J. Funct. Program.}}  \bibinfo{volume}{27}
  (\bibinfo{year}{2017}), \bibinfo{pages}{e4}.
\newblock
\href{https://doi.org/10.1017/S0956796816000289}{doi:\nolinkurl{10.1017/S0956796816000289}}


\bibitem[Padovani and Zavattaro(2025)]%
        {DBLP:conf/ecoop/PadovaniZ25}
\bibfield{author}{\bibinfo{person}{Luca Padovani} {and}
  \bibinfo{person}{Gianluigi Zavattaro}.} \bibinfo{year}{2025}\natexlab{}.
\newblock \showarticletitle{Fair Termination of Asynchronous Binary Sessions}.
  In \bibinfo{booktitle}{\emph{39th European Conference on Object-Oriented
  Programming, {ECOOP} 2025, June 30 to July 2, 2025, Bergen, Norway}}
  \emph{(\bibinfo{series}{LIPIcs}, Vol.~\bibinfo{volume}{333})},
  \bibfield{editor}{\bibinfo{person}{Jonathan Aldrich} {and}
  \bibinfo{person}{Alexandra Silva}} (Eds.). \bibinfo{publisher}{Schloss
  Dagstuhl - Leibniz-Zentrum f{\"{u}}r Informatik},
  \bibinfo{pages}{24:1--24:29}.
\newblock
\href{https://doi.org/10.4230/LIPICS.ECOOP.2025.24}{doi:\nolinkurl{10.4230/LIPICS.ECOOP.2025.24}}


\bibitem[Reynolds(2000)]%
        {Reynolds00}
\bibfield{author}{\bibinfo{person}{John~C. Reynolds}.}
  \bibinfo{year}{2000}\natexlab{}.
\newblock \bibinfo{booktitle}{\emph{The Meaning of Types---From Intrinsic to
  Extrinsic Semantics}}.
\newblock \bibinfo{type}{{T}echnical {R}eport} RS-00-32.
  \bibinfo{institution}{BRICS}.
\newblock


\bibitem[Scalas et~al\mbox{.}(2017)]%
        {ScalasDHY17}
\bibfield{author}{\bibinfo{person}{Alceste Scalas}, \bibinfo{person}{Ornela
  Dardha}, \bibinfo{person}{Raymond Hu}, {and} \bibinfo{person}{Nobuko
  Yoshida}.} \bibinfo{year}{2017}\natexlab{}.
\newblock \showarticletitle{A Linear Decomposition of Multiparty Sessions for
  Safe Distributed Programming}. In \bibinfo{booktitle}{\emph{{ECOOP}}}
  \emph{(\bibinfo{series}{LIPIcs}, Vol.~\bibinfo{volume}{74})}.
  \bibinfo{publisher}{Schloss Dagstuhl - Leibniz-Zentrum f{\"{u}}r Informatik},
  \bibinfo{pages}{24:1--24:31}.
\newblock


\bibitem[Scalas and Yoshida(2016)]%
        {DBLP:conf/ecoop/ScalasY16}
\bibfield{author}{\bibinfo{person}{Alceste Scalas} {and}
  \bibinfo{person}{Nobuko Yoshida}.} \bibinfo{year}{2016}\natexlab{}.
\newblock \showarticletitle{Lightweight Session Programming in Scala}. In
  \bibinfo{booktitle}{\emph{30th European Conference on Object-Oriented
  Programming, {ECOOP} 2016, July 18-22, 2016, Rome, Italy}}
  \emph{(\bibinfo{series}{LIPIcs}, Vol.~\bibinfo{volume}{56})},
  \bibfield{editor}{\bibinfo{person}{Shriram Krishnamurthi} {and}
  \bibinfo{person}{Benjamin~S. Lerner}} (Eds.). \bibinfo{publisher}{Schloss
  Dagstuhl - Leibniz-Zentrum f{\"{u}}r Informatik},
  \bibinfo{pages}{21:1--21:28}.
\newblock
\href{https://doi.org/10.4230/LIPICS.ECOOP.2016.21}{doi:\nolinkurl{10.4230/LIPICS.ECOOP.2016.21}}


\bibitem[Scalas and Yoshida(2019)]%
        {ScalasY19}
\bibfield{author}{\bibinfo{person}{Alceste Scalas} {and}
  \bibinfo{person}{Nobuko Yoshida}.} \bibinfo{year}{2019}\natexlab{}.
\newblock \showarticletitle{Less is more: multiparty session types revisited}.
\newblock \bibinfo{journal}{\emph{Proc. {ACM} Program. Lang.}}
  \bibinfo{volume}{3}, \bibinfo{number}{{POPL}} (\bibinfo{year}{2019}),
  \bibinfo{pages}{30:1--30:29}.
\newblock


\bibitem[Scalas et~al\mbox{.}(2019)]%
        {ScalasYB19}
\bibfield{author}{\bibinfo{person}{Alceste Scalas}, \bibinfo{person}{Nobuko
  Yoshida}, {and} \bibinfo{person}{Elias Benussi}.}
  \bibinfo{year}{2019}\natexlab{}.
\newblock \showarticletitle{Verifying message-passing programs with dependent
  behavioural types}. In \bibinfo{booktitle}{\emph{{PLDI}}}.
  \bibinfo{publisher}{{ACM}}, \bibinfo{pages}{502--516}.
\newblock


\bibitem[Thiemann(2023)]%
        {Thiemann23}
\bibfield{author}{\bibinfo{person}{Peter Thiemann}.}
  \bibinfo{year}{2023}\natexlab{}.
\newblock \showarticletitle{Intrinsically Typed Sessions with Callbacks
  (Functional Pearl)}.
\newblock \bibinfo{journal}{\emph{Proc. {ACM} Program. Lang.}}
  \bibinfo{volume}{7}, \bibinfo{number}{{ICFP}} (\bibinfo{year}{2023}),
  \bibinfo{pages}{711--739}.
\newblock


\bibitem[Viering et~al\mbox{.}(2021)]%
        {VieringHEZ21}
\bibfield{author}{\bibinfo{person}{Malte Viering}, \bibinfo{person}{Raymond
  Hu}, \bibinfo{person}{Patrick Eugster}, {and} \bibinfo{person}{Lukasz
  Ziarek}.} \bibinfo{year}{2021}\natexlab{}.
\newblock \showarticletitle{A multiparty session typing discipline for
  fault-tolerant event-driven distributed programming}.
\newblock \bibinfo{journal}{\emph{Proc. {ACM} Program. Lang.}}
  \bibinfo{volume}{5}, \bibinfo{number}{{OOPSLA}} (\bibinfo{year}{2021}),
  \bibinfo{pages}{1--30}.
\newblock


\bibitem[Yoshida et~al\mbox{.}(2013)]%
        {YoshidaHNN13}
\bibfield{author}{\bibinfo{person}{Nobuko Yoshida}, \bibinfo{person}{Raymond
  Hu}, \bibinfo{person}{Rumyana Neykova}, {and} \bibinfo{person}{Nicholas Ng}.}
  \bibinfo{year}{2013}\natexlab{}.
\newblock \showarticletitle{The {S}cribble Protocol Language}. In
  \bibinfo{booktitle}{\emph{{TGC}}} \emph{(\bibinfo{series}{Lecture Notes in
  Computer Science}, Vol.~\bibinfo{volume}{8358})}.
  \bibinfo{publisher}{Springer}, \bibinfo{pages}{22--41}.
\newblock


\bibitem[Zhou et~al\mbox{.}(2020)]%
        {ZhouHNY20}
\bibfield{author}{\bibinfo{person}{Fangyi Zhou}, \bibinfo{person}{Francisco
  Ferreira}, \bibinfo{person}{Raymond Hu}, \bibinfo{person}{Rumyana Neykova},
  {and} \bibinfo{person}{Nobuko Yoshida}.} \bibinfo{year}{2020}\natexlab{}.
\newblock \showarticletitle{Statically verified refinements for multiparty
  protocols}.
\newblock \bibinfo{journal}{\emph{Proc. {ACM} Program. Lang.}}
  \bibinfo{volume}{4}, \bibinfo{number}{{OOPSLA}} (\bibinfo{year}{2020}),
  \bibinfo{pages}{148:1--148:30}.
\newblock


\end{thebibliography}

\clearpage
\clearpage
{
\appendix

    \begin{center}
        {\Huge Appendices}
    \end{center}

     \etoctocstyle{1}{Appendix Contents}
     \etocdepthtag.toc{mtappendix}
     \etocsettagdepth{mtchapter}{none}
     \etocsettagdepth{mtappendix}{subsection}
     \tableofcontents

\newpage

\appendix
\section{Details of Case Study Protocols}\label{ap:case-studies}

In this section we detail the protocols and sequence diagrams for the two case
studies.

\subsection{Robots}

The robots protocol can be found below, both as a Scribble global type and a
sequence diagram. Role \lstinline+R+ stands for Robot, \lstinline+D+ stands for
Door, and \lstinline+W+ stands for Warehouse.

\begin{minipage}{0.6\textwidth}
\begin{lstlisting}[language=scribble]
global protocol Robot(role R, role D, role W) {
  Want(PartNum) from R to D;
  choice at D {
    Busy() from D to R;
    Cancel() from D to W;
  } or {
    GoIn() from D to R;
    Prepare(PartNum) from D to W;
    Inside() from R to D;
    Prepared() from W to D;
    Deliver() from D to W;
    Delivered() from W to R;
    PartTaken() from R to W;
    WantLeave() from R to D;
    GoOut() from D to R;
    Outside() from R to D;
    TableIdle() from W to D;
  }
}
\end{lstlisting}
\end{minipage}
\hfill
\begin{minipage}{0.35\textwidth}
    \includegraphics[width=\textwidth]{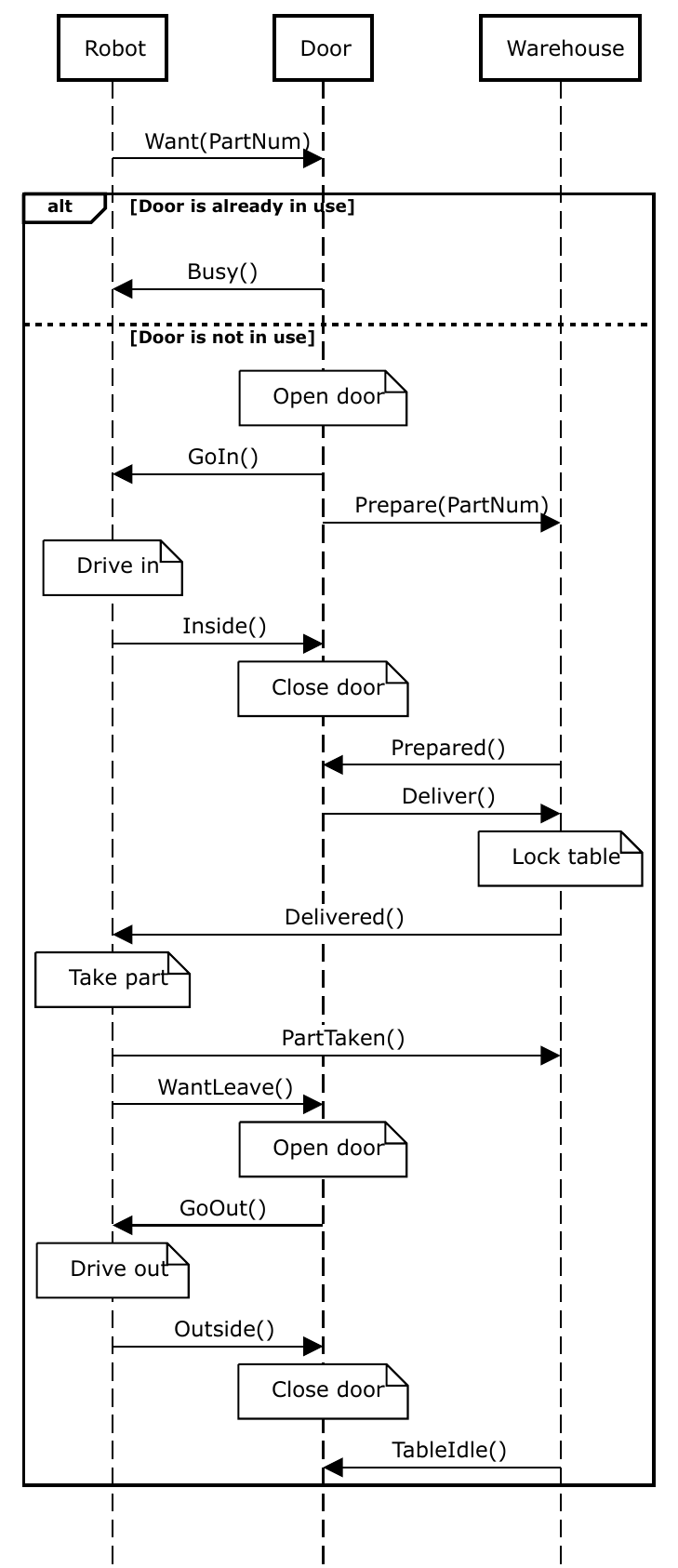}
\end{minipage}

Below is the straightforward user code for a Door actor to
repeatedly register for an unbounded number of Robot sessions.
The Door actor will safely handle all Robot sessions concurrently, coordinated
by its encapsulated state (e.g., \CODE{isBusy}).
The generated \CODE{ActorDoor} API provides a \CODE{register} method for the
formal $\calcwd{register}$ operation, and \CODE{d1Suspend} is a user-defined
handler that registers once more after every session initiation.

{
\begin{lstlisting}[numberstyle=\tiny,numbers=left,columns=flexible]
class Door(pid: Pid, port: Int, apHost: Host, apPort: Int) extends ActorDoor(pid) {
  private var isBusy = false  // Shared state -- n.b. every actor is a single-threaded event loop
  def spawn(): Unit = { super.spawn(this.port); regForInit(new DataD(...)) }
  def regForInit(d: DataD) = register(this.port, apHost, apPort, d, d1Suspend)
  def d1Suspend(d: DataD, s: D1Suspend): Done.type = { regForInit(new DataD(...)); s.suspend(d, d1) }
  ...  // def d1(d: DataD, s: D1): Done.type ... etc.
\end{lstlisting}}
\vspace{-1em}

\subsection{Chat Server}

\begin{minipage}{0.6\textwidth}
    \begin{lstlisting}[language=scribble]
global protocol ChatServer(role C, role S) {
  choice at C {
    LookupRoom(RoomName) from C to S;
    choice at S {
      RoomPort(RoomName, Port) from S to C;
    } or {
      RoomNotFound(RoomName) from S to C;
    }
    do ChatServer(C, S);
  } or {
    CreateRoom(RoomName) from C to S;
    choice at S {
      CreateRoomSuccess(RoomName) from S to C;
    } or {
      RoomExists(RoomName) from S to C;
    }
    do ChatServer(C, S);
  } or {
    ListRooms() from C to S;
    RoomList(StringList) from S to C;
    do ChatServer(C, S);
  } or {
    Bye(String) from C to S;
  }
}

global protocol ChatSessionCtoR(role C, role R) {
  choice at C {
    OutgoingChatMessage(String) from C to R;
    do ChatSessionCtoR(C, R);
  } or {
    LeaveRoom() from C to R;
  }
}

global protocol ChatSessionRtoC(role R, role C){
  choice at R {
    IncomingChatMessage(String) from R to C;
    do ChatSessionRtoC(R, C);
  } or {
    Bye() from R to C;
  }
}
\end{lstlisting}
\end{minipage}
\hfill
\begin{minipage}{0.35\textwidth}
    \includegraphics[width=\textwidth]{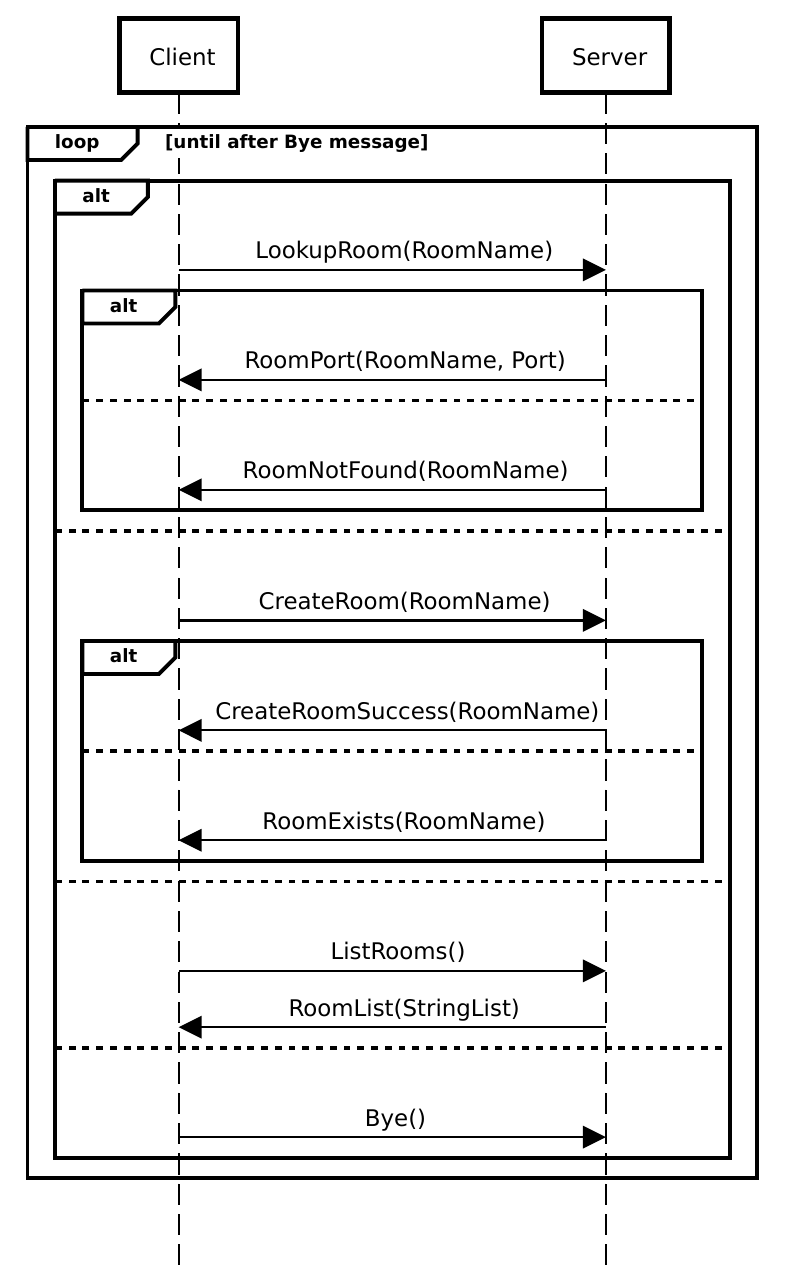}

    \includegraphics[width=\textwidth]{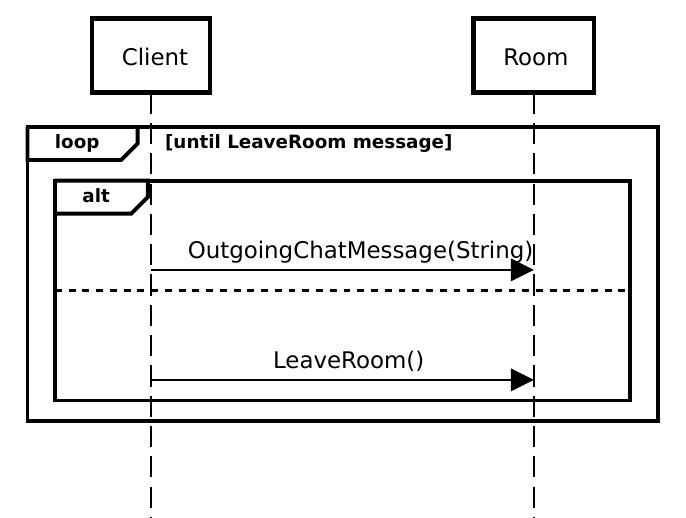}

    \includegraphics[width=\textwidth]{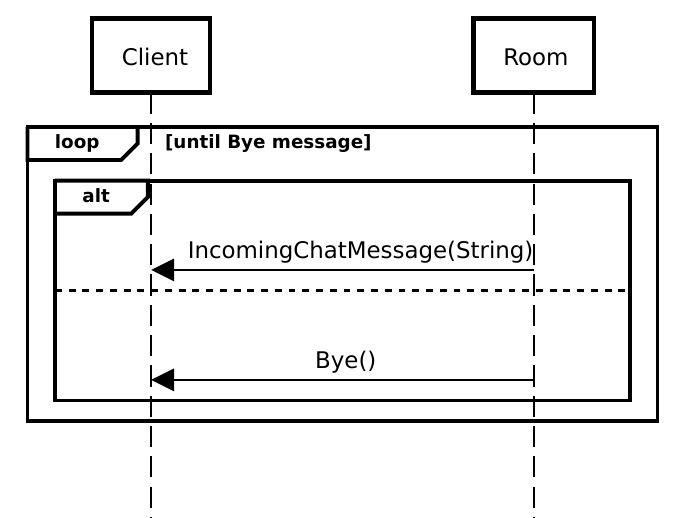}
\end{minipage}

 \clearpage
\section{Supplement to Section~\ref{sec:metatheory}}\label{ap:proofs}

\subsection{Omitted Definitions}

Term reduction $\tma \teval \tmb$ is standard $\beta$-reduction:

\headersig{Term reduction rules}{$\tma \teval \tmb$}
    \[
        \begin{array}{rcl}
            \efflet{x}{\effreturn{\vala}}{\tma} & \teval & \tma \{ \vala / x \} \\
            (\fun{x}{\tma}) \app \vala & \teval & \tma \{ \vala / x \} \\
            (\rec{f}{x}{\tma}) \app \vala & \teval & \tma \{ \rec{f}{x}{\tma} / f, \vala / x \}  \\
            \ite{\ttrue}{\tma}{\tmb} & \teval & \tma \\
            \ite{\ffalse}{\tma}{\tmb} & \teval & \tmb \\
            \quad \ctxe[\tma]  & \teval & \ctxe[\tmb]
            \quad (\text{if } \tma \teval \tmb)
        \end{array}
    \]

\subsection{Preservation}

We begin with some unsurprising auxiliary lemmas.

\begin{lemma}[Substitution]\label{lem:substitution}
    If $\mtseqst{\tyenv, x : \tyb}{\tyc}{\sta}{\tma}{\tya}{\stb}$ and
    $\vseq{\tyenv}{V}{\tyb}$, then
    $\mtseqst{\tyenv}{\sta}{\tyc}{\tma \{ V / x \}}{\tya}{\stb}$.
\end{lemma}
\begin{proof}
    By induction on the derivation of $\mtseqst{\tyenv, x : \tya}{\tyc}{\sta}{\tma}{\tyb}{\stb}$.
\end{proof}

\begin{lemma}[Subterm typability]\label{lem:subterm-typability}
    Suppose $\derivd$ is a derivation of $\mtseqst{\tyenv}{\tyc}{\sta}{\ctxe[\tma]}{\tya}{\stb}$.
    Then there exists some subderivation $\derivd'$ of $\derivd$ concluding
    $\mtseqst{\tyenv}{\tyc}{\sta}{\tma}{\tyb}{\sta'}$ for some type $\tyb$ and session
    type $\sta'$, where the position of $\derivd'$ in $\derivd$ corresponds to
    that of the hole in $\ctxe$.
\end{lemma}
\begin{proof}
    By induction on the structure of $\ctxe$.
\end{proof}

\begin{lemma}[Replacement]\label{lem:replacement}
    If:

    \begin{enumerate}
        \item $\derivd$ is a derivation of
            $\mtseqst{\tyenv}{\tyc}{\sta}{\ctxe[\tma]}{\tya}{\stb}$
        \item $\derivd'$ is a subderivation of $\derivd$ concluding
            $\mtseqst{\tyenv}{\tyc}{\sta}{\tma}{\tyb}{\stb'}$ where the position of
            $\derivd'$ in $\derivd$ corresponds to that of the hole in $\ctxe$
        \item $\mtseqst{\tyenv}{\tyc}{\sta'}{\tmb}{\tyb}{\stb'}$
    \end{enumerate}
    then $\mtseqst{\tyenv}{\tyc}{\sta'}{\ctxe[\tmb]}{\tya}{\stb}$.
\end{lemma}
\begin{proof}
    By induction on the structure of $\ctxe$.
\end{proof}

Since type environments are unrestricted, we also obtain a weakening result.

\begin{lemma}[Weakening]\label{lem:weakening}
    \begin{enumerate}
        \item If $\vseq{\tyenv}{\vala}{\tyb}$ and $x \not\in \dom{\tyenv}$, then
            $\vseq{\tyenv, x : \tya}{\vala}{\tyb}$.
      \item If $\mtseqst{\tyenv}{\tyc}{\sta}{\tma}{\tyb}{\stb}$ and $x \not\in
          \dom{\tyenv}$, then $\mtseqst{\tyenv, x : \tya}{\tyc}{\sta}{\tma}{\tyb}{\stb}$.
      \item If $\hstateseq{\tyenv}{\rtenv}{\tyc}{\hstate}$ and $x \not\in
          \dom{\tyenv}$, then $\hstateseq{\tyenv, x : \tya}{\rtenv}{\tyc}{\hstate}$.
      \item If $\istateseq{\tyenv}{\rtenv}{\tyc}{\istate}$ and $x \not\in
          \dom{\tyenv}$, then $\istateseq{\tyenv, x : \tya}{\rtenv}{\tyc}{\istate}$.
      \item If $\cseq{\tyenv}{\rtenv}{\config{C}}$ and $x \not\in \dom{\tyenv}$, then
          $\cseq{\tyenv, x : \tya}{\rtenv}{\config{C}}$.
    \end{enumerate}
\end{lemma}
\begin{proof}
    By mutual induction on all premises.
\end{proof}

\begin{lemma}[Preservation (terms)]\label{lem:term-pres}
    If $\mtseqst{\tyenv}{\sta}{\tyc}{\tma}{\tya}{\stb}$ and $\tma \teval \tmb$, then
    $\mtseqst{\tyenv}{\sta}{\tyc}{\tmb}{\tya}{\stb}$.
\end{lemma}
\begin{proof}
    A standard induction on the derivation of $\tma \teval \tmb$, noting that
    functional reduction does not modify the session type.
\end{proof}

Next, we introduce some MPST-related lemmas that are helpful for proving
preservation of configuration reduction. We often make use of these lemmas
implicitly.

\begin{lemma}\label{lem:split-safe}
    If $\safe{\rtenv, \rtenv'}$, then $\safe{\rtenv}$.
\end{lemma}
\begin{proof}[Proof sketch]
    Splitting a context only removes potential reductions.
    Only by adding reductions could we violate safety.
\end{proof}

\begin{lemma}\label{lem:split-safety}
    If $\safe{\rtenv_1, \rtenv_2}$ and $\rtenv_1 \equivsynceval \rtenv'_1$, then
    $\safe{\rtenv'_1, \rtenv_2}$.
\end{lemma}
\begin{proof}
    By induction on the derivation of $\rtenv_1 \equiv \synceval{\synclbl}
    \equiv \rtenv'_1$.

    It suffices to consider the cases where reduction could potentially make the
    combined environments unsafe.

    In the case of \textsc{Lbl-Sync-Send}, the
resulting reduction adds a message $\qentry{\prole}{\qrole}{\ell_i}{\tya_i}$ to a
    queue $Q$.

    The only way this could violate safety is if there were some entry
    $\roleidx{s}{\qrole} : \localoffer{\prole}{\msg{\ell_i}{\tya_i}\then{\sta_i}}_{i \in I}$,
    and $Q \equiv \qentry{\prole}{\qrole}{\ell_j}{\tya_j} \cdot Q'$ where $j \in I$,
    but $(Q \cdot \qentry{\prole}{\qrole}{\ell_k}{\tya_k} \equiv
    \qentry{\prole}{\qrole}{\ell_k}{\tya_k} \cdot Q''$ with $k \not\in I$.
However, this is impossible since it is not possible to permute this message
    ahead of the existing message because of the side-conditions on queue
    equivalence.

    A similar argument applies for \textsc{Lbl-Sync-Recv}.
\end{proof}

\begin{lemma}\label{lem:queue}
    If $\cseq{\tyenv}{\rtenv, s : Q}{\qproc{s}{\sigma}}$ and
    $\vseq{\tyenv}{V}{\tya}$, then
    $\cseq
        {\tyenv}
        {\rtenv, s : (Q \cdot\qentry{p}{q}{\ell}{\tya})}
        {\qproc{s}{\sigma \cdot \qentry{p}{q}{\ell}{\vala}}}$
\end{lemma}
\begin{proof}
    A straightforward induction on the derivation of
    $\cseq{\tyenv}{\rtenv, s : Q}{\qproc{s}{\sigma}}$.
\end{proof}

\begin{lemma}\label{lem:separate-envs-safe}
    If $\safe{\rtenv_1}$ and $\safe{\rtenv_2}$ for environments $\rtenv_1,
    \rtenv_2$ such that $\snames{\rtenv_1} \cap \snames{\rtenv_2} = \emptyset$
    and where $\rtenv_1, \rtenv_2$ is defined, then $\safe{\rtenv_1, \rtenv_2}$.
\end{lemma}
\begin{proof}
    By inspection of the definition of $\safe{-}$ and the environment reduction
    rules, noting that each are only defined on a single session.
\end{proof}

\begin{lemma}\label{lem:env-reductions-separate}
    Given environments $\rtenv_1, \rtenv_2$ such that $\safe{\rtenv_1,
    \rtenv_2}$ and $\snames{\rtenv_1} \cap \snames{\rtenv_2} = \emptyset$
    and $\rtenv_1, \rtenv_2 \equivsynceval^? \rtenv'$ such that
    $\safe{\rtenv'}$, either:

    \begin{enumerate}
        \item $\rtenv' = \rtenv$; or
        \item $\rtenv' = \rtenv'_1, \rtenv_2$ such that $\rtenv_1 \equivsynceval
            \rtenv'_1$ and $\safe{\rtenv'_1}$; or
        \item $\rtenv' = \rtenv_1, \rtenv'_2$ such that $\rtenv_2 \equivsynceval
            \rtenv'_2$ and $\safe{\rtenv'_2}$.
    \end{enumerate}
\end{lemma}
\begin{proof}
    By inspection of the reduction rules for $\equivsynceval$, noting that
    reduction only affects a single session and that the session names in
    $\rtenv_1, \rtenv_2$ are disjoint.
\end{proof}

\begin{lemma}[Preservation (Equivalence)]\label{lem:equiv-pres}
    If $\cseq{\tyenv}{\rtenv}{\config{C}}$ and $\config{C} \equiv \config{D}$
    then there exists some $\rtenv' \equiv \rtenv$ such that $\cseq{\tyenv}{\rtenv'}{\config{D}}$.
\end{lemma}
\begin{proof}
    By induction on the derivation of $\config{C} \equiv \config{D}$. The only
    case that causes the type environment to change is queue message reordering,
    which can be made typable by mirroring the change in the queue type.
\end{proof}

\begin{lemma}[Preservation (Configuration reduction)]\label{lem:ceval-pres}
If $\cseq{\tyenv}{\rtenv}{\config{C}}$ with $\safe{\rtenv}$
            and $\config{C} \ceval \config{D}$,
            then there exists some $\rtenv'$ such that $\rtenv \equivsynceval^?
            \rtenv'$ such that $\safe{\rtenv'}$
            and $\cseq{\tyenv}{\rtenv'}{\config{D}}$.
\end{lemma}
\begin{proof}
    By induction on the derivation of $\config{C} \ceval \config{D}$.
    In each case where $\rtenv \equivsynceval \rtenv'$ for some $\rtenv'$, by
    the definition of safety it follows that $\safe{\rtenv'}$.

    \begin{proofcase}{E-Send}

        \begin{mathpar}
            \inferrule
            {}
            {
                \actor
                { \sessthread{s}{p}{\ctxe[\send{q}{\ell}{\vala}]} }
                { \hstate }
                { \istate }
                \parallel
                \qproc{s}{\qcontents}
\ceval
\actor
                { \sessthread{s}{p}{\ctxe[\effreturn{()}]} }
                { \hstate }
                { \istate }
                \parallel
                \qproc{s}{\qcontents \cdot \qentry{p}{q}{\ell}{\vala}}
            }
        \end{mathpar}

        Assumption:
        {\small
        \begin{mathpar}
            \inferrule*
            {
                \inferrule*
                {
                    \inferrule*
                    {
                        \mtseqst
                            {\tyenv}
                            {\tyc}
                            {\sta}
                            {\ctxe[\send{q}{\ell}{\vala}]}
                            {\tyc}
                            {\localend}
                    }
                    {
                        \threadseq
                            {\tyenv}
                            {\roleidx{s}{\prole} : \sta}
                            {\tyc}
                            {\sessthread{s}{p}{\ctxe[\send{q}{\ell}{\vala}]}}
                    }
                    \\
                    \hstateseq{\tyenv}{\rtenv_2}{\tyc}{\hstate} \\
                    \istateseq{\tyenv}{\rtenv_3}{\tyc}{\istate} \\
                }
                { \cseq
                    {\tyenv}
                    {\roleidx{s}{\prole} : \sta, \rtenv_2, \rtenv_3, a}
                    {
                        \actor
                            { \sessthread{s}{p}{\ctxe[\send{q}{\ell}{\vala}]} }
                            { \hstate }
                            { \istate }
                    }
                }
                \\
                \cseq{\tyenv}{s : Q}{\qproc{s}{\qcontents}}
            }
            {
                \cseq
                {\tyenv}
                { \roleidx{s}{\prole} : S, \rtenv_2, \rtenv_3, s : Q, a }
                {
                    \actor
                        { \sessthread{s}{p}{\ctxe[\send{q}{\ell}{\vala}]} }
                        { \hstate }
                        { \istate }
                    \parallel
                        \qproc{s}{\qcontents}
                }
            }
        \end{mathpar}
    }

        By Lemma~\ref{lem:subterm-typability} we have that
        $\mtseqst
            {\tyenv}
            {\tyc}
            {\localselect{\qrole}{\msg{\ell_i}{\tya_i} : \stb_i}_{i \in I}}
            {\send{q}{\ell_j}{\vala}}
            {\one}
            {\stb_j}$
        and therefore that
        $\sta = {\localselect{\qrole}{\msg{\ell_i}{\tya_i} : \stb_i}_{i \in I}}$.

        Since
        $\mtseqst
            {\tyenv}
            {\tyc}
            {\stb_j}
            {\effreturn{()}}
            {\one}
            {\stb_j}$,
        we can show
        by Lemma~\ref{lem:replacement} we have that
        $\mtseqst
            {\tyenv}
            {\tyc}
            {\stb_j}
            {\ctxe[\effreturn{()}]}
            {\tyc}
            {\localend}$.

            By Lemma~\ref{lem:queue},
            $\cseq
                {\tyenv}
                {s : Q \cdot \qentry{p}{q}{\ell_j}{\tya_j}}
                {\qproc{s}{\qcontents \cdot \qentry{p}{q}{\ell_j}{\vala}}}$.

        Therefore, recomposing:

        {\footnotesize
        \begin{mathpar}
            \inferrule*
            {
                \inferrule*
                {
                    \inferrule*
                    {
                        \mtseqst{\tyenv}{\tyc}{\stb_j}{\ctxe[\effreturn{()}]}{\tyc}{\localend}
                    }
                    {
                        \threadseq
                            {\tyenv}
                            {\roleidx{s}{\prole} : \stb_j}
                            {\tyc}
                            {\sessthread{s}{p}{\ctxe[\effreturn{()}]}}
                    }
                    \\
                    {\bl
                    \hstateseq{\tyenv}{\rtenv_2}{\tyc}{\hstate} \\
                    \istateseq{\tyenv}{\rtenv_3}{\tyc}{\istate} \\
                    \el
                    }
                }
                { \cseq
                    {\tyenv}
                    {\roleidx{s}{\prole} : \stb_j, \rtenv_2, \rtenv_3, a}
                    {
                        \actor
                            { \sessthread{s}{p}{\ctxe[\effreturn{()}]} }
                            { \hstate }
                            { \istate }
                    }
                }
                \\
                \cseq
                    {\tyenv}
                    {s : Q \cdot \qentry{\prole}{\qrole}{\ell_j}{\tya_j}}
                    {\qproc{s}{\qcontents \cdot \qentry{\prole}{\qrole}{\ell_j}{\vala}}}
            }
            {
                \cseq
                {\tyenv}
                { \roleidx{s}{\prole} : \stb_j, \rtenv_2, \rtenv_3, s : Q \cdot
                \qentry{\prole}{\qrole}{\ell_j}{\tyb_j}, a }
                {
                    \actor
                        { \sessthread{s}{p}{\ctxe[\effreturn{()}]} }
                        { \hstate }
                        { \istate }
                    \parallel
                    \qproc
                        {s}
                        {\qcontents \cdot \qentry{\prole}{\qrole}{\ell_j}{\vala}}
                }
            }
        \end{mathpar}
    }

    Finally,

       $ \roleidx{s}{\prole} : \localselect{\qrole}{\msg{\ell_i}{\tya_i} :
       \stb_i}_{i \in I}, \rtenv_2, \rtenv_3, s : Q, a \equivsynceval
         \roleidx{s}{\prole} : \stb_j, \rtenv_2, \rtenv_3, s : Q \cdot
         \qentry{\prole}{\qrole}{\ell_j}{\tyb_j}, a$ by \textsc{Lbl-Send} as
         required.

\end{proofcase}

\begin{proofcase}{E-React}
    \begin{mathpar}
        \inferrule
        {
            \ell(x) \mapsto \tma \in \seq{H}
        }
        {
            \actor
            { \idle{\valc} }
            { \hstate[\storedhandler{s}{p}{\handler{q}{\var{st}}{\seq{H}}}] }
            { \istate }
            \parallel
            \qproc{s}{\qentry{q}{p}{\ell}{\vala} \cdot \qcontents}
\ceval
\actor
            { \sessthread{s}{\prole}{\tma \{ V / x, \valc / \var{st} \}} }
            { \hstate }
            { \istate }
            \parallel
            \qproc{s}{\qcontents}
        }
    \end{mathpar}

    For simplicity (and equivalently) let us refer to $\ell$ as $\ell_j$.

    Let $\derivd$ be the following derivation:

    {\footnotesize
    \begin{mathpar}
        \inferrule*
        {
            \inferrule*
            {
                \inferrule*
                {
                    (\mtseqst
                        {\tyenv, x_i : \tyb_i, \var{st} : \tyc}
                        {\tyc}
                        {\sta_i}
                        {\tma_i}
                        {\tyc}
                        {\localend})_{i \in I}
                }
                {
                    \vseq
                        {\tyenv }
                        { \handler{q}{\var{st}}{(\msg{\ell_i}{x_i} \mapsto \tma_i)_{i \in I}}}
                        { \handlerty{\stin}{\tyc} }
                }
                \\
                \hstateseq{\tyenv}{\rtenv_2}{\tyc}{\hstate}
            }
            { \hstateseq
                {\tyenv}
                {\rtenv_2, \roleidx{s}{p} : \stin }
                {\tyc}
                { \hstate[\storedhandler{s}{p}{\handler{q}{\var{st}}{\seq{H}}}] }
            }  \\
            \inferrule*
            { \vseq{\tyenv}{\valc}{\tyc} }
            { \threadseq{\tyenv}{\tyc}{\cdot}{\idle{\valc}} }
            \\
            \istateseq{\tyenv}{\rtenv_3}{\tyc}{\istate}
        }
        {
            \cseq
            { \tyenv }
            { \rtenv_2, \rtenv_3, \roleidx{s}{p} : \stin, a }
            {
                \actor
                    { \idle{\valc} }
                    { \hstate[\storedhandler{s}{p}{\handler{q}{\var{st}}{\seq{H}}}] }
                    { \istate }
            }
        }
    \end{mathpar}
    }

    Assumption:

    {\small
    \begin{mathpar}
        \inferrule*
        {
            \derivd
            \\
            \inferrule*
            {
                \vseq{\tyenv}{\vala}{\tya}
                \\
                \cseq
                { \tyenv }
                { s : Q }
                { \qproc{s}{\qcontents} }
            }
            {
                \cseq
                { \tyenv }
                { \roleidx{s}{p} : \stin, s : (\qentry{q}{p}{\ell_j}{\tya} \cdot Q) }
                { \qproc{s}{\qentry{q}{p}{\ell_j}{\vala} \cdot \qcontents} }
            }
        }
        {
            \cseq
            { \tyenv }
            { \rtenv_2, \rtenv_3, \roleidx{s}{p} : \stin, s :
            (\qentry{q}{p}{\ell_j}{\tya} \cdot Q), a }
            {
                \actor
                    { \idle{\valc} }
                    { \hstate[\storedhandler{s}{p}{\handler{q}{\var{st}}{\seq{H}}}] }
                    { \istate }
                    \parallel
                    \qproc{s}{\qentry{q}{p}{\ell_j}{\vala} \cdot \qcontents}
            }
        }
    \end{mathpar}
    }

    where $\stin = \localoffer{\prole}{\msg{\ell_i}{\tyb_i}.\sta_i}_{i \in I}$.

    Since
    $\safe{\rtenv_2, \rtenv_3, \roleidx{s}{p} : \stin, s :
    (\qentry{q}{p}{\ell_j}{\tya} \cdot Q)}, a$ we have that $j \in I$ and $\tya = \tyb_j$.

    Similarly since $\msg{\ell_j}{x_j} \mapsto \tma \in \seq{H}$ we have that
    $\mtseqst{\tyenv, x_j : \tyb_j, \var{st} : \tyc}{\tyc}{\sta_j}{\tma}{\tyc}{\localend}$.

    By Lemma~\ref{lem:substitution}, $\mtseqst{\tyenv}{\tyc}{\sta_j}{\tma \{
    \vala / x_j, \valc / \var{st} \}}{\tyc}{\localend}$.

    Let $\derivd'$ be the following derivation:

    \begin{mathpar}
        \inferrule*
        {
            \inferrule*
            {
                \mtseqst{\tyenv}{\sta_j}{\tyc}{\tma \{ \vala / x_j, \valc /
                \var{st} \}}{\tyc}{\localend}
            }
            { \threadseq
                {\tyenv}
                {\roleidx{s}{\prole} : \sta_j}
                {\tyc}
                { \sessthread{s}{p}{\tma \{ \vala / x_j, \valc / \var{st} \}} }
            } \\
            \hstateseq{\tyenv}{\rtenv_2}{\tyc}{\hstate} \\
            \istateseq{\tyenv}{\rtenv_3}{\tyc}{\istate}
        }
        {
            \cseq
            { \tyenv }
            { \rtenv_2, \rtenv_3, \roleidx{s}{p} : \sta_j, a }
            {
                \actor
                    { \sessthread{s}{p}{\tma \{ \vala / x_j, \valc / \var{st} \}} }
                    { \hstate }
                    { \istate }
            }
        }
    \end{mathpar}

    Recomposing:
    \begin{mathpar}
        \inferrule*
        {
            \derivd'
            \\
            \cseq{\tyenv}{s : \qty}{\qproc{s}{\qcontents}}
        }
        {
            \cseq
            { \tyenv }
            { \rtenv_2, \rtenv_3, \roleidx{s}{p} : \sta_j, s : \qty, a }
            {
                \actor
                    { \sessthread{s}{\prole}{\tma \{ V / x_j, \valc / \var{st} \}} }
                    { \hstate }
                    { \istate }
                \parallel
                    \qproc{s}{\qcontents}
            }
        }
    \end{mathpar}

    Finally, we note that
$\rtenv_2, \rtenv_3, \roleidx{s}{p} : \stin, s :
    (\qentry{q}{p}{\ell_j}{\tya} \cdot Q), a
    \equivsynceval
    \rtenv_2, \rtenv_3, \roleidx{s}{p} : \sta_j, s : \qty, a$ by \textsc{Lbl-Recv}
    as required.

\end{proofcase}

\begin{proofcase}{E-Suspend}

    \begin{mathpar}
       \inferrule
        {}
        { \actor
            {\sessthread{s}{\prole}{\ctxe[\suspend{\vala}{\valb}]}}
            {\hstate}
            {\istate}
\ceval
\actor
            {\idle{\valb}}
            {\hstate[\storedhandler{s}{\prole}{\vala}]}
            {\istate}
        }
    \end{mathpar}

    Assumption:

    \begin{mathpar}
        \inferrule
        {
            \inferrule
            { \mtseqst{\tyenv}{\tyc}{\sta}{\ctxe[\suspend{\vala}{\valb}]}{\tyc}{\localend} }
            {  \threadseq
                    {\tyenv}
                    {\roleidx{s}{\prole} : \sta}
                    {\tyc}
                    {\sessthread{s}{\prole}{\ctxe[\suspend{\vala}{\valb}]}}
            }
            \\
            \hstateseq{\tyenv}{\rtenv_2}{\tyc}{\hstate} \\
            \istateseq{\tyenv}{\rtenv_3}{\tyc}{\istate} \\
        }
        {
            \cseq
            {
                \tyenv
            }
            {
                \roleidx{s}{\prole} : \sta, \rtenv_2, \rtenv_3, a
            }
            {
                \actor
                {\sessthread{s}{\prole}{\ctxe[\suspend{\vala}{\valb}]}}
                {\hstate}
                {\istate}
            }
        }
    \end{mathpar}

    By Lemma~\ref{lem:subterm-typability} we have that:
    \begin{mathpar}
        \inferrule
        {
            \vseq{\tyenv}{\vala}{\handlerty{\stin}{\tyc}} \\
            \vseq{\tyenv}{\valb}{\tyc}
        }
        { \mtseqst{\tyenv}{\tyc}{\stin}{\suspend{\vala}{\valb}}{\tya}{\stb} }
    \end{mathpar}

    for any arbitrary $\tya, \stb$, and showing that $\sta = \stin$.

    Recomposing:

    {\small
    \begin{mathpar}
        \inferrule
        {
            \inferrule*
            { \vseq{\tyenv}{\valb}{\tyc} }
            { \threadseq{\tyenv}{\cdot}{\tyc}{\idle{\valb}} }
            \\
            \inferrule*
            { \vseq{\tyenv}{\vala}{\handlerty{\stin}{\tyc}} \\
            \hstateseq{\tyenv}{\rtenv_2}{\tyc}{\hstate} }
            { \hstateseq
                {\tyenv}
                {\rtenv_2, \roleidx{s}{\prole} : \stin}
                {\tyc}
                {\hstate[\roleidx{s}{\prole} \mapsto \vala]} } \\
                \istateseq{\tyenv}{\rtenv_3}{\tyc}{\istate}
        }
        {
            \cseq
            {
                \tyenv
            }
            {
                \roleidx{s}{\prole} : \stin, \rtenv_2, \rtenv_3, a
            }
            {
                \actor
                {\idle{\valb}}
                {\hstate[\roleidx{s}{\prole} \mapsto \vala]}
                {\istate}
            }
        }
    \end{mathpar}
    }

    as required.

    \end{proofcase}

\begin{proofcase}{E-Spawn}

    \begin{mathpar}
        \inferrule
        {}
        { \actor
            { \threadctx[\spawn{M}] }
            { \hstate }
            { \istate }
\ceval
(\nu b)(
          \actor
            { \threadctx[\effreturn{()}] }
            { \hstate }
            { \istate }
          \parallel
          \actor
            [b]
            { M }
            { \epsilon }
            { \epsilon }
            )
        }
    \end{mathpar}
    (with $b$ fresh)

    There are two subcases based on whether $\threadctx = \ctxe[-]$ or
    $\threadctx = \sessthread{s}{\prole}{\ctxe[-]}$.
    Both are similar so we will prove the latter case.

    Assumption:

    \begin{mathpar}
        \inferrule*
        {
            \inferrule*
            { \mtseqst{\tyenv}{\tyc}{\sta}{\ctxe[\spawn{\tma}]}{\tyc}{\localend} }
            { \threadseq
                {\tyenv}
                {\roleidx{s}{\prole} : \sta}
                {\tyc}
                {\sessthread{s}{\prole}{\ctxe[\spawn{\tma}]} }
            }
            \\
            {
            \begin{array}{l}
            \hstateseq{\tyenv}{\rtenv_2}{\tyc}{\hstate} \\
            \istateseq{\tyenv}{\rtenv_3}{\tyc}{\istate} \\
            \end{array}
            }
        }
        {
            \cseq
            { \tyenv }
            { \rtenv_2, \rtenv_3, \roleidx{s}{\prole} : \sta, a }
            {
                \actor
                { \sessthread{s}{\prole}{\ctxe[\spawn{M}]} }
                { \hstate }
                { \istate }
            }
        }
    \end{mathpar}

    By Lemma~\ref{lem:subterm-typability}:

    \begin{mathpar}
        \inferrule
        { \mtseqst{\tyenv}{\tya}{\localend}{\tma}{\tya}{\localend} \\
        }
        { \mtseqst{\tyenv}{\tyc}{\sta}{\spawn{\tma}}{\one}{\sta} }
    \end{mathpar}

    By Lemma~\ref{lem:replacement},
    $\mtseqst{\tyenv}{\tyc}{\sta}{\ctxe[\effreturn{()}]}{\tyc}{\localend}$.

    Thus, recomposing:

    {\footnotesize
    \begin{mathpar}
        \inferrule*
        {
            \inferrule*
            {
                \inferrule*
                {
                    \inferrule*
                    {
                        \mtseqst{\tyenv}{\tyc}{\sta}{\ctxe[\effreturn{()}]}{\tyc}{\localend}
                    }
                    { \threadseq{\tyenv}{\roleidx{s}{\prole} : \sta}{\tyc}{\sessthread{s}{\prole}{\ctxe[\effreturn{()}]}} }
                    \\
                    {
                        \begin{array}{l}
                            \hstateseq{\tyenv}{\rtenv_2}{\tyc}{\hstate} \\
                            \istateseq{\tyenv}{\rtenv_3}{\tyc}{\istate} \\
                        \end{array}
                    }
                }
                {
                    \cseq
                    { \tyenv }
                    { \rtenv_2, \rtenv_3, \roleidx{s}{\prole} : \sta, a }
                    {
                        \actor
                            { \sessthread{s}{\prole}{\ctxe[\effreturn{()}]} }
                            { \hstate }
                            { \istate }
                    }
                }
                \\
                \inferrule*
                {
                    \inferrule*
                    {
                        \mtseqst{\tyenv}{\tya}{\localend}{\tma}{\tya}{\localend}
                    }
                    { \threadseq{\tyenv}{\cdot}{\tya}{\tma} }
                    \\
                    {
                        \begin{array}{l}
                            \hstateseq{\tyenv}{\cdot}{\tya}{\epsilon} \\
                            \istateseq{\tyenv}{\cdot}{\tya}{\epsilon} \\
                        \end{array}
                    }
                }
                {
                    \cseq
                    { \tyenv }
                    { b }
                    { \actor[b]{\tma}{\epsilon}{\epsilon} }
                }
            }
            {
                \cseq
                { \tyenv }
                { \rtenv_2, \rtenv_3, \roleidx{s}{\prole} : \sta, a, b }
                {
                    \actor
                    { \sessthread{s}{\prole}{\ctxe[\effreturn{()}]} }
                    { \hstate }
                    { \istate }
                    \parallel
                    \actor
                    [b]
                    { \tma }
                    { \epsilon }
                    { \epsilon }
                }
            }
        }
        {
            \cseq
                { \tyenv }
                { \rtenv_2, \rtenv_3, \roleidx{s}{\prole} : \sta, a }
                {
                    (\nu b)(
                    \actor
                    { \sessthread{s}{\prole}{\ctxe[\effreturn{()}]} }
                    { \hstate }
                    { \istate }
                    \parallel
                    \actor
                    [b]
                    { \tma }
                    { \epsilon }
                    { \epsilon }
                    )
                }
        }
    \end{mathpar}
    }

    as required.
    \end{proofcase}

\begin{proofcase}{E-Reset}
    \begin{mathpar}
        \inferrule
        {}
        {
            \actor
              { \tlthreadctx[\effreturn{\vala}] }
              { \hstate }
              { \istate }
\ceval
\actor
              { \idle{\vala} }
              { \hstate }
              { \istate }
        }
    \end{mathpar}

    There are two subcases based on whether $\tlthreadctx = [-]$ or $\tlthreadctx =
    \sessthread{s}{p}{[-]}$. We prove the latter case; the former is similar but
    does not require a context reduction.

    Assumption:

    \begin{mathpar}
        \inferrule
        {
            \inferrule*
            {
                \inferrule*
                {
                    \vseq{\tyenv}{\vala}{\tyc}
                }
                {
                    \mtseqst
                        {\tyenv}
                        { \tyc }
                        {\localend}
                        {\effreturn{\vala}}
                        {\tyc}
                        {\localend}
                }
            }
            {
                \threadseq
                    {\tyenv}
                    {\roleidx{s}{\prole} : \localend}
                    { \tyc }
                    {\sessthread{s}{p}{\effreturn{\vala}}}
            }
            \\
            {\bl
                \hstateseq{\tyenv}{\rtenv_2}{\tyc}{\hstate} \\
                \istateseq{\tyenv}{\rtenv_3}{\tyc}{\istate}
            \el
            }
        }
        {
            \cseq
            { \tyenv }
            { \rtenv_2, \rtenv_3, \roleidx{s}{\prole} : \localend, a}
            {
                \actor
                { \sessthread{s}{p}{\effreturn{\vala}}}
                { \hstate }
                { \istate }
            }
        }
    \end{mathpar}

    We can show that
    $\rtenv_2, \rtenv_3, \roleidx{s}{\prole} : \localend, a
    \synceval{\lblsyncend{s}{\prole}} \rtenv_2, \rtenv_3, a$, so we can
    reconstruct:

    \begin{mathpar}
        \inferrule
        {
            \inferrule*
            { \vseq{\tyenv}{\vala}{\tyc} }
            { \threadseq{\tyenv}{\cdot}{\tyc}{\idle{\vala}} }  \\
            \hstateseq{\tyenv}{\rtenv_2}{\tyc}{\hstate} \\
            \istateseq{\tyenv}{\rtenv_3}{\tyc}{\istate}
        }
        {
            \cseq
            { \tyenv }
            { \rtenv_2, \rtenv_3, a }
            {
                \actor
                { \idle{\vala} }
                { \hstate }
                { \istate }
            }
        }
    \end{mathpar}

    as required.
    \end{proofcase}

\begin{proofcase}{E-NewAP}
    \begin{mathpar}
        \inferrule
        { \apname \text{ fresh}}
        {
          \actor
            { \threadctx[\newap{(\role{p}_i : \sta_i )_{i \in I}}] }
            { \hstate }
            { \istate }
\ceval
(\nu \apname)(
              \actor
                { \threadctx[\effreturn{\apname}] }
                { \hstate }
                { \istate }
\parallel
\ap{\apname}{(\prole_i \mapsto \epsilon)_{i \in I}}
            )
        }
    \end{mathpar}

    As usual we prove the case where $\threadctx = \sessthread{s}{\prole}{\ctxe[-]}$; the
    case where $\threadctx = (\ctxe[-])$ is similar.

    Assumption:
    \begin{mathpar}
        \inferrule*
        {
            \inferrule*
            {
                \inferrule*
                {}
                { \mtseqst
                    {\tyenv}
                    { \tyc }
                    {\stb}
                    {\ctxe[\newap{(\role{p}_i : \sta_i )_{i \in I}}]}
                    {\tyc}
                    {\localend}
                }
            }
            { \threadseq
                { \tyenv }
                { \roleidx{s}{\prole} : \stb }
                {\tyc}
                {\sessthread{s}{p}{\ctxe[\newap{(\role{p}_i : \sta_i )_{i \in I}}]}} }
            \\
            {
                \begin{array}{l}
                    \hstateseq{\tyenv}{\rtenv_2}{\tyc}{\hstate} \\
                    \istateseq{\tyenv}{\rtenv_3}{\tyc}{\istate}
                \end{array}
            }
        }
        {
            \cseq
            { \tyenv }
            { \rtenv_2, \rtenv_3, a }
            {
                \actor
                    { \sessthread{s}{p}{\ctxe[\newap{(\role{p}_i : \sta_i )_{i \in I}}]} }
                    { \hstate }
                    { \istate }
            }
        }
    \end{mathpar}

    By Lemma~\ref{lem:subterm-typability}:
    \begin{mathpar}
        \inferrule
        {
            \comp{(\role{p}_i : \sta_i)_{i \in I}}
        }
        { \mtseqst
            {\tyenv}
            { \tyc }
            { \stb }
            {\newap{(\role{p}_i : \sta_i )_{i \in I}}}
            {\apty{(\role{p}_i : \sta_i )_{i \in I}}}
            { \stb }
        }
    \end{mathpar}

    By Lemma~\ref{lem:replacement},
    $\mtseqst
        {\tyenv, \apname : \apty{(\role{p}_i : \sta_i )_{i \in I}}}
        { \tyc }
        {\stb}
        {\ctxe[\effreturn{\apname}]}
        {\tyc}
        {\localend}$.

    Let $\tyenv' = \tyenv, \apname : \apty{(\role{p}_i : \sta_i )_{i \in I}}$.

    By Lemma~\ref{lem:weakening}, since $\apname$ is fresh we have that
    $\hstateseq{\tyenv'}{\rtenv_2}{\tyc}{\hstate}$ and
    $\istateseq{\tyenv'}{\rtenv_3}{\tyc}{\istate}$.

    Recomposing:

    {\footnotesize
    \begin{mathpar}
        \inferrule*
        {
            \inferrule*
            {
                \inferrule*
                {
                    \inferrule*
                    {
                        \inferrule*
                        {}
                        { \mtseqst
                            {\tyenv'}
                            {\tyc}
                            {\stb}
                            {\ctxe[\effreturn{\apname}]}
                            {\tyc}
                            {\localend}
                        }
                    }
                    { \threadseq
                        { \tyenv' }
                        { \roleidx{s}{\prole} : \stb }
                        {\tyc}
                        {\sessthread{s}{p}{\ctxe[\effreturn{\apname}]}} }
                    \\
                    {
                        \begin{array}{l}
                            \hstateseq{\tyenv'}{\rtenv_2}{\tyc}{\hstate} \\
                            \istateseq{\tyenv'}{\rtenv_3}{\tyc}{\istate}
                        \end{array}
                    }
                }
                {
                    \cseq
                    { \tyenv' }
                    { \rtenv_2, \rtenv_3, \roleidx{s}{p} : \stb, a }
                    {
                        \actor
                            { \sessthread{s}{p}{\ctxe[\effreturn{\apname}]} }
                            { \hstate }
                            { \istate }
                    }
                }
                \\
                \inferrule
                {
                    \apname : \apty{(\prole_i: S_i)_{i \in I}} \in \tyenv' \\
                    (\inittokseq{\cdot}{\epsilon}{\sta_i})_{i \in I} \\\\
                    \comp{(\prole_i: \sta_i)_{i \in I}} \\
                }
                { \cseq
                    {\tyenv'}
                    {\apname}
                    {\ap{\apname}{(\maptwo{\prole_i}{\epsilon})_{i \in I}}}
                }
            }
            {
                \cseq
                { \tyenv' }
                { \rtenv_2, \rtenv_3, \roleidx{s}{p} : \stb, a, \apname }
                {
                    \actor
                        { \sessthread{s}{p}{\ctxe[\effreturn{\apname}] } }
                        { \hstate }
                        { \istate }
                    \parallel
                        \ap{\apname}{(\prole_i \mapsto \epsilon)_{i \in I}}
                }
            }
        }
        {
            \cseq
            { \tyenv }
            { \rtenv_2, \rtenv_3, \roleidx{s}{p} : \stb}
            {
                (\nu \apname)
                (
                    \actor
                    { \sessthread{s}{p}{\ctxe[\effreturn{\apname}]} }
                        { \hstate }
                        { \istate }
                    \parallel
                        \ap{\apname}{(\prole_i \mapsto \epsilon)_{i \in I}}
                )
            }
        }
    \end{mathpar}
    }

    as required.
\end{proofcase}

\begin{proofcase}{E-Register}
    {\small
    \begin{mathpar}
        \inferrule
        { \inittok \text{ fresh} }
        { \actor
                {\threadctx[\register{\apname}{\prole}{\vala}]}
                {\hstate}
                {\istate}
                \parallel
            \ap{\apname}{\apstate[\prole \mapsto \setseq{\inittok'}]}
            \ceval
            (\nu \inittok)
            (\actor
                {\threadctx[\effreturn{()}]}
                {\hstate}
                {\istate[\maptwo{\inittok}{\vala}]} \parallel
            \ap{\apname}{\apstate[\prole \mapsto \setseq{\inittok'} \cup
            \set{\inittok}]})
        }
    \end{mathpar}
    }

    Again, we prove the case where $\threadctx = \sessthread{s}{\qrole}{\ctxe[-]}$
    and let $\prole = \prole_j$ for some $j$.

    Let $\rtenv = \rtenv_2, \rtenv_3, \rtenv_4, \setseq{\inittokneg_j : \sta_j},
    \roleidx{s}{p} : \stb, a, \apname$.

    Let $\derivd$ be the following derivation:

    \begin{mathpar}
        \inferrule*
        {
            \inferrule*
            {
                \mtseqst
                    {\tyenv}
                    {\tyc}
                    {\stb}
                    {\ctxe[\register{\apname}{\prole_j}{\vala}]}
                    {\tyc}
                    {\localend}
            }
            { \threadseq
                {\tyenv}
                {\roleidx{s}{q} : \stb}
                {\tyc}
                {\sessthread{s}{q}{\ctxe[\register{\apname}{\prole_j}{\vala}]}}
            }
            \\
            {
            \begin{array}{l}
                \hstateseq{\tyenv}{\rtenv_2}{\tyc}{\hstate} \\
                \istateseq{\tyenv}{\rtenv_3}{\tyc}{\istate}
            \end{array}
            }
        }
        {
            \cseq
            { \tyenv }
            { \rtenv_2, \rtenv_3, \roleidx{s}{q} : \stb, a }
            {
                \actor
                    {\sessthread{s}{q}{\ctxe[\register{\apname}{\prole_j}{\vala}]}}
                    {\hstate}
                    {\istate}
            }
        }
    \end{mathpar}

    Assumption:

    {\footnotesize
    \begin{mathpar}
        \inferrule*
        {
            \derivd
            \\
            \inferrule*
            {
                \inferrule*
                {
                    \apseq
                        {(\prole_i: S_i)_{i \in 1..n}}
                        { \rtenv_4 }
                        { \apstate }
                }
                { \apseq
                    {(\prole_i: S_i)_{i \in 1..n}}
                    {\rtenv_4, \seq{\inittoknegprime_j : \sta_j}}
                    {\apstate[\prole_j \mapsto \setseq{\inittok'}]}
                }
                \\
                {
                    \begin{array}{l}
                        \apname : \apty{(\prole_i: S_i)_{i \in 1..n}} \in \tyenv \\
                        \comp{(\prole_i: S_i)_{i \in 1..n}} \\
                    \end{array}
                }
            }
            {
                \cseq
                { \tyenv }
                { \rtenv_4, \seq{\inittoknegprime_j : \sta_j}, \apname }
                { \ap{\apname}{\apstate[\prole_{j} \mapsto \setseq{\inittok'}]} }
            }
        }
        {
            \cseq
                { \tyenv }
                { \rtenv }
                {
                    \actor
                        {\sessthread{s}{q}{\ctxe[\register{\apname}{\prole_j}{\vala}]}}
                        {\hstate}
                        {\istate}
                    \parallel
                    \ap{\apname}{\apstate[\prole_j \mapsto \setseq{\inittok'}]}
                }
        }
    \end{mathpar}
    }

    By Lemma~\ref{lem:subterm-typability}:
    \begin{mathpar}
        \inferrule*
        {
            \vseq{\tyenv}{\apname}{\apty{(\prole_i : \sta_i)_i}} \\
            \vseq{\tyenv}{\vala}{\sttyfun{\tyc}{\tyc}{\sta_j}{\localend}{\tyc}}
}
        {
            \mtseqst
            { \tyenv }
            { \tyc }
            { \stb }
            { \register{\apname}{\prole_j}{\vala} }
            { \one }
            { \stb }
        }
    \end{mathpar}

    By Lemma~\ref{lem:replacement},
    $\mtseqst{\tyenv}{\tyc}{\stb}{\ctxe[\effreturn{()}]}{\tyc}{\localend}$.

    Now, let $\derivd'$ be the following derivation:

    {\footnotesize
    \begin{mathpar}
        \inferrule*
        {
            \inferrule*
            {
                \mtseqst{\tyenv}{\tyc}{\stb}{\ctxe[\effreturn{()}]}{\tyc}{\localend}
            }
            { \threadseq{\tyenv}{\roleidx{s}{q} : \stb}{\tyc}{\sessthread{s}{q}{\ctxe[\effreturn{()}]}} }
            \\
            \inferrule*
            {
                \vseq{\tyenv}{\vala}{\sttyfun{\tyc}{\tyc}{\sta_j}{\localend}{\tyc}}
                \\
                \istateseq{\tyenv}{\rtenv_3}{\tyc}{\istate}
            }
            { \istateseq
                {\tyenv}
                {\rtenv_3, \inittokpos : \sta_j }
                { \tyc }
                {\istate[\inittokpos \mapsto \vala]}
            }
            \\
            \hstateseq{\tyenv}{\rtenv_2}{\tyc}{\hstate}
        }
        {
            \cseq
            { \tyenv }
            { \rtenv_2, \rtenv_3, \roleidx{s}{q} : \sta, \inittokpos : \sta_j, a }
            {
                \actor
                    {\sessthread{s}{q}{\ctxe[\effreturn{()}]}}
                    {\hstate}
                    {\istate}
            }
        }
    \end{mathpar}
    }

    Finally, we can recompose:

    {\footnotesize
    \begin{mathpar}
        \inferrule*
        {
            \inferrule*
            {
                \derivd
                \\
                \inferrule*
                {
                    \inferrule*
                    {
                        \apseq
                            {(\prole_i: S_i)_{i \in 1..n}}
                            { \rtenv_4 }
                            { \apstate }
                    }
                    { \apseq
                        {(\prole_i: S_i)_{i \in 1..n}}
                        {\rtenv_4, \seq{\inittoknegprime_j : \sta_j},
                        \inittokneg : \sta_j}
                        {\apstate[\prole_j \mapsto \setseq{\inittok'} \cup \set{\inittok}]}
                    }
                    \\
                    {
                        \begin{array}{l}
                            \apname : \apty{(\prole_i: S_i)_{i \in 1..n}} \in \tyenv \\
                            \comp{(\prole_i: S_i)_{i \in 1..n}}
                        \end{array}
                    }
                }
                {
                    \cseq
                    { \tyenv }
                    { \rtenv_4, \seq{\inittoknegprime_j : \sta_j}, \inittokneg : \sta_j, \apname }
                    { \ap{\apname}{\apstate[\prole_{j} \mapsto \setseq{\inittok'} \cup \set{\inittok}]}}
                }
            }
            {
                \cseq
                    { \tyenv }
                    { \rtenv, \inittokpos : \sta_j, \inittokneg : \sta_j }
                    {
                        \actor
                            {\sessthread{s}{q}{\ctxe[\effreturn{()}]}}
                            {\hstate}
                            {\istate}
                        \parallel
                        \ap{\apname}{\apstate[\prole_j \mapsto
                        \setseq{\inittok'} \cup \set{\inittok}]}
                    }
            }
        }
        {
            \cseq
            { \tyenv }
            { \rtenv }
            {
                (\nu \inittok)
                (
                    \actor
                        {\sessthread{s}{q}{\ctxe[\effreturn{()}]}}
                        {\hstate}
                        {\istate}
                    \parallel
                        \ap{\apname}{\apstate[\prole_j \mapsto
                        \setseq{\inittok'} \cup \set{\inittok}]}
                )
            }
        }
    \end{mathpar}
    }

    as required.
\end{proofcase}

\begin{proofcase}{E-Init}

    \begin{mathpar}
        \inferrule
    { \sessname \text{ fresh} }
    {
        (\nu \inittok_{\prole_i})_{i \in 1..n}(
        \ap
            {\apname}
            {(\prole_i \mapsto \setseq{\inittok'_{\prole_i}} \cup \set{\inittok_{\prole_i}})_{i \in 1..n} }
        \parallel
        \actor
            [a_i]
            {\idle{\valc_i}}
            {\hstate_i}
            {\istate_i[\inittok_{\prole_i} \mapsto (\fun{\var{st}_i}{\tma_i})]}_{i \in 1..n})
\cevalann{\tau}\\\\
(\nu \sessname)(
            \ap{\apname}{(\prole_i \mapsto \setseq{\inittok'_{\prole_i}})_{i \in 1..n}}
            \parallel
            \qproc{s}{\epsilon}
            \parallel
            \actor
                [a_i]
                {\sessthreadd{\sessname}{\prole_i}{\subst{\tma_i}{\valc_i}{\var{st}_i}}}
                {\hstate_i}
                {\istate_i}_{i \in 1..n})
    }
    \end{mathpar}
    For each actor composed in parallel we have:

    {\footnotesize
    \begin{mathpar}
        \inferrule
        {
            \inferrule*
            {
                \vseq
                    {\tyenv}
                    {\fun{\var{st}_i}{\tma_i}}
                    {\sttyfun{\tyc_i}{\tyc_i}{\sta_i}{\localend}}{\tyc_i}  \\
                    \istateseq{\tyenv}{\rtenv_{i_3}}{\tyc_i}{\istate}
            }
            { \istateseq
                {\tyenv}
                {\tyc_i}
                {\rtenv_{i_3}, \inittokpos_i : \sta_i}
                {\istate_i[\inittok_{\prole_i} \mapsto \fun{\var{st}_i}{\tma_i} ]}
            }
            \\
            \inferrule*
            { \vseq{\tyenv}{\valc_i}{\tyc_i}}
            { \threadseq{\tyenv}{\cdot}{\tyc_i}{\idle{\valc_i}}}
            \\
            \hstateseq
                {\tyenv}
                {\rtenv_{i_2}}
                {\tyc_i}
                {\hstate_i}
        }
        { \cseq
            { \tyenv }
            { \rtenv_{i_2}, \rtenv_{i_3}, \inittokpos_i : \sta_i, \aname_i }
            { \actor
                [\aname_i]
                {\idle{\valc_i}}
                {\hstate_i}
                {\istate_i[\inittok_{\prole_i} \mapsto (\fun{\var{st}_i}{\tma_i})]}
            }
        }
    \end{mathpar}
    }

    Let:

    \begin{itemize}
        \item $\rtenv_{\var{tok+}} = \inittokpos_1 : \sta_1, \ldots, \inittokpos_n : \sta_n$
        \item $\rtenv_{\var{tok-}} = \inittokneg_1 : \sta_1, \ldots, \inittokneg_n : \sta_n$
        \item $\rtenv_{\var{a}} = \rtenv_{1_2}, \rtenv_{1_3}, \ldots,
            \rtenv_{n_2}, \rtenv_{n_3}, \aname_1, \ldots, \aname_n$
        \item $\rtenv_{\var{b}} = \rtenv_{\var{a}}, \rtenv_{\var{tok+}}$
    \end{itemize}

    Then by repeated use of \textsc{TC-Par} we have that
    $\cseq
        {\tyenv}
        { \rtenv_{\var{a}}, \rtenv_{\var{tok+}} }
        { (\actor
            {\idle{\valc_i}}
            {\hstate_i}
            {\istate_i[\inittok_{\prole_i} \mapsto \fun{\var{st}_i}{\tma_i}]}
          )_{ i \in 1..n}
        }
    $

    Assumption (given some $\rtenv$):

    {\footnotesize
    \begin{mathpar}
        \inferrule*
        {
            \inferrule*
            {
                \inferrule*
                {
                    \apname : \apty{(\prole_i : \sta_i)_i} \in \tyenv \\\\
                    \apseq
                        {\prole_i : \sta_i}
                        {\rtenv, \rtenv_{\var{tok-}}}
                        {(\prole_i \mapsto \setseq{\inittok'_{\prole_i}} \cup \set{\inittok_{\prole_i}})_{i \in 1..n} } \\\\
                        \comp{(\prole_i : \sta_i)_{i \in 1..n}} \\
                }
                {
                    \cseq
                        { \tyenv }
                        {\rtenv, \rtenv_{\var{tok-}}}
                        {
                            \ap
                                {\apname}
                                {(\prole_i \mapsto \setseq{\inittok'_{\prole_i}} \cup \set{\inittok_{\prole_i}})_{i \in 1..n} }
                        }
                }
                \\
                \inferrule*
                {}
                {
                    \cseq
                        {\tyenv}
                        { \rtenv_{\var{a}}, \rtenv_{\var{tok+}} }
                        {(\actor
                            {\idle{\valc_i}}
                            {\hstate_i}
                            {\istate_i[\inittok_{\prole_i} \mapsto \fun{\var{st}_i}{\tma_i}]})_{i \in 1..n} }
                }
            }
            {
                \cseq
                { \tyenv }
                { \rtenv, \rtenv_{\var{a}}, \rtenv_{\var{tok+}}, \rtenv_{\var{tok-}} }
                {
                    \ap
                        {\apname}
                        {(\prole_i \mapsto \setseq{\inittok'_{\prole_i}} \cup \set{\inittok_{\prole_i}})_{i \in 1..n} }
                    \parallel
                    (\actor
                        {\idle{\valc_i}}
                        {\hstate_i}
                        {\istate_i[\inittok_{\prole_i} \mapsto \fun{\var{st}_i}{\tma_i}]})_{i \in 1..n}
                }
            }
        }
        {
            \cseq
                {\tyenv}
                { \rtenv, \rtenv_{\var{a}} }
                { (\nu \inittok_1) \cdots (\nu \inittok_n)
                    (
                        \ap
                            {\apname}
                            {(\prole_i \mapsto \setseq{\inittok'_{\prole_i}} \cup \set{\inittok_{\prole_i}})_{i \in 1..n} }
                        \parallel
                        (\actor
                            {\idle{\valc_i}}
                            {\hstate_i}
                            {\istate_i[\inittok_{\prole_i} \mapsto \fun{\var{st}_i}{\tma_i}]})_{i \in 1..n}
                    )
                }
        }
    \end{mathpar}
    }

    By Lemma~\ref{lem:substitution} we can show that for each callback function
    $\fun{\var{st}_i}{\tma_i}$, it is the case that
    $\mtseqst{\tyenv}{\tyc_i}{\sta_i}{\subst{\tma_i}{\valc_i}{\var{st}_i}}{\tyc_i}{\localend}$.

    Through the access point typing rules we can show that we can remove each $\inittok_{\prole_i}$ from the access point:
    $
    \cseq
                    { \tyenv }
                    {\rtenv }
                    {
                        \ap
                            {\apname}
                            {(\prole_i \mapsto \setseq{\inittok'_{\prole_i}})_{i \in 1..n}}
                    }
    $.

    Similarly, for each actor composed in parallel we can construct:

    \begin{mathpar}
        \inferrule
        {
            \inferrule*
            {
                \mtseqst{\tyenv}{\tyc_i}{\sta_i}{\subst{\tma_i}{\valc_i}{\var{st}_i}}{\tyc_i}{\localend}
            }
            {
                \threadseq
                    {\tyenv}
                    {\roleidxx{s}{\prole_i} : \sta_i}
                    { \tyc_i }
                    {\sessthreadd{s}{\prole_i}{\subst{\tma_i}{\valc_i}{\var{st}_i}}}
            } \\
            \hstateseq{\tyenv}{\rtenv_{i_2}}{\tyc_i}{\hstate_i} \\
            \istateseq{\tyenv}{\rtenv_{i_3}}{\tyc_i}{\istate_i} \\
        }
        { \cseq
            { \tyenv }
            { \rtenv_{i_2}, \rtenv_{i_3}, \roleidxx{s}{\prole_i} : \sta_i }
            { \actor
                {\sessthreadd{s}{\prole_i}{\subst{\tma_i}{\valc_i}{\var{st}_i}}}
                {\hstate_i}
                {\istate_i}
            }
        }
    \end{mathpar}

    Let $\rtenv_{s} = \roleidxx{s}{\prole_1} : \sta_1, \ldots, \roleidxx{s}{\prole_n} : \sta_n$

    Then by repeated use of \textsc{TC-Par} we have that
    $\cseq
        {\tyenv}
        {\rtenv_a, \rtenv_{\var{s}} }
        {{\actor{\sessthreadd{s}{\prole_i}{\subst{\tma_i}{\valc_i}{\var{st}_i}}}{\hstate_i}{\istate_i}}_{ i \in 1..n}}$.

    Recomposing:

    {\footnotesize
    \begin{mathpar}
        \inferrule*
        {
            \inferrule*
            {
                {
                    \bl
                        \comp{\rtenv_{\var{s}}} \\
                        \cseq
                            { \tyenv }
                            { \rtenv }
                            {
                                \ap
                                    {\apname}
                                    {(\prole_i \mapsto \setseq{\inittok'_{\prole_i}})_{i \in 1..n} }
                            }
                    \el
                }
                \\
                \inferrule*
                {
                    \inferrule*
                    { }
                    { \cseq{\tyenv}{s : \epsilon}{\qproc{s}{\epsilon}} }
                    \\
                    \cseq
                        {\tyenv}
                        { \rtenv_{\var{a}}, \rtenv_{\var{s}} }
                        {(\actor
                            {\sessthreadd{s}{\prole_i}{\subst{\tma_i}{\valc_i}{\var{st}_i}}}
                            {\hstate_i}
                            {\istate_i})_{i \in 1..n}
                        }
                }
                {
                    \cseq
                        {\tyenv}
                        { \rtenv_{\var{a}}, \rtenv_{\var{s}}, s : \epsilon }
                        {\qproc{s}{\epsilon} \parallel
                            (\actor
                                {\sessthreadd{s}{\prole_i}{\subst{\tma_i}{\valc_i}{\var{st}_i}}}
                                {\hstate_i}
                                {\istate_i})_{i \in 1..n} }
                }
            }
            {
                \cseq
                { \tyenv }
                { \rtenv, \rtenv_{\var{a}}, \rtenv_{\var{s}}, s : \epsilon}
                {
                    \ap
                        {\apname}
                        {(\prole_i \mapsto \setseq{\inittok'_{\prole_i}})_{i \in 1..n} }
                    \parallel
                        \qproc{s}{\epsilon}
                    \parallel
                    {(\actor{\sessthreadd{s}{\prole_i}{\subst{\tma_i}{\valc_i}{\var{st}_i}}}{\hstate_i}{\istate_i})_{i \in 1..n}}
                }
            }
        }
        {
            \cseq
                {\tyenv}
                { \rtenv, \rtenv_{\var{a}} }
                { (\nu s)
                    (
                        \ap
                            {\apname}
                            {(\prole_i \mapsto \setseq{\inittok'_{\prole_i}})_{i \in 1..n} }
                        \parallel
                            \qproc{s}{\epsilon}
                        \parallel
                        (\actor{\sessthreadd{s}{\prole_i}{\subst{\tma_i}{\valc_i}{\var{st}_i}}}{\hstate_i}{\istate_i})_{i
                        \in 1..n}
                    )
                }
        }
    \end{mathpar}
    }
    as required.
    \end{proofcase}

\begin{proofcase}{E-Lift}
    Immediate by Lemma~\ref{lem:term-pres}.
\end{proofcase}

\begin{proofcase}{E-Nu}

    There are different subcases based on whether $\alpha$ is an access point
    name, initialisation token name, actor name, or session name. All except
    session names follow from a straightforward application of the induction
    hypothesis so we prove the case where $\alpha = s$ for some session name
    $s$.

    Assumption:

    \begin{mathpar}
        \inferrule
        { \rtenv_s = \{ \roleidxx{s}{\role{p}_i} : S_{\role{p}_{i}} \}_{i \in 1..n}, s : \qty \\
            \compliant{\rtenv_s} \\
            s \not\in \snames{\rtenv} \\\\
          \cseq{\tyenv}{\rtenv, \rtenv_s}{\config{C}} \\
        }
        { \cseq{\tyenv}{\rtenv}{(\nu s)\config{C}} }
    \end{mathpar}
    with $\config{C} \ceval \config{C}'$.

Since $\compliant{\rtenv_s}$ we have that $\safe{\rtenv_s}$ and $\df{\rtenv_s}$.

Since $s \not \in \rtenv$ and therefore $\snames{\rtenv} \cap \snames{\rtenv_s} = \emptyset$,
by Lemma~\ref{lem:separate-envs-safe} we have that $\safe{\rtenv, \rtenv_s}$. 

By the IH we have that there exists some $\rtenv'$ such that
$\rtenv, \rtenv_s \equivsynceval^? \rtenv'$,
where $\safe{\rtenv'}$ and $\cseq{\tyenv}{\rtenv'}{\config{C}'}$.

By Lemma~\ref{lem:env-reductions-separate}, there are three subcases: 

\begin{itemize}
    \item $\rtenv' = \rtenv$, which follows trivially.
    \item $\rtenv' = \rtenv'', \rtenv_s$ where $\rtenv \equivsynceval \rtenv''$
        with $\safe{\rtenv''}$ and we can therefore show:

        \begin{mathpar}
            \inferrule
            { \rtenv_s = \{ \roleidxx{s}{\role{p}_i} : S_{\role{p}_{i}} \}_{i \in 1..n}, s : \qty \\
                \compliant{\rtenv_s} \\
                s \not\in \snames{\rtenv''} \\\\
              \cseq{\tyenv}{\rtenv'', \rtenv_s}{\config{C}'} \\
            }
            { \cseq{\tyenv}{\rtenv''}{(\nu s)\config{C}'} }
        \end{mathpar}
        as required.

    \item $\rtenv' = \rtenv, \rtenv'_s$ where
        $\rtenv_s \equivsynceval  \rtenv'_s$ and $\safe{\rtenv'_s}$.
        It follows from the definition of progress that $\df{\rtenv'_s}$ and
        thus $\compliant{\rtenv'_s}$. We can therefore show:

        \begin{mathpar}
            \inferrule
            { \rtenv'_s = \{ \roleidxx{s}{\role{p}_i} : S'_{\role{p}_{i}} \}_{i \in 1..m}, s : \qty' \\
                \compliant{\rtenv'_s} \\
                s \not\in \snames{\rtenv} \\\\
              \cseq{\tyenv}{\rtenv, \rtenv'_s}{\config{C}'} \\
            }
            { \cseq{\tyenv}{\rtenv}{(\nu s)\config{C}'} }
        \end{mathpar}
        as required.
\end{itemize}

\end{proofcase}

\begin{proofcase}{E-Par}
    Immediate by the IH and Lemma~\ref{lem:split-safety}.
\end{proofcase}

\begin{proofcase}{E-Struct}
    Immediate by the IH and Lemma~\ref{lem:equiv-pres}.
\end{proofcase}

\end{proof}

\preservation*
\begin{proof}
    Immediate from Lemmas~\ref{lem:equiv-pres}
    and~\ref{lem:ceval-pres}.
\end{proof}

\clearpage

\subsection{Progress}

Let $\tyenvrt$ be a type environment containing only references to access
points:
\[
    \tyenvrt ::= \cdot \midspace \tyenvrt, p : \apty{(\prole_i : \sta_i)_i}
\]

Functional reduction satisfies progress.

\begin{lemma}[Term Progress]\label{lem:term-progress}
    If $\mtseq{\tyenvrt}{\sta_1}{\tma}{\tya}{\sta_2}$ then either:
    \begin{itemize}
        \item $\tma = \effreturn{\vala}$ for some value $\vala$; or
        \item there exists some $\tmb$ such that $\tma \teval \tmb$; or
        \item $\tma$ can be written $\ctxe[\tma']$ where $\tma'$ is a communication or concurrency construct, i.e.
            \begin{itemize}
                \item $\tma = \spawn{\tmb}$ for some $\tmb$; or
                \item $\tma = \send{p}{\ell}{\vala}$ for some role $\prole$ and
                    message $\msg{\ell}{\vala}$; or
                \item $\tma = \suspend{\vala}{\valb}$ or some $\vala, \valb$; or
                \item $\tma = \newap{(\prole_i : \stb_i)}$ for some collection of
                    participants $(\prole_i : \stb_i)$
                \item $\tma = \register{\vala}{\prole}{\valb}$ for some values
                    $\vala, \valb$
                    and role $\prole$
            \end{itemize}
    \end{itemize}
\end{lemma}
\begin{proof}
    A standard induction on the derivation of
    $\mtseq{\tyenvrt}{\sta_1}{\tma}{\tya}{\sta_2}$;
    there are $\beta$-reduction rules for all STLC terms, leaving only values
    and communication / concurrency terms.
\end{proof}

The key \emph{thread progress} lemma shows that each actor is either idle, or
can reduce; the proof is by inspection of $\threadt$, noting there are reduction
rules for each construct; the runtime typing rules ensure the presence of any
necessary queues or access points.

\begin{restatable}[Thread Progress]{lemma}{threadprog}\label{lem:thread-progress}
    Let $\config{C} = \confctx[\actor{\threadt}{\hstate}{\istate}]$.
    If $\cseq{\cdot}{\cdot}{\config{C}}$ then
    either $\threadt = \idle{\vala}$ for some value $\vala$,
    or there exist $\confctx', \threadt', \hstate', \istate', \vala'$ such that
    $\config{C} \ceval \confctx'[\actor{\threadt'}{\hstate'}{\istate'}]$ is a
    thread reduction for $a$.
\end{restatable}
\begin{proof}
    If $\threadt = \idle{\vala}$ then the theorem is satisfied, so consider the cases
    where $\threadt = \tma$ or $\threadt = \sessthread{s}{p}{\tma}$.
By Lemma~\ref{lem:term-progress}, either $\tma$ can reduce (and the
    configuration can reduce via \textsc{E-Lift}), $\tma$ is a value (and the
    thread can reduce by \textsc{E-Reset}), or $\tma$ is a state, communication or
    concurrency construct. Of these:
    \begin{itemize}
        \item $\calcwd{get}$ and $\calcwd{set}$ can reduce by \textsc{E-Get} and
            \textsc{E-Set} respectively
        \item $\spawn{\tmb}$ can reduce by \textsc{E-Spawn}
        \item $\suspend{\vala}$ can reduce by \textsc{E-Suspend}
        \item $\newap{(\prole_i : \sta_i)_i}$ can reduce by \textsc{E-NewAP}
    \end{itemize}
    Next, consider $\register{p}{\prole}{\tma}$. Since we begin with a closed
    environment, it must be the case that $p$ is $\nu$-bound so by
    \textsc{T-APName} and \textsc{T-AP} there must exist some subconfiguration
    $\ap{p}{\apstate}$ of $\ctxg$; the configuration can therefore reduce by
    \textsc{E-Register}.

    Finally, consider $\tma = \send{\qrole}{\ell}{\vala}$.
It cannot be the case that $\threadt = \send{\qrole}{\ell}{\vala}$ since by
    \textsc{T-Send} the term must have an output session type as a precondition,
    whereas \textsc{TT-NoSess} assigns a precondition of $\localend$. Therefore,
    it must be the case that $\threadt =
    \sessthread{s}{p}{\send{\qrole}{\ell}{\vala}}$ for some $s, \prole$.
Again since the initial runtime typing environment is empty, it must be the
    case that $s$ is $\nu$-bound and so by \textsc{T-SessionName} and
    \textsc{T-EmptyQueue}/\textsc{T-ConsQueue} there must be some session queue
    $\qproc{s}{\qcontents}$. The thread must therefore be able to reduce by
    \textsc{E-Send}.
\end{proof}

\begin{proposition}\label{prop:canonical-forms}
If $\cseq{\tyenv}{\rtenv}{\config{C}}$ then there exists a $\config{D}
    \equiv \config{C}$ where $\config{D}$ is in canonical form.
\end{proposition}

\progress*
\begin{proof}
    By Proposition~\ref{prop:canonical-forms} $\config{C}$ can be written in
    canonical form:
\[
        (\nu \tilde{\inittok})
        (\nu \apname_{i \in 1..l})
        (\nu \sessname_{j \in 1..m})
        (\nu \aname_{k \in 1..n})
        (
        \ap{\apname_i}{\apstate_i}_{i \in 1..l} \parallel
        (\qproc{s_j}{\qcontents_j})_{j \in 1..m} \parallel 
        \actor[\aname_k]{\threadt_k}{\hstate_k}{\istate_k}_{k \in 1..n}
        )
    \]
By repeated applications of Lemma~\ref{lem:thread-progress}, either the
    configuration can reduce or all threads are idle:
\[
        (\nu \tilde{\inittok})
        (\nu \apname_{i \in 1..l})
        (\nu \sessname_{j \in 1..m})
        (\nu \aname_{k \in 1..n})
        (
        \ap{\apname_i}{\apstate_i}_{i \in 1..l} \parallel
        (\qproc{s_j}{\qcontents_j})_{j \in 1..m} \parallel 
        \actor[\aname_k]{\idle{\valc_k}}{\hstate_k}{\istate_k}_{k \in 1..n}
        )
    \]

By the linearity of runtime type environments $\rtenv$, each role endpoint
    $\roleidx{s}{p}$  must be contained in precisely one actor. There are two
    ways an endpoint can be used: either by \textsc{TT-Sess} in order to run a
    term in the context of a session, or by \textsc{TH-Handler} to record a
    receive session type as a handler. Since all threads are idle, it must be
    the case the only applicable rule is \textsc{TH-Handler} and therefore each
    role must have an associated stored handler.

    Since the types for each session must satisfy progress, the collection of
    local types must reduce. Since all session endpoints must have a receive
    session type, the only type reductions possible are through
    \textsc{Lbl-Sync-Recv}. Since all threads are idle we can pick the top
    message from any session queue and reduce the actor with the associated
    stored handler by \textsc{E-React}.

    The only way we could not do such a reduction is if there were to be no
    sessions, leaving us with a configuration of the form:
\[
        (\nu \tilde{\inittok})
        (\nu \apname_{i \in 1..m})
        (\nu \aname_{j \in 1..n})
        (
        \ap{\apname_i}{\apstate_i}_{i \in 1..m} \parallel
        \actor[\aname_j]{\idle{\valc_j}}{\hstate_j}{\istate_j}_{j \in 1..n}
        )
    \]

\end{proof}

\subsection{Global Progress}

\independence*
\begin{proof}
By inspection of the reduction rules. The only non-administrative rules that can
be thread reductions are \textsc{E-Send}, \textsc{E-Spawn},
\textsc{E-Reset}, \textsc{E-NewAP}, \textsc{E-Register}, and
\textsc{E-Lift}. It suffices to show that reduction in one actor does not inhibit
reduction in another. Pick two arbitrary actors $a$ and $b$ in $\config{C}$ and
proceed by case analysis on pairs of thread reductions.  The left-hand-sides of
rules \textsc{E-Spawn}, \textsc{E-Reset}, \textsc{E-NewAP}, and \textsc{E-Lift} only depend on a single actor and
are therefore unaffected by other rules. \textsc{E-Send} places no restrictions on the
session queue, and so other sends or receives to the queue do not affect the
ability of \textsc{E-Send} to fire; similar reasoning applies for
\textsc{E-Register}.
\end{proof}

\idlereduces*
\begin{proof}
    Follows as a direct consequence of Theorem~\ref{thm:progress}: in an
    idle configuration, typing ensures that every ongoing session has a
    suspended handler, and compliance ensures that the session queue has a
    message that can activate each handler. Thus the configuration can
    reduce by \textsc{E-React} for each ongoing session in the system.
\end{proof}

\globalprogress*
\begin{proof}
    Pick an arbitrary active session $s \in
    \activesessions{\config{C}}$.
    By the structural congruence rules and Theorem~\ref{thm:preservation}, $\mathcal{C} \equiv (\nu s) \mathcal{D}$
    and $\cseq{\cdot}{\cdot}{(\nu s)\config{C}}$.
    By Corollary~\ref{cor:gp:reduction-idle},
    $(\nu s)\config{D}$ can reduce to an idle configuration
$\config{D}'$.
If that reduction sequence contains a reduction on session $s$
then the property is satisfied and we are done; otherwise we conclude by
Lemma~\ref{lem:gp:idle-reduces}.
\end{proof}

 \clearpage
\section{Supplement to Section~\ref{sec:failure}}\label{ap:extensions-proofs}

This appendix details the full formal development and proofs for
$\langnamezap$ (Section~\ref{sec:failure}).

First, it is useful to show that safety is preserved even if several roles are
cancelled; we use this lemma implicitly throughout the preservation proof.

Let us write $\roles{\rtenv} = \set{\prole \midspace \roleidx{s}{p} : \sta \in \rtenvzap}$
to retrieve the roles from an environment.
Let us also define the operation $\zaproles{\rtenvzap}{\setseq{\prole}}$ that
cancels any role in the given set, i.e., $\zaproles{\roleidxx{s}{\prole_1} :
\sta_1, \roleidxx{s}{\prole_2} : \sta_2, a}{\set{\prole_1}} =
\roleidxx{s}{\prole_1} : \zapped, \roleidxx{s}{\prole_2} : \sta_2, a$.
\begin{lemma}
    If $\safe{\rtenvzap}$ then $\safe{\zaproles{\rtenvzap}{\setseq{\prole}}}$
    for any $\setseq{\prole} \subseteq \roles{\rtenvzap}$.
\end{lemma}
\begin{proof}
    Zapping a role does not affect safety; the only way to violate safety is by
    \emph{adding} further unsafe communication reductions.
\end{proof}

\zappres*
\begin{proof}
    Preservation of typability by structural congruence is straightforward, so
    we concentrate on preservation of typability by reduction. We proceed by
    induction on the derivation of $\config{C} \ceval \config{D}$, concentrating
    on the new rules rather than the adapted rules (which are straightforward
    changes to the existing proof).

    \begin{proofcase}{E-Monitor}
        \[
            \zapactor{\threadctx[\monitor{b}{\vala}]}{\hstate}{\istate}{\monstate} 
            \cevaltau
            \zapactor{\threadctx[\effreturn{()}]}{\hstate}{\istate}{\monstate
            \cup \set{\metapair{b}{\vala}}}
        \]

        We consider the case where $\threadctx = \ctxe[-]$ for some $\ctxe$; the
        case in the context of a session is similar.

        Assumption:
        \begin{mathpar}
            \inferrule
            {
                \inferrule
                { \mtseqst{\tyenv}{\sta}{\tyc}{\ctxe[\monitor{b}{\vala}]}{\tyc}{\localend} }
                { \threadseq{\tyenv}{\cdot}{\tyc}{\ctxe[\monitor{b}{\vala}]} }
                \\
                \hstateseq{\tyenv}{\rtenvzap_1}{\tyc}{\hstate} \\
                \istateseq{\tyenv}{\rtenvzap_2}{\tyc}{\istate}
            }
            { \cseq
                { \tyenv }
                { \rtenvzap_1, \rtenv_2, \aname }
                { \zapactor
                    {\ctxe[\monitor{b}{\vala}]}
                    {\hstate}
                    {\istate}
                    {\monstate}
                }
            }
        \end{mathpar}
        where
        $\forall \metapair{b}{\valb} \in \monstate .\;
                \vseq{\tyenv}{b}{\typid} \,\wedge\,
                \vseq{\tyenv}{\valb}{\sttyfun{\tyc}{\tyc}{\localend}{\localend}{\tyc}}$.

        By Lemma~\ref{lem:subterm-typability}, we know:

        \begin{mathpar}
            \inferrule
            {
                \vseq{\tyenv}{b}{\typid} \\
                \vseq{\tyenv}{\vala}{\sttyfun{\tyc}{\tyc}{\localend}{\localend}{\tyc}}
            }
            { \mtseqst{\tyenv}{\tyc}{\sta}{\monitor{b}{\vala}}{\one}{\sta} }
        \end{mathpar}

        By Lemma~\ref{lem:replacement} we know
        $\mtseqst{\tyenv}{\tyc}{\sta}{\ctxe[\effreturn{()}]}{\tyc}{\localend}$.

        Recomposing:
        \begin{mathpar}
            \inferrule
            {
                \inferrule
                { \mtseqst{\tyenv}{\tyc}{\sta}{\ctxe[\effreturn{()}]}{\tyc}{\localend} }
                { \threadseq{\tyenv}{\cdot}{\tyc}{\ctxe[\effreturn{()}]} }
                \\
                \hstateseq{\tyenv}{\rtenvzap_1}{\tyc}{\hstate} \\
                \istateseq{\tyenv}{\rtenvzap_2}{\tyc}{\istate}
            }
            { \cseq
                { \tyenv }
                { \rtenvzap_1, \rtenv_2, \aname }
                { \swactor
                    {\ctxe[\effreturn{()}]}
                    {\hstate}
                    {\istate}
                    {\monstate \cup \metapair{b}{\vala}}
                }
            }
        \end{mathpar}
        noting that $\monstate \cup \metapair{b}{\vala}$ is well-typed since 
        $\vseq{\tyenv}{b}{\typid}$ and
        $\vseq{\tyenv}{\vala}{\sttyfun{\tyc}{\tyc}{\localend}{\localend}{\tyc}}$, as
        required.
    \end{proofcase}

    \begin{proofcase}{E-InvokeM}
        \[
            \zapactor{\idle{\valc}}{\hstate}{\istate}{\monstate \cup
            \set{\metapair{b}{\vala}}} \parallel \zap{b}
                \cevaltau
            \zapactor{\vala \app \valc}{\hstate}{\istate}{\monstate} \parallel \zap{b}
         \]

         Assumption:

         {\footnotesize
         \begin{mathpar}
             \inferrule*
             {
              \inferrule*
                {
                    \inferrule*
                    { \vseq{\tyenv}{\valc}{\tyc} }
                    { \threadseq{\tyenv}{\cdot}{\tyc}{\idle{\valc}} }
                    \\
                    \hstateseq{\tyenv}{\rtenvzap_1}{\tyc}{\hstate} \\
                    \istateseq{\tyenv}{\rtenvzap_2}{\tyc}{\istate}
                }
                { \cseq
                    { \tyenv }
                    { \rtenvzap_1, \rtenvzap_2, \aname }
                    { \zapactor
                        {\idle{\valc}}
                        {\hstate}
                        {\istate}
                        {\monstate \cup \set{\metapair{b}{\vala}}}
                    }
                }
                \\
                \inferrule*
                { }
                { \cseq{\tyenv}{b}{\zap{b}} }
            }
            {
                \cseq
                { \tyenv }
                { \rtenvzap_1, \rtenvzap_2, \aname, b}
                {
                    \zapactor
                        {\idle{\valc}}
                        {\hstate}
                        {\istate}
                        {\monstate \cup \set{\metapair{b}{\vala}}}
                    \parallel
                    \zap{b}
                }
            }
        \end{mathpar}
        }

        where
        $\forall \metapair{a'}{\valb} \in \monstate \cup \set{\metapair{b}{\vala}} .\; \vseq{\tyenv}{b}{\typid}
        \,\wedge\, \vseq{\tyenv}{\valb}{\sttyfun{\tyc}{\tyc}{\localend}{\localend}{\tyc}}$.
        
        Recomposing:

        \begin{mathpar}
            \inferrule*
            {
                \inferrule*
                {
                    \inferrule*
                    { 
                        \inferrule*
                        {
                            \vseq{\tyenv}{\vala}{\sttyfun{\tyc}{\tyc}{\localend}{\localend}{\tyc}}
                            \\
                            \vseq{\tyenv}{\valc}{\tyc}
                        }
                        { \mtseqst{\tyenv}{\tyc}{\localend}{\vala \app \valc}{\tyc}{\localend} }
                    }
                    { \threadseq{\tyenv}{\cdot}{\tyc}{\vala \app \valc} } \\
                    \hstateseq{\tyenv}{\rtenvzap_1}{\valc}{\hstate} \\
                    \istateseq{\tyenv}{\rtenvzap_2}{\valc}{\istate}
                }
                { \cseq
                    { \tyenv }
                    { \rtenvzap_1, \rtenvzap_2, \aname }
                    { \swactor
                        {\vala \app \valc}
                        {\hstate}
                        {\istate}
                        {\monstate}
                    }
                }
                \\
                \inferrule*
                { }
                { \cseq{\tyenv}{b}{\zap{b}} }
            }
            {
                \cseq
                { \tyenv }
                { \rtenvzap_1, \rtenvzap_2, \aname, b }
                {
                    \zapactor
                        {\vala \app \valc}
                        {\hstate}
                        {\istate}
                        {\monstate}
                    \parallel
                    \zap{b}
                }
            }
        \end{mathpar}
        as required.

    \end{proofcase}

    \begin{proofcase}{E-Raise}
        Similar to \textsc{E-RaiseS}.
    \end{proofcase}

    \begin{proofcase}{E-RaiseS}
        \[
            \zapactor{\sessthread{s}{p}{\ctxe[\raiseexn]}}{\hstate}{\istate}{\monstate} \cevaltau
             \zap{\aname} \parallel \zap{\roleidx{s}{p}} \parallel \zap{\hstate} \parallel \zap{\istate}
        \]

        \begin{mathpar}
           \inferrule
             {
                 \inferrule*
                 {
                     \mtseqst{\tyenv}{\tyc}{\sta}{\ctxe[\raiseexn]}{\one}{\localend}
                 }
                 {
                     \threadseq{\tyenv}{\roleidx{s}{p} : \sta}{\tyc}{\sessthread{s}{p}{\ctxe[\raiseexn]}}
                 }
                 \\
                 \hstateseq{\tyenv}{\rtenvzap_1}{\tyc}{\hstate} \\
                 \istateseq{\tyenv}{\rtenvzap_2}{\tyc}{\istate} \\
                 \vseq{\tyenv}{\valc}{\tyc}
             }
             { \cseq
                 { \tyenv }
                 { \rtenvzap_1, \rtenvzap_2, \roleidx{s}{p} : {\sta}, \aname }
                 { \swactor
                     {\sessthread{s}{p}{\ctxe[\raiseexn]}}
                     {\hstate}
                     {\istate}
                     {\monstate}
                 }
             } 
        \end{mathpar}

        Let us write $\zapenv{\rtenvzap} = \set{ \roleidx{s}{p} : \zapped \mid
        \roleidx{s}{p} : \sta \in \rtenvzap}$. It follows that for a given
        environment, $\rtenvzap \zaplbleval{}^{*} \zapenv{\rtenvzap}$.

        The result follows by noting that due to \textsc{TH-Handler} and
        \textsc{TI-Callback} we have that $\fn{\rtenvzap_1} = \fn{\hstate}$ and
        $\fn{\rtenvzap_2} = \fn{\istate}$. Thus:
        \begin{itemize}
            \item $\cseq{\tyenv}{\zapenv{\rtenvzap_1}}{\zap{\hstate}}$, 
            \item $\cseq{\tyenv}{\zapenv{\rtenvzap_2}}{\zap{\istate}}$, 
            \item $\cseq{\tyenv}{\zapenv{\rtenvzap_1}, \zapenv{\rtenvzap_2}, \roleidx{s}{p} : \zapped, a}{\zap{a}
        \parallel \zap{\roleidx{s}{p}} \parallel \zap{\hstate} \parallel
        \zap{\istate}}$
        \end{itemize}

        with the environment reduction:
        \[
            \rtenvzap_1, \rtenvzap_2, \roleidx{s}{p} : {\sta}, \aname
            \zaplbleval{}^{+}
            \zapenv{\rtenvzap_1}, \zapenv{\rtenvzap_2}, \roleidx{s}{p} :
            \zapped, \aname
        \]

        as required.
    \end{proofcase}

    \begin{proofcase}{E-CancelMsg}
        \[
             \qproc{s}{\qentry{p}{q}{\ell}{\vala} \cdot \qcontents} \parallel \zap{\roleidx{s}{q}}
             \cevaltau \qproc{s}{\qcontents} \parallel \zap{\roleidx{s}{q}}
        \]

        Assumption:

        \begin{mathpar}
            \inferrule*
            {
                \inferrule*
                {
                    \vseq{\tyenv}{\vala}{\tya} \\
                    \cseq{\tyenv}{s : \qty}{\qproc{s}{\qcontents}}
                }
                { \cseq
                    {\tyenv}
                    {s : \qentry{p}{q}{\ell}{\tya} \cdot \qty}
                    { \qproc{s}{\qentry{p}{q}{\ell}{\vala} \cdot \qcontents}}
                }
                \\
                \cseq{\tyenv}{\roleidx{s}{q} : \zapped}{\zap{\roleidx{s}{q}}}
            }
            {
                \cseq
                {\tyenv}
                {
                    \roleidx{s}{q} : \zapped,
                    s : \qentry{p}{q}{\ell}{\vala} \cdot \qty
                }
                { \qproc{s}{\qentry{p}{q}{\ell}{\vala} \cdot \qcontents}
                \parallel \zap{\roleidx{s}{q}} }
            }
        \end{mathpar}

        Recomposing, we have:

        \begin{mathpar}
            \inferrule*
            {
                \cseq{\tyenv}{s : \qty}{\qproc{s}{\qcontents}}
                \\
                \cseq{\tyenv}{\roleidx{s}{q} : \zapped}{\zap{\roleidx{s}{q}}}
            }
            {
                \cseq
                {\tyenv}
                {
                    \roleidx{s}{q} : \zapped,
                    s : \qty
                }
                { \qproc{s}{\qcontents} \parallel \zap{\roleidx{s}{q}} }
            }
        \end{mathpar}
        with
        \[
\roleidx{s}{q} : \zapped, s : \qentry{p}{q}{\ell}{\vala} \cdot \qty
\lbleval{\lblzapmsg{s}{p}{q}{\ell}}
\roleidx{s}{q} : \zapped, s : \qty
        \]
        as required.
    \end{proofcase}

    \begin{proofcase}{E-CancelAP}
\[
     (\nu \inittok)(\ap{p}{\apstate[\prole \mapsto \setseq{\inittok'} \cup \set{\inittok}]} \parallel \zap{\inittok})
     \cevaltau
      \ap{p}{\apstate[\prole \mapsto \setseq{\inittok'}]}
\]

Assumption:
{\footnotesize
\begin{mathpar}
    \inferrule*
    {
        \inferrule*
        {
            \inferrule*
            {
                p : \apty{\prole_i : \sta_i}_i
                \\
                \inferrule*
                {
                    \apseq{(\prole_i : \sta_i)_i}{\rtenvzap}{\apstate}
                }
                {
                    \apseq
                        {(\prole_i : \sta_i)_i}
                        {\rtenvzap, \setseq{\inittok^{'-} : \sta_j}, \inittokneg : \sta_j}
                        {\apstate[\prole_j \mapsto \setseq{\inittok'} \cup \set{\inittok}] }
                }
            }
            { \cseq
                {\tyenv}
                {\rtenvzap, \setseq{\inittok^{'-} : \sta_j}, \inittokneg : \sta_j}
                {\ap{p}{\apstate[\prole_j \mapsto \setseq{\inittok'} \cup \set{\inittok}]}}
            }
            \\
            \inferrule*
            { }
            { \cseq{\tyenv}{\inittokpos : \sta_j}{\zap{\inittok}} }
        }
        {
            \cseq
                {\tyenv}
                {\rtenvzap, \setseq{\inittok^{'-} : \sta_j}, \inittokpos :
                    \sta_j, \inittokneg : \sta_j, p}
                {\ap{p}{\apstate[\prole_j \mapsto \setseq{\inittok'} \cup \set{\inittok}]} \parallel \zap{\inittok}}
        }
    }
    {
      \cseq
        {\tyenv}
        {\rtenvzap, \setseq{\inittok^{'-} : \sta_j}, p}
        {
          (\nu \inittok)(\ap{p}{\apstate[\prole_j \mapsto \setseq{\inittok'} \cup \set{\inittok}]} \parallel \zap{\inittok})
        }
    }
\end{mathpar}
}

Recomposing:

\begin{mathpar}
    \inferrule*
    {
        p : \apty{\prole_i : \sta_i}_i
        \\
        \inferrule*
        {
            \apseq{(\prole_i : \sta_i)_i}{\rtenvzap}{\apstate}
        }
        {
            \apseq
                {(\prole_i : \sta_i)_i}
                {\rtenvzap, \setseq{\inittok^{'-} : \sta_j}}
                {\apstate[\prole_j \mapsto \setseq{\inittok'}] }
        }
    }
    { \cseq
        {\tyenv}
        {\rtenvzap, \setseq{\inittok^{'-} : \sta_j}, p}
        {\ap{p}{\apstate[\prole_j \mapsto \setseq{\inittok'}]}}
    }
\end{mathpar}

as required.
\end{proofcase}

    \begin{proofcase}{E-CancelH}
        \[
            \bl
            \zapactor
            {\idle{\valc}}
                {\hstate[\roleidx{s}{p} \mapsto \metapair{\vala}{\valb} }
                {\istate}
                {\monstate}
                \parallel
                \qproc{s}{\qcontents}
                \parallel \zap{\roleidx{s}{q}}
                \cevaltau
                \\
                \quad
                \zapactor
                    {\valb \app \valc}
                    {\hstate}
                    {\istate}
                    {\monstate}
                    \parallel
                    \qproc{s}{\qcontents}
                    \parallel
                    \zap{\roleidx{s}{q}} \parallel \zap{\roleidx{s}{p}}
\quad \text{if }
                                  \messages{\qrole}{\prole}{\qcontents} =
                                  \emptyset
              \el
        \]

        Let $\derivd$ be the following derivation:
        
    {\footnotesize
        \begin{mathpar}
            \inferrule
                {
                    \inferrule*
                    { \vseq{\tyenv}{\valc}{\tyc} }
                    { \threadseq{\tyenv}{\cdot}{\tyc}{\idle{\valc}} }
                    \\
                    \inferrule*
                    {
                        \stb = \localoffer{q}{\msg{\ell_i}{x_i}\mapsto \sta_i}_i \\
                        \vseq{\tyenv}{\vala}{\handlerty{\stb}} \\\\
                        \vseq{\tyenv}{\valb}{\sttyfun{\tyc}{\tyc}{\localend}{\localend}{\tyc}}
                        \\
                        \hstateseq{\tyenv}{\rtenvzap_1}{\tyc}{\hstate}
                    }
                    { \hstateseq
                        {\tyenv}
                        {\rtenvzap_1, \roleidx{s}{p} : \stb }
                        {\tyc}
                        { \hstate[\roleidx{s}{p} \mapsto \metapair{\vala}{\valb}] }
                    }
                    \\
                    \istateseq{\tyenv}{\rtenvzap_2}{\tyc}{\istate} \\
                    \vseq{\tyenv}{\valc}{\tyc}
}
                { \cseq
                    { \tyenv }
                    { \rtenvzap_1, \rtenvzap_2, \roleidx{s}{p} : \stb, \aname }
                    { \zapactor
                        {\idle{\valc}}
                        {\hstate[\roleidx{s}{p}\mapsto \metapair{\vala}{\valb}]}
                        {\istate}
                        {\monstate}
                    }
                }
        \end{mathpar}
    }

        Assumption:

        {\footnotesize
        \begin{mathpar}
            \inferrule
            {
                \derivd
                \\
                \inferrule*
                {
                    \cseq{\tyenv}{s : \qty}{\qproc{s}{\qcontents}}
                    \\
                    \cseq{\tyenv}{\roleidx{s}{p} : \zapped}{\zap{\roleidx{s}{p}}}
                }
                { \cseq
                    {\tyenv}
                    {s : \qty, \roleidx{s}{p} : \zapped}
                    {\qproc{s}{\qcontents} \parallel \zap{\roleidx{s}{p}}}
                }
            }
            {
                { \cseq
                    { \tyenv }
                    { \rtenvzap_1, \rtenvzap_2, \roleidx{s}{p} : \stb,
                    s : \qty, \roleidx{s}{q} : \zapped, \aname }
                    { \zapactor
                        {\idle{\valc}}
                        {\hstate[\roleidx{s}{p}\mapsto \metapair{\vala}{\valb}]}
                        {\istate}
                        {\monstate}
                        \parallel
                      \qproc{s}{\qcontents} \parallel \zap{\roleidx{s}{p}}
                    }
                }
            }
        \end{mathpar}
        }

        We can recompose as follows. Let $\derivd'$ be the following
        derivation:

        \begin{mathpar}
            \inferrule
                {
                    \inferrule*
                    {
                        \inferrule*
                        {
                          \vseq{\tyenv}{\valb}{\sttyfun{\tyc}{\tyc}{\localend}{\localend}{\tyc}} \\
                          \vseq{\tyenv}{\valc}{\tyc}
                        }
                        { \mtseqst{\tyenv}{\tyc}{\localend}{\valb \app \valc}{\tyc}{\localend} }
                    }
                    { \threadseq{\tyenv}{\cdot}{\tyc}{\valb \app \valc} }
                    \\
                    \hstateseq{\tyenv}{\rtenvzap_1}{\tyc}{\hstate} \\
                    \istateseq{\tyenv}{\rtenvzap_2}{\tyc}{\istate}
                }
                { \cseq
                    { \tyenv }
                    { \rtenvzap_1, \rtenvzap_2, \aname }
                    { \zapactor
                        {\valb \app \valc}
                        {\hstate}
                        {\istate}
                        {\monstate}
                    }
                }
        \end{mathpar}

        Then we can construct the remaining derivation:
        {\footnotesize
        \begin{mathpar}
            \inferrule
            {
                \derivd'
                \\
                \inferrule*
                {
                    \inferrule*
                    {}
                    { \cseq
                        {\tyenv}
                        {s : \qty}
                        {\qproc{s}{\qcontents} }
                    }
                    \\
                    \inferrule*
                    {
                        \inferrule*
                        { }
                        { \cseq{\tyenv}{\roleidx{s}{p} : \zapped}{\zap{\roleidx{s}{p}}} }
                        \\
                        \inferrule*
                        { }
                        { \cseq{\tyenv}{\roleidx{s}{q} : \zapped}{\zap{\roleidx{s}{q}}} }
                    }
                    {
                        \cseq
                        { \tyenv }
                        { \roleidx{s}{p} : \zapped, \roleidx{s}{q} : \zapped }
                        { \zap{\roleidx{s}{p}} \parallel \zap{\roleidx{s}{q}} }
                    }
                }
                {
                    \cseq
                        {\tyenv}
                        {s : \qty, \roleidx{s}{p} : \zapped, \roleidx{s}{q} : \zapped}
                        {
                            \qproc{s}{\qcontents} \parallel \zap{\roleidx{s}{p}}
                                                  \parallel \zap{\roleidx{s}{q}}
                        }
                }
            }
            {
                { \cseq
                    { \tyenv }
                    { \rtenvzap_1, \rtenvzap_2,
                    s : \qty, \roleidx{s}{p} : \zapped, \roleidx{s}{q} :
                \zapped, \aname }
                    { \zapactor
                        {\valb \app \valc}
                        {\hstate}
                        {\istate}
                        {\monstate}
                        \parallel
                      \qproc{s}{\qcontents} \parallel \zap{\roleidx{s}{p}}
                      \parallel \zap{\roleidx{s}{q}}
                    }
                }
            }
        \end{mathpar}
        }

        Finally, we need to show environment reduction:

        \[
            \rtenvzap_1, \rtenvzap_2, \roleidx{s}{p} : \stb,
                    s : \qty, \roleidx{s}{q} : \zapped, \aname
\lbleval{\lblzapdequeue{s}{p}{q}}
\rtenvzap_1, \rtenvzap_2,
                    s : \qty, \roleidx{s}{p} : \zapped, \roleidx{s}{q} :
                \zapped, \aname
        \]
        as required.
    \end{proofcase}
\end{proof}

\subsection{Progress}

Thread progress needs to change to take into account the possibility of an
exception due to \textsc{E-Raise} or \textsc{E-RaiseExn}:

\begin{lemma}[Thread Progress]\label{lem:zap-thread-progress}
    Let $\config{C} = \confctx[\actor{\threadt}{\hstate}{\istate}]$.
    If $\cseq{\cdot}{\cdot}{\config{C}}$ then
    either:
    \begin{itemize}
        \item $\threadt = \idle{\vala}$, or
        \item there exist $\confctx', \threadt', \hstate',
            \istate'$ such that $\config{C} \ceval
            \confctx'[\actor{\threadt'}{\hstate'}{\istate'}]$, or
        \item $\config{C} \ceval \confctx'[\zap{\aname} \parallel \zap{\hstate}
            \parallel \zap{\istate}]$ if $\threadt = \ctxe[\raiseexn]$, or
        \item $\config{C} \ceval \confctx'[\zap{\aname} \parallel
            \zap{\roleidx{s}{p}} \parallel \zap{\hstate} \parallel
            \zap{\istate}]$ if $\threadt = \sessthread{s}{p}{\ctxe[\raiseexn]}$.
    \end{itemize}
\end{lemma}
\begin{proof}
    As with Lemma~\ref{lem:thread-progress} but taking into account that:
    \begin{itemize}
        \item $\monitor{b}{\vala}$ can always reduce by \textsc{E-Monitor};
        \item $\raiseexn$ can always reduce by either \textsc{E-Raise} or
            \textsc{E-RaiseS}.
    \end{itemize}
\end{proof}

As before, all well-typed configurations can be written in canonical form; as usual the proof relies on the fact that structural congruence is type-preserving.

\begin{lemma}\label{lem:zap-canonical-forms}
    If $\cseq{\tyenv}{\rtenvzap}{\config{C}}$ then there exists a $\config{D}
    \equiv \config{C}$ where $\config{D}$ is in canonical form.
\end{lemma}

It is also useful to see that the progress property on environments is preserved
even if some roles become cancelled.

\begin{lemma}
    If $\dfzap{\rtenvzap}$ then $\dfzap{\zaproles{\rtenvzap}{\setseq{\prole}}}$
    for any $\setseq{\prole} \subseteq \roles{\rtenvzap}$.
\end{lemma}
\begin{proof}
    Zapping a role may prevent \textsc{Lbl-Recv} from firing, but in this case
    would enable either a \textsc{Lbl-ZapRecv} or \textsc{Lbl-ZapMsg} reduction.
\end{proof}

\zapprog*
\begin{proof}
    The reasoning is similar to that of Theorem~\ref{thm:progress}.
    By Lemma~\ref{lem:zap-canonical-forms}, $\config{C}$ can be written in
    canonical form:
\[
        (\nu \tilde{\inittok})
        (\nu \apname_{i \in 1..l})(\nu \sessname_{j \in 1..m})
        (\nu \aname_{k \in 1..n})(
        \ap{\apname_i}{\apstate_i}_{i \in 1..l} \parallel 
        (\qproc{s_j}{\qcontents_j})_{j \in 1..m} \parallel
        \zapactor[a_k]{\threadt_k}{\hstate_k}{\istate_k}{\monstate_k}_{k \in 1..{n' - 1}} \parallel
        \setseq{\zap{\nma}}
        )
    \]

    with $(\zap{a_k})_{k \in n' .. n}$ contained in $\setseq{\zap{\nma}}$.

    By repeated applications of Lemma~\ref{lem:zap-thread-progress}, either the
    configuration can reduce or all threads are idle:
\[
        (\nu \tilde{\inittok})
        (\nu \apname_{i \in 1..l})(\nu \sessname_{j \in 1..m})
        (\nu \aname_{k \in 1..n})(
        \ap{\apname_i}{\apstate_i}_{i \in 1..l} \parallel 
        (\qproc{s_j}{\qcontents_j})_{j \in 1..m} \parallel
        \zapactor[a_k]{\idle{\valc_k}}{\hstate_k}{\istate_k}{\monstate_k}_{k \in 1..{n' - 1}} \parallel
        \setseq{\zap{\nma}}
        )
    \]

    By the linearity of runtime type environments $\rtenv$, each role endpoint
    $\roleidx{s}{p}$  must either be contained in an actor, or exist as a zapper
    thread $\zap{\roleidx{s}{p}} \in \setseq{\zap{\nma}}$.
Let us first consider the case that the endpoint is contained in an actor;
    we know by previous reasoning that each role must have an associated stored
    handler.

    Since the types for each session must satisfy progress, the collection of
    local types must reduce.
    There are two potential reductions: either \textsc{Lbl-Sync-Recv} in the
    case that the queue has a message, or \textsc{Lbl-ZapRecv} if the sender is
    cancelled and the queue does not have a message. In the case of
    \textsc{Lbl-Sync-Recv}, since all actors are idle we can reduce using
    \textsc{E-React} as usual. In the case of \textsc{Lbl-ZapRecv} typing
    dictates that we have a zapper thread for the sender and so can reduce by
    \textsc{E-CancelH}.

    It now suffices to reason about the case where all endpoints are zapper
    threads (and thus by linearity, where all handler environments are empty).
    In this case we can repeatedly reduce with \textsc{E-CancelMsg} until all
    queues are cleared, at which point we have a configuration of the form:

    \[
        (\nu \tilde{\inittok})
        (\nu \apname_{i \in 1..l})
        (\nu \sessname_{j \in 1..m})
        (\nu \aname_{k \in 1..n})(
        \ap{\apname_i}{\apstate_i}_{i \in 1..l} \parallel 
        (\qproc{s_j}{\epsilon})_{j \in 1..m} \parallel
        \zapactor[a_k]{\idle{\valc_k}}{\epsilon}{\istate_k}{\monstate_k}_{k \in 1..{n' - 1}} \parallel
        \setseq{\zap{\nma}}
        )
    \]

    We must now account for the remaining zapper threads.
    If there exists a zapper thread $\zap{\aname}$ where $\aname$ is contained
    within some monitoring environment $\monstate$ then we can reduce with
    \textsc{E-InvokeM}. If $\aname$ does not occur free in any initialisation
    callback or monitoring callback then we can eliminate it using the garbage
    collection congruence $(\nu \aname)(\zap{\aname}) \parallel \config{C}
    \equiv \config{C}$.

    Next, we eliminate all zapper threads for initialisation tokens using
    \textsc{E-CancelAP}.

    Finally, we can eliminate all failed sessions $(\nu
    s)(\zap{\roleidxx{s}{\prole_1}} \parallel \cdots \parallel
    \zap{\roleidxx{s}{\prole_n}} \parallel \qproc{s}{\epsilon})$, and we are
    left with a configuration of the form:

\[
    (\nu \tilde{\inittok})(\nu \apname_{i \in 1..m})(\nu \aname_{j \in 1..n})
    (
    \ap{\apname_1}{\apstate_1}_{i \in 1..m} \parallel
    \zapactor[a_j]{\idle{\valc_k}}{\epsilon}{\istate_j}{\monstate_j}_{j \in 1..{n'-1}} \parallel
    (\zap{a_j})_{j \in {n'..n}}
    )
\]
as required.
\end{proof}

\subsubsection{Global Progress}

A modified version of global progress holds: for every active session, in a
finite number of reductions, either the session can make a communication action,
or all endpoints become cancelled and can be garbage collected.

\zapgprog*
\begin{proof}
    Follows the same structure as the proof of
    Corollary~\ref{cor:global-progress}, the main difference being that instead
    of \textsc{E-React} firing, it may be the case that \textsc{E-CancelH} fires
    to propagate a failure. In this case, if all session endpoints for an active
    session $s$ are cancelled, then it would be possible to use the garbage
    collection congruence to eliminate the failed session.
\end{proof}
 \clearpage
\section{Formal Model of Session Switching Extension}\label{ap:switching}

\begin{figure}[t]
    {\footnotesize
    \header{Modified syntax}
    \[
        \begin{array}{lrcl}
        \text{Session names} & \sname{s}, \sname{t} \\
        \text{Values} & \vala, \valb & ::= & \cdots \midspace (\vala, \valb) \\
        \text{Computations} & \tma, \tmb & ::= &
        \cdots \midspace \letin{(x, y)}{\tma}{\tmb} \\
                            & & \mid &
            \suspendsend{s}{\vala}{\valb} \midspace
            \suspendrecv{\vala}{\valb}
            \midspace \become{s}{\vala} \\
        \text{Send-suspended sessions} & \suspentry & ::= & \metapair{\roleidx{s}{p}}{\vala} \\
        \text{Handler state} & \hstate & ::= & \epsilon \midspace
        \hstate, \roleidx{s}{p} \mapsto \vala \midspace \hstate, \sname{s}
        \mapsto \seq{\suspentry} \\
        \text{Switch request queue} & \swstate & ::= & \epsilon \midspace \swstate \cdot \metapair{\sname{s}}{\vala} \\
        \text{Configurations} & \config{C}, \config{D} & ::= & \cdots \midspace
        \swactor{\threadt}{\hstate}{\istate}{\vala}{\swstate} \\
        \end{array}
    \]
    \headertwo
        {Modified typing rules}
        {\framebox{$\vseq{\tyenv}{\vala}{\tya}$}~\framebox{$\mtseqst{\tyenv}{\tyc}{\sta}{\tma}{\tya}{\stb}$}}

    \begin{mathpar}
    \inferrule
    {
        \vseq{\tyenv}{\vala}{\tya} \\
        \vseq{\tyenv}{\valb}{\tyb}
    }
    { \vseq{\tyenv}{(\vala, \valb)}{(\tya \times \tyb)} }

    \inferrule
    {
        \vseq{\tyenv}{\vala}{(\tya_1 \times \tya_2)}
        \\
        \mtseqst{\tyenv}{\tyc}{\sta}{\tma}{\tyb}{\stb}
    }
    { \mtseqst{\tyenv}{\tyc}{\sta}{\letin{(x, y)}{\vala}{\tma}}{\tyb}{\stb} }

    \inferrule
    [T-Suspend$_{?}$]
    { \vseq{\tyenv}{\vala}{\handlerty{\stin}{\tyc}} \\
      \vseq{\tyenv}{\valb}{\tyc}
    }
    { \mtseqst{\tyenv}{\tyc}{\stin}{\suspendrecv{\vala}{\valb}}{\tya}{\stb} }

    \inferrule
    [T-Suspend$_{!}$]
    { \siglookup{s}{\stout}{\tya} \\\\
      \vseq{\tyenv}{\vala}{\sttyfun{(\tya \times \tyc)}{\tyc}{\stout}{\localend}{\tyc}} }
    { \mtseqst{\tyenv}{\tyc}{\stout}{\suspendsend{s}{\vala}}{\tyb}{\stb} }

    \inferrule
    [T-Become]
    {
        \siglookup{s}{\stb}{\tya} \\
        \vseq{\tyenv}{\vala}{\tya}
    }
    { \mtseqst{\tyenv}{\tyc}{\sta}{\become{s}{\vala}}{\one}{\sta} }
    \end{mathpar}

    \headertwo
        {Modified configuration typing rules}
        {\framebox{$\cseq{\tyenv}{\rtenv}{\config{C}}$}~
            \framebox{$\hstateseq{\tyenv}{\rtenv}{\tyc}{\hstate}$}~
            \framebox{$\swstateseq{\tyenv}{\swstate}$}
        }

    \begin{mathpar}
       \inferrule
       [T-Actor]
       {
           \threadseq{\tyenv}{\rtenv_1}{\valc}{\threadt} \\
           \hstateseq{\tyenv}{\rtenv_2}{\valc}{\hstate} \\\\
           \istateseq{\tyenv}{\rtenv_3}{\valc}{\istate} \\
           \swstateseq{\tyenv}{\swstate}
       }
       { \cseq
           { \tyenv }
           { \rtenv_1, \rtenv_2, \rtenv_3, \aname }
           { \swactor
               {\threadt}
               {\hstate}
               {\istate}
               {\swstate}
           }
       }

       \inferrule
       [TH-SendHandler]
       {
           \hstateseq{\tyenv}{\rtenv}{\tyc}{\hstate} \\\\
           \siglookup{s}{\stout}{\tya} \\
           (\vseq{\tyenv}{\vala_i}{\sttyfun{(\tya \times \tyc)}{\tyc}{\stout}{\localend}{\tyc}})_i
       }
       { \hstateseq
           {\tyenv}
           {\rtenv, (\roleidxx{s_i}{\prole_i}: \stout)_i}
           {\tyc}
           {\hstate, \sname{s} \mapsto \metapair{\roleidxx{s_i}{\prole_i}}{\vala_i}_i }
       }

       \inferrule
       [TR-Empty]
       { }
       { \swstateseq{\tyenv}{\epsilon} }

       \inferrule
       [TR-Request]
       {
           \swstateseq{\tyenv}{\swstate} \\
           \siglookup{s}{\stout}{\tya} \\
           \vseq{\tyenv}{\vala}{\tya}
       }
       { \swstateseq{\tyenv}{\swstate \cdot \metapair{\sname{s}}{\vala}} }
    \end{mathpar}

    \headersig{Modified reduction rules}{$\config{C} \ceval \config{D}$}

    \[
        \begin{array}{lrcl}
            \textsc{E-Suspend}_{!}\textsc{-1}
            &
            \swactor
                 {\sessthread{s}{\prole}{\ctxe[\suspendsend{s}{\vala}{\valb}]}}
                 {\hstate}
                 {\istate}
                 {\swstate}
            & \cevalann{\tau} &
             \swactor
                 {\idle{\valb}}
                 {\hstate[\sname{s} \mapsto \metapair{\roleidx{s}{\prole}}{\vala} ]}
                 {\istate}
                 {\swstate}
                 \quad (\sname{s} \not\in \dom{\sigma})
            \\
            \textsc{E-Suspend}_{!}\textsc{-2}
            &
            \swactor
                 {\sessthread{s}{\prole}{\ctxe[\suspendsend{s}{\vala}{\valb}]}}
                 {\hstate[\sname{s} \mapsto \seq{\suspentry}]}
                 {\istate}
                 {\swstate}
            & \cevalann{\tau} &
             \swactor
                 {\idle{\valb}}
                 {\hstate[\sname{s} \mapsto \seq{\suspentry} \cdot \metapair{\roleidx{s}{\prole}}{\vala} ]}
                 {\istate}
                 {\swstate}
            \\
\textsc{E-Become}
            &
            \swactor{\threadctx[\become{s}{\vala}]}{\hstate}{\istate}{\swstate}
            & \cevalann{\tau} &
            \swactor{\threadctx[\effreturn{()}]}{\hstate}{\istate}{\swstate
            \cdot \metapair{\sname{s}}{\vala}}
            \\
\textsc{E-Activate} &
        \swactor
         {\idle{\valc}}
         {\hstate[\sname{s} \mapsto \sendentry{\roleidx{s}{\prole}}{\vala} \cdot \seq{\suspentry}]}
         {\istate}
         {\sendentry{\sname{s}}{\valb} \cdot \swstate}
            & \cevalann{\tau} &
       \swactor
         {\sessthread{s}{\prole}{\vala (\valb, \valc)}}
         {\hstate[\sname{s} \mapsto \seq{\suspentry}]}
         {\istate}
         {\swstate}
        \end{array}
    \]
}
    \caption{\langnamesw: Modified syntax, typing, and reduction rules}
    \label{fig:extensions-switching}
\end{figure}

In Section~\ref{sec:implementation} we described the implementation of \langname
to support proactive switching between sessions. In this appendix we introduce a
formal model of a similar feature that switches between sessions by queueing
requests to invoke a send-suspended session, and activates send-suspended
sessions when the event loop reverts to being idle.

Suppose that we want to adapt our Shop example to maintain a long-running
session with a supplier and request a delivery whenever an item runs out of stock.
The key difference to our original example is that we need to \emph{switch} to the
Restock session \emph{as a consequence} of receiving a \mkwd{buy} message in a 
customer session.

We can
describe a \mkwd{Restock} session with the following simple local types:

 \begin{minipage}{0.45\textwidth}
 {\footnotesize
 \[
     \bl
     \mkwd{ShopRestock} \defeq \\
     \quad \rectyone{\mkwd{loop}} \\
     \qquad \localselectsingle{Supplier}{\mkwd{order}}{([\mkwd{ItemID}] \times
     \mkwd{Quantity})} \then \\
     \qquad \localoffersingle{Supplier}{\mkwd{ordered}}{\mkwd{Quantity}} \then
     \mkwd{loop}
     \el
 \]
 }
 \end{minipage}
 \begin{minipage}{0.45\textwidth}
 {\footnotesize
 \[
     \bl
     \mkwd{SupplierRestock} \defeq \\
     \quad \rectyone{\mkwd{loop}} \\
     \qquad \localoffersingle{Shop}{\mkwd{order}}{([\mkwd{ItemID}] \times
     \mkwd{Quantity})} \then \\
     \qquad \localselectsingle{Shop}{\mkwd{ordered}}{\mkwd{Quantity}} \then
     \mkwd{loop}
     \el
 \]
 }
 \end{minipage}

Whereas before we only needed to suspend an actor in a \emph{receiving} state,
this workflow requires us to also suspend an actor in a \emph{sending} state,
and switch into the session at a later stage. We call this extension
$\langnamesw$.
Below, we can see the extension of the shop example with the ability to switch
into the restocking session; the new constructs are shaded.

{\small
\[
    \bl
\mkwd{ShopRestock} \defeq \\
    \quad \rectyone{\mkwd{loop}} \\
    \qquad \localselectsingle{Supplier}{\mkwd{order}}{([\mkwd{ItemID}] \times
    \mkwd{Quantity})} \then \\
    \qquad \localoffersingle{Supplier}{\mkwd{ordered}}{\mkwd{Quantity}} \then
    \mkwd{loop} \\ \\
\mkwd{custReqHandler} \defeq \\
    \:
        \bl
        \handlerone{Customer}{\var{st}} \\
        \quad \msg{\mkwd{getItemInfo}}{\var{itemID}} \mapsto [\ldots] \\
        \quad \msg{\mkwd{checkout}}{(\var{itemIDs}, \var{details})} \mapsto \\
        \qquad \letintwo{\var{items}}{\stget} \\
        \qquad \ithen{\mkwd{inStock}(\var{itemIDs}, \var{items})} \: [\ldots] \\
        \qquad \calcwd{else} \\
        \qqquad \send{Customer}{\mkwd{outOfStock}}{}; \\
        \qqquad \shade{\become{Restock}{\var{itemIDs}}}; \\
        \qqquad \suspendrecv{\mkwd{custReqHandler}}{\var{st}} \\
        \}
        \el
        \\ \\
\mkwd{shop} \defeq \lambda (\var{custAP}, \var{restockAP}) . \\
        \:
        \bl
\registertwo{\var{custAP}}{\role{Shop}}  \\
            \quad {
                (\lambda \var{st}.
                \mkwd{shop} \: (\var{custAP}, \role{Shop}) \;
                \fun{\var{st}}{\suspendrecv{\mkwd{itemReqHandler}}{\var{st}}})}; \\
\register
                {\var{restockAP}}
                {\role{Shop}}
                {(\lambda \var{st} . \;
                    {\begin{array}[t]{l}
                        \shade{\suspendsend{Restock}{\mkwd{restockHandler}}{\var{st}}});
                        \\
                    \mkwd{initialStock}
                    \el}
                }\\
        \el \\ \\
        \mkwd{restockHandler} \defeq \lambda (\var{itemIDs}, \var{st}) \then \\
    \:
    \bl
    \send{Supplier}{\mkwd{order}}{(\var{itemIDs}, 10)}; \\
    \suspendrecvzero \: ( \\
    \quad \handlerone{Supplier}{\var{st}} \\
    \qquad \msg{\mkwd{ordered}}{\var{quantity}} \mapsto \\
    \qqquad \mkwd{increaseStock}(\var{itemIDs}, \var{quantity}); \\
    \qqquad \shade{\suspendsend{Restock}{\mkwd{restockHandler}}{\var{st}}}
    \})
    \el
\el
\]
}

The program is implicitly parameterised by a mapping from static names like
$\sname{Restock}$ to pairs of session types and payload types 
(in our
scenario, $\sname{Restock}$ maps to $(\mkwd{ShopRestock}, [\mkwd{ItemID}])$ to
show that an actor can suspend when its session type is \mkwd{ShopRestock}, and
must provide a list of \mkwd{ItemID}s when switching back into the session).
We split the $\calcwd{suspend}$ construct into
$\suspendrecv{\vala}$ (to suspend awaiting an incoming message, as previously),
and $\suspendsend{s}{\vala}$ (to suspend session with name $\sname{s}$ given a
function $\vala$, until switched into), and introduce the $\become{s}{\vala}$
construct to switch into a suspended session. 
Specifically, $\become{s}{\vala}$ queues $\sname{s}$ to run when the actor is
next idle.
We modify the $\mkwd{shop}$ definition to also
register with the $\var{restockAP}$ access point, suspending the session (in a
state that is ready to send) with the $\mkwd{restockHandler}$. The
\mkwd{restockHandler} takes an item ID, sends an \mkwd{order} message to the
supplier, and suspends again.

\paragraph*{Metatheory.}
$\langnamesw$ satisfies preservation.  Since (by design) $\calcwd{become}$
operations are dynamic and not encoded in the protocol (for example, we might
wish to queue two invocations of a send-suspended session to be executed in
turn), there is no type-level mechanism of guaranteeing that a send-suspended
session is invoked, so $\langnamesw$ instead enjoys progress up-to invocation of
send-suspended sessions (see Appendix~\ref{ap:extensions-proofs}).

Our extension to allow session switching is shown in
Figure~\ref{fig:extensions-switching}. We introduce a set of distinguished
\emph{session identifiers} $\sname{s}$; each session identifier is associated
with a local type and a payload in an environment $\Sigma$, i.e., for each
$\sname{s}$ we have $\siglookup{s}{\stout}{\tya}$ for some $\stout, \tya$.
We then split the $\calcwd{suspend}$ construct into two:
$\suspendrecv{\vala}{\valb}$ (which, as before, installs a message handler
$\vala$ and suspends
an actor with updated state $\valb$) and
$\suspendsend{s}{\vala}{\valb}$, which suspends a session in a
\emph{send} state, installing a function $\vala$ taking a payload of the given type.
Finally we introduce a $\become{s}{\vala}$ construct that queues a request for
the event loop to invoke $\sname{s}$ next time the actor is idle and a
send-suspended session is available.

\subsection{Metatheory}

\subsubsection{Preservation}
As would be expected, \langnamesw satisfies preservation.

\begin{restatable}[Preservation]{theorem}{preservationsw}\label{thm:preservation-sw}
    \hspace{-0.5em}
    Preservation (as defined in Theorem~\ref{thm:preservation}) continues to
    hold in {\langnamesw}.
\end{restatable}
\begin{proof}
    Preservation of typing under structural congruence follows
    straightforwardly.

    For preservation of typing under reduction, we proceed by induction on the
    derivation of $\config{C} \ceval \config{D}$.

    \begin{proofcase}{\textsc{E-Suspend}$_{!}$\textsc{-1}}
        Similar to {\textsc{E-Suspend}$_{!}$\textsc{-2}}.
    \end{proofcase}

    \begin{proofcase}{\textsc{E-Suspend}$_{!}$\textsc{-2}}
        \[
            \swactor
                 {\sessthread{s}{\prole}{\ctxe[\suspendsend{s}{\vala}{\valb}]}}
                 {\hstate[\sname{s} \mapsto \seq{\suspentry}]}
                 {\istate}
                 {\swstate}
\cevalann{\tau}
\swactor
                 {\idle{\valb}}
                 {\hstate[\sname{s} \mapsto \seq{\suspentry} \cdot \metapair{\roleidx{s}{\prole}}{\vala} ]}
                 {\istate}
                 {\swstate}
        \]

        Assumption:

        {\footnotesize
        \begin{mathpar}
            \inferrule*
            {
                \inferrule
                {
                    \mtseqst
                        {\tyenv}
                        {\tyc}
                        {\sta}
                        {\ctxe[\suspendsend{s}{\vala}{\valb}]}
                        {\tyc}
                        {\localend}
                }
                { \threadseq
                    {\tyenv}
                    { \roleidx{s}{p} : \sta}
                    {\tyc}
                    {\sessthread{s}{p}{\ctxe[\suspendsend{s}{\vala}{\valb}]}}
                }
                \\
                \inferrule
                {
                    \hstateseq{\tyenv}{\rtenv_1}{\tyc}{\hstate} \\
                   \siglookup{s}{\stout}{\tya} \\\\
                   (\vseq{\tyenv}{\valb_i}{\sttyfun{(\tya \times \tyc)}{\one}{\stout}{\localend}{\tyc}})_i
                }
                { \hstateseq
                    {\tyenv}
                    {\rtenv_1,
                        (\roleidxx{s_i}{\qrole_i}: \stout)_i
                    }
                    {\tyc}
                    {\hstate[\sname{s} \mapsto \metapair{\roleidxx{s_i}{\qrole_i}}{\valb_i}_i]}
                }
                \\
                {\bl
                    \istateseq{\tyenv}{\rtenv_2}{\tyc}{\istate}\\
                    \swstateseq{\tyenv}{\swstate}
                \el}
            }
            {
                \cseq
                { \tyenv }
                { \rtenv_1, \rtenv_2, \roleidx{s}{p} : \sta,
                (\roleidxx{s_i}{\qrole_i}: \stout)_i, \aname}
                {
                    \swactor
                         {\ctxe[\suspendsend{s}{\vala}{\valb}]}
                         {\hstate[\sname{s} \mapsto \metapair{\roleidxx{s_i}{\qrole_i}}{\valb_i}_i]}
                         {\istate}
                         {\swstate}
                }
            }
        \end{mathpar}
        }

        Consider the subderivation 
        $
                    \mtseq
                        {\tyenv}
                        {\sta}
                        {\ctxe[\suspendsend{s}{\vala}{\valb}]}
                        {\one}
                        {\localend}
        $.
By Lemma~\ref{lem:subterm-typability} there exists a subderivation:

        \begin{mathpar}
            \inferrule
            {
               \siglookup{s}{\stout}{\tya} \\
               \vseq{\tyenv}{\vala}{\sttyfun{(\tya \times \tyc)}{\tyc}{\stout}{\localend}{\tyc}}
               \\
               \vseq{\tyenv}{\valb}{\tyc}
            }
            { \mtseq{\tyenv}{\stout}{\suspendsend{s}{\vala}{\valb}}{\tyb}{\localend} }
        \end{mathpar}

        Therefore we have that $\sta = \stout$.

        Recomposing:

        {\footnotesize
        \begin{mathpar}
            \inferrule*
            {
                \inferrule
                { \vseq{\tyenv}{\valb}{\tyc} }
                { \threadseq
                    {\tyenv}
                    { \cdot }
                    {\tyc}
                    { \idle{\valb} }
                }
                \\
                \inferrule
                {
                    \hstateseq{\tyenv}{\rtenv_1}{\tyc}{\hstate} \\
                   \siglookup{s}{\stout}{\tya} \\\\
                   (\vseq{\tyenv}{\valb_i}{\sttyfun{(\tya \times \tyc)}{\tyc}{\stout}{\localend}{\tyc}})_i \\
                   \vseq{\tyenv}{\vala}{\sttyfun{(\tya \times \tyc)}{\tyc}{\stout}{\localend}{\tyc}}
                }
                { \hstateseq
                    {\tyenv}
                    {\rtenv_1,
                        (\roleidxx{s_i}{\qrole_i}: \stout)_i, \roleidx{s}{\prole} : \stout
                    }
                    {\tyc}
                    {\hstate[\sname{s} \mapsto \metapair{\roleidxx{s_i}{\qrole_i}}{\valb_i}_i \cdot \metapair{\roleidx{s}{p}}{\vala}]}
                }
                \\
                {\bl
                    \istateseq{\tyenv}{\rtenv_2}{\tyc}{\istate}\\
                \swstateseq{\tyenv}{\swstate}
                \el}
            }
            {
                \cseq
                { \tyenv }
                { \rtenv_1, \rtenv_2, \roleidx{s}{p} : \sta,
                (\roleidxx{s_i}{\qrole_i}: \stout)_i, \aname}
                {
                    \swactor
                         {\idle{\valb}}
                         {\hstate[\sname{s} \mapsto \metapair{\roleidxx{s_i}{\qrole_i}}{\valb_i}_i \cdot \metapair{\roleidx{s}{p}}{\vala}]}
                         {\istate}
                         {\swstate}
                }
            }
        \end{mathpar}
        }

        as required.
    \end{proofcase}

    \begin{proofcase}{\textsc{E-Become}}

        \[
            \swactor{\threadctx[\become{s}{\vala}]}{\hstate}{\istate}{\swstate}
            \cevalann{\tau}
            \swactor{\threadctx[\effreturn{()}]}{\hstate}{\istate}{\swstate \cdot \metapair{\sname{s}}{\vala}}
        \]

        Assumption (considering the case that $\threadctx = \ctxe[-]$ for some
        $\ctxe$; the case in the context of a session is identical):

        \begin{mathpar}
          \inferrule*
          {
              \inferrule*
              { \mtseqst{\tyenv}{\sta}{\tyc}{\ctxe[\become{s}{\vala}]}{\tyc}{\localend} }
              { \threadseq{\tyenv}{\cdot}{\tyc}{\ctxe[\become{s}{\vala}]} } \\
              {\bl
                  \hstateseq{\tyenv}{\rtenv_1}{\tyc}{\hstate} \\
                  \istateseq{\tyenv}{\rtenv_2}{\tyc}{\istate} \\
                  \swstateseq{\tyenv}{\swstate}
              \el
              }
          }
          { \cseq
              { \tyenv }
              { \rtenv_1, \rtenv_2, \aname }
              { \swactor
                  {\threadt}
                  {\hstate}
                  {\istate}
                  {\swstate}
              }
          }
        \end{mathpar}

        By Lemma~\ref{lem:subterm-typability} we have:

        \begin{mathpar}
            \inferrule
            {
                \siglookup{s}{\stb}{\tya} \\
                \vseq{\tyenv}{\vala}{\tya}
            }
            {
                \mtseqst{\tyenv}{\sta}{\tyc}{\become{s}{\vala}}{\one}{\sta}
            }
        \end{mathpar}

        By Lemma~\ref{lem:replacement} we can show that
        $\mtseqst{\tyenv}{\sta}{\tyc}{\ctxe[\effreturn{()}]}{\tyc}{\localend}$.

        Recomposing:

        {\footnotesize
        \begin{mathpar}
          \inferrule*
          {
              \inferrule*
              { \mtseqst{\tyenv}{\tyc}{\sta}{\ctxe[\effreturn{()}]}{\one}{\localend} }
              { \threadseq{\tyenv}{\cdot}{\tyc}{\ctxe[\effreturn{()}]} } \\
              {\bl
              \hstateseq{\tyenv}{\rtenv_1}{\tyc}{\hstate} \\
              \istateseq{\tyenv}{\rtenv_2}{\tyc}{\istate}
              \el
              }\\
\inferrule*
              {
                   \swstateseq{\tyenv}{\swstate} \\
                   \siglookup{s}{\stout}{\tya} \\
                   \vseq{\tyenv}{\vala}{\tya}
              }
              { \swstateseq{\tyenv}{\swstate \cdot \metapair{\sname{s}}{\vala}} }
          }
          { \cseq
              { \tyenv }
              { \rtenv_1, \rtenv_2, \aname }
              { \swactor
                  {\ctxe[\effreturn{()}]}
                  {\hstate}
                  {\istate}
                  {\swstate \cdot \metapair{\sname{s}}{\vala} }
              }
          }
        \end{mathpar}        
        }
        as required.
    \end{proofcase}

    \begin{proofcase}{\textsc{E-Activate}}
        \[
         \swactor
         {\idle{\valc}}
         {\hstate[\sname{s} \mapsto \sendentry{\roleidx{s}{\prole}}{\vala} \cdot \seq{\suspentry}]}
         {\istate}
         {\sendentry{\sname{s}}{\valb} \cdot \swstate}
            \cevalann{\tau}
       \swactor
         {\sessthread{s}{\prole}{\vala \app (\valb, \valc)}}
         {\hstate[\sname{s} \mapsto \seq{\suspentry}]}
         {\istate}
         {\swstate}
        \]

        Let $\derivd$ be the subderivation:
        \begin{mathpar}
            \inferrule*
            {
                \hstateseq{\tyenv}{\rtenv_1}{\tyc}{\hstate} \\\\
                \siglookup{s}{\stout}{\tya} \\
                \vseq{\tyenv}{\vala}{\sttyfun{(\tya \times \tyc)}{\tyc}{\stout}{\localend}{\tyc}} \\
                (\vseq{\tyenv}{\vala_i}{\sttyfun{(\tya \times \tyc)}{\tyc}{\stout}{\localend}{\tyc}})_i
            }
            { \hstateseq
                {\tyenv}
                {\rtenv_1, \roleidx{s}{p} : \stout, (\roleidxx{s_i}{\prole_i}: \stout)_i}
                {\tyc}
                {\hstate, \sname{s} \mapsto
                    \metapair
                        {\roleidx{s}{\prole}}
                        {\vala}
                    \cdot
                    \metapair{\roleidxx{s_i}{\prole_i}}{\vala_i}_i }
            }
        \end{mathpar}

        Assumption:
        {\footnotesize
        \begin{mathpar}
            \inferrule*
            {
                \inferrule*
                { \vseq{\tyenv}{\valc}{\tyc} }
                { \threadseq{\tyenv}{\cdot}{\tyc}{\idle{\valc}} }
                \\
                \derivd
                \\
                \istateseq{\tyenv}{\rtenv_2}{\tyc}{\istate}
                \\
                \inferrule*
                {
                    \swstateseq{\tyenv}{\swstate} \\
                    \siglookup{s}{\stout}{\tya} \\
                    \vseq{\tyenv}{\valb}{\tya}
                }
                { \swstateseq{\tyenv}{
                    \metapair{\sname{s}}{\valb}
                    \cdot
                    \swstate} }
            }
            {
              \cseq
              {\tyenv}
              {\rtenv_1, \rtenv_2, \roleidx{s}{p} : \stout,
              (\roleidxx{s_i}{\prole_i}: \stout)_i, a}
              {
              \swactor
              {\idle{\valc}}
                  {\hstate[\sname{s} \mapsto \sendentry{\roleidx{s}{\prole}}{\vala} \cdot
                      \metapair{\roleidxx{s_i}{\prole_i}}{\vala_i}_i] }
                  {\istate}
                  {\sendentry{\sname{s}}{\valb} \cdot \swstate}
              }
            }
        \end{mathpar}
        }

        Let $\derivd'$ be the subderivation:

        {\small
        \begin{mathpar}
            \inferrule*
            {
                \inferrule*
                {
                    \vseq{\tyenv}{\vala}{\sttyfun{\tya}{\tyc}{\stout}{\localend}{\tyc}} \\
                    \inferrule*
                    { \vseq{\tyenv}{\valb}{\tya} \\ \vseq{\tyenv}{\valc}{\tyc} }
                    { \vseq{\tyenv}{(\valb, \valc)}{(\tya \times \tyc)} }
                }
                { \mtseqst{\tyenv}{\tyc}{\stout}{\vala \app (\valb, \valc)}{\tyc}{\localend} }
            }
            { \threadseq{\tyenv}{\roleidx{s}{p} :
            \stout}{\tyc}{\sessthread{s}{p}{\vala \app (\valb, \valc)}} }
        \end{mathpar}
        }

        Recomposing:

        {\footnotesize
        \begin{mathpar}
            \inferrule*
            {
                \derivd
                \\
                \inferrule*
                {
                    \hstateseq{\tyenv}{\rtenv_1}{\tyc}{\hstate} \\\\
                    \siglookup{s}{\stout}{\tya} \\
                    (\vseq{\tyenv}{\vala_i}{\sttyfun{(\tya \times \tyc)}{\tyc}{\stout}{\localend}{\tyc}})_i
                }
                { \hstateseq
                    {\tyenv}
                    {\rtenv_1, (\roleidxx{s_i}{\prole_i}: \stout)_i}
                    {\tyc}
                    {\hstate, \sname{s} \mapsto
                        \metapair{\roleidxx{s_i}{\prole_i}}{\vala_i}_i }
                }
                \\
                \istateseq{\tyenv}{\rtenv_2}{\tyc}{\istate}
                \\
                \swstateseq{\tyenv}{\swstate}
            }
            {
              \cseq
              {\tyenv}
              {\rtenv_1, \rtenv_2, \roleidx{s}{p} : \stout,
              (\roleidxx{s_i}{\prole_i}: \stout)_i, a}
              {
              \swactor
                  {\sessthread{s}{p}{\vala \app (\valb, \valc)}}
                  {\hstate[\sname{s} \mapsto
                  \metapair{\roleidxx{s_i}{\prole_i}}{\vala_i}_i] }
                  {\istate}
                  {\swstate}
              }
            }
        \end{mathpar}
        }
         as required.
    \end{proofcase}
\end{proof}

\subsubsection{Progress}
Since (by design) $\calcwd{become}$ operations are dynamic and not
encoded in the protocol (for example, we might wish to queue two invocations of
a send-suspended session to be executed in turn), there is no type-level mechanism
of guaranteeing that a send-suspended session is ever invoked. Although all
threads can reduce as before, \langnamesw satisfies a weaker version of progress
where non-reducing configurations can contain send-suspended sessions.

\begin{restatable}
    [Progress (\langnamesw)]
    {theorem}
    {progresssw}\label{thm:progress-sw}
If $\cseq[\dfprop]{\cdot}{\cdot}{\config{C}}$, then either there
exists some $\config{D}$ such that $\config{C} \ceval \config{D}$, or
$\config{C}$ is structurally congruent to the following canonical form:
\[
        (\nu \tilde{\inittok})
        (\nu \apname_{i \in 1..l})
        (\nu \sessname_{j \in 1..m})
        (\nu \aname_{k \in 1..n})(
        \ap{\apname_i}{\apstate_i}_{i \in 1..l} \parallel
        (\qproc{s_j}{\qcontents_j})_{j \in 1..m} \parallel
        \swactor[a_k]{\idle{\valc_k}}{\hstate_k}{\istate_k}{\swstate_k}_{k \in 1..n}
        )
    \]
    where for each session $s_j$ there exists some mapping
    $\roleidx{s_j}{\prole} \mapsto (\sname{s}, \vala)$ (for some role $\prole$,
    static session name $\sname{s}$, and callback $\vala$)
    contained in some $\hstate_k$ where $\swstate_k$ does not contain any
    requests for $\sname{s}$.
\end{restatable}
\begin{proof}
    The proof follows that of Theorem~\ref{thm:progress}. Thread progress
    (Lemma~\ref{lem:thread-progress}) holds as before, since we can always
    evaluate $\calcwd{become}$ by \textsc{E-Become}, and we can always evaluate
    $\calcwd{suspend}_{!}$ by E-Suspend-!$_1$ or E-Suspend-!$_2$.

    Following the same reasoning as Theorem~\ref{thm:progress} we can write
    $\config{C}$ in canonical form, where all threads are idle:

    \[
        (\nu \tilde{\inittok})
        (\nu \apname_{i \in 1..l})
        (\nu \sessname_{j \in 1..m})
        (\nu \aname_{k \in 1..n})
        (
        \ap{\apname_i}{\apstate_i}_{i \in 1..l} \parallel
        (\qproc{s_j}{\qcontents_j})_{j \in 1..m} \parallel 
        \swactor[\aname_k]{\idle{\vala_k}}{\hstate_k}{\istate_k}{\swstate_k}_{k \in 1..n}
        )
    \]

    However, there are now \emph{three} places each role endpoint
    $\roleidx{s}{p}$ can be used: either by \textsc{TT-Sess} to run a term in
    the context of a session or by \textsc{TH-Handler} to record a
    receive-suspended session type as before, but now also by
    \textsc{TH-SendHandler} to record a send-suspended session type. As before,
    the former is impossible as all threads are idle, so now we must consider
    the cases for \textsc{TH-Handler}.

    Following the same reasoning as Theorem~\ref{thm:progress}, we can reduce
    any handlers that have waiting messages. Thus we are finally left with the
    scenario where the session type LTS can reduce, but not the configuration:
    this can only happen when the sending reduction is send-suspended, as
    required.
\end{proof}

\end{document}